\documentclass[acmsmall,screen]{acmart}
\pdfoutput=1 

\usepackage{mathtools}
\usepackage{stmaryrd}
\mathtoolsset{showonlyrefs=true}
\usepackage{tikz,pgfplots}
\pgfplotsset{compat=1.18}
\usetikzlibrary{cd}
\usepackage{mathpartir}
\usepackage{listings}
\usepackage{enumitem}
\setitemize{leftmargin=15pt}
\usepackage{multirow}

\usepackage{ifthen}
\newboolean{longversion}
\setboolean{longversion}{true}
\newcommand{\referappendix}[2]{\ifthenelse{\boolean{longversion}}{Appendix~\ref{#1}}{\cite[Appendix {#2}]{arxiv}}}
\newcommand{\referappendixnocite}[2]{\ifthenelse{\boolean{longversion}}{Appendix~\ref{#1}}{Appendix~{#2}}}

\usepackage{macros.basic}

\newcommand{\answertype}{\mathbf{Prop}}
\newcommand{\answerobj}{\Omega}
\newcommand{\underlying}[1]{| #1 |}
\newcommand{\letrec}[4]{\mathbf{let}\ \mathbf{rec}\ #1\ #2 = #3\ \mathbf{in}\ #4}
\newcommand{\letfix}[4]{\mathbf{let}\ \mathbf{fix}\ #1\ #2 = #3\ \mathbf{in}\ #4}
\newcommand{\fixpoint}[2]{\mathbf{fix}\ #1. #2}
\newcommand{\ifexpr}[3]{\mathbf{if}\ #1\ \mathbf{then}\ #2 \ \mathbf{else}\ #3}

\newcommand{\caseexpr}[5]{\delta(#1, #2. #3, #4. #5)}
\newcommand{\probbranch}[3]{#1 \oplus_{#2} #3}
\newcommand{\BaseTypes}{\mathbf{Base}}
\newcommand{\AtomicPreds}{\mathbf{AP}}
\newcommand{\BasicOps}{\mathbf{BO}}
\newcommand{\lambdaHFL}{\lambda_{\mathrm{HFL}}}
\newcommand{\emalgsymbol}{\zeta}
\newcommand{\emptyenv}{\emptyset}

\AtBeginDocument{%
  }

\ifthenelse{\boolean{longversion}}{
	\setcopyright{none}
}{
	\setcopyright{cc}
	\copyrightyear{2024}
}
\acmYear{2024}
\acmDOI{10.1145/3674662}

\acmConference[ICFP '24]{Make sure to enter the correct
  conference title from your rights confirmation email}{Sep 2--7,2024}{Milan, Italy}
\acmISBN{978-1-4503-XXXX-X/18/06}

\AtEndPreamble{%
\theoremstyle{acmdefinition}
\newtheorem{remark}[theorem]{Remark}}

\begin{document}

\title{Automated Verification of Higher-Order Probabilistic Programs via a Dependent Refinement Type System}

\author{Satoshi Kura}
\orcid{0000-0002-3954-8255}
\affiliation{%
	\institution{Waseda University}
	\city{Tokyo}
	\country{Japan}
}
\email{satoshi.kura@aoni.waseda.jp}

\author{Hiroshi Unno}
\orcid{0000-0002-4225-8195}
\affiliation{%
	\institution{Tohoku University}
	\city{Sendai}
	\country{Japan}}
\email{hiroshi.unno@acm.org}


\begin{abstract}
	Verification of higher-order probabilistic programs is a challenging problem.
	We present a verification method that supports several quantitative properties of higher-order probabilistic programs.
	Usually, extending verification methods to handle the quantitative aspects of probabilistic programs often entails extensive modifications to existing tools, reducing compatibility with advanced techniques developed for qualitative verification.
	In contrast, our approach necessitates only small amounts of modification, facilitating the reuse of existing techniques and implementations.
	On the theoretical side, we propose a dependent refinement type system for a generalised higher-order fixed point logic (HFL).
	Combined with continuation-passing style encodings of properties into HFL, our dependent refinement type system enables reasoning about several quantitative properties, including weakest pre-expectations, expected costs, moments of cost, and conditional weakest pre-expectations for higher-order probabilistic programs with continuous distributions and conditioning.
	The soundness of our approach is proved in a general setting using a framework of categorical semantics so that we don't have to repeat similar proofs for each individual problem.
	On the empirical side, we implement a type checker for our dependent refinement type system that reduces the problem of type checking to constraint solving.
	We introduce \emph{admissible predicate variables} and \emph{integrable predicate variables} to constrained Horn clauses (CHC) so that we can soundly reason about the least fixed points and samplings from probability distributions.
	Our implementation demonstrates that existing CHC solvers developed for non-probabilistic programs can be extended to a solver for the extended CHC with only small efforts.
	We also demonstrate the ability of our type checker to verify various concrete examples.
\end{abstract}




\maketitle

\section{Introduction}\label{sec:introduction}
\newcommand{\ExtNonnegRealType}{\mathbf{real}_{\ge 0}^{\infty}}

\emph{Probabilistic programs} have been studied from two perspectives: for writing randomised algorithms and for describing probabilistic models and Bayesian inferences.
For these purposes, many probabilistic programming languages have been proposed and developed, including Anglican \cite{wood2014}, LazyPPL \cite{dash2023}, Church \cite{goodman2008}, and Hakaru \cite{narayanan2016}.
To facilitate programming, these languages have several features such as sampling from continuous distributions, conditioning for modelling posterior distributions, and sometimes higher-order functions.
On the other hand, these features pose challenges to the verification of probabilistic programs, invalidating many existing verification methods for non-probabilistic programs.
Although many studies have been devoted to this challenge \cite{mciver2005,olmedo2018,avanzini2021,sato2019a}, \emph{automated} verification is still a challenging problem.
Furthermore, there are relatively few studies on automated verification of \emph{higher-order} probabilistic programs compared to imperative probabilistic programs.

In this paper, we propose a dependent refinement type system for verifying various properties of higher-order probabilistic programs with continuous distributions (Fig.~\ref{fig:outline}).
We consider a generalised higher-order fixed point logic (HFL).
Given a probabilistic program, we encode properties of the program as terms of the HFL using CPS (Continuation-Passing Style) transformations.
Such encodings are available for several properties, including weakest pre-expectations, expected costs, moments of cost, and conditional weakest pre-expectations \cite{avanzini2021,kura2023}.
Then, we give a specification of a probabilistic program as a refinement type.
Finally, we type-check the term against the specification using our dependent refinement type system to verify whether the specification is satisfied.
In our approach, only top-level declarations require annotations of refinement types, and such annotations need \emph{not} be inductive because our type checker can automatically synthesise inductive invariants.

Our main contribution is the \emph{automation} of the CPS-based verification methods \cite{avanzini2021,kura2023} for higher-order probabilistic programs, which is achieved by introducing a new dependent refinement type system for HFL.
The only method that we are aware of for reasoning about CPS-transformed probabilistic programs is a program logic EHOL proposed by \cite{avanzini2021}.
However, it is not obvious how to automate reasoning by EHOL since logical formulas considered in EHOL contain higher-order terms.
Therefore, we developed a more automatable method, namely, a refinement type system for CPS-transformed programs.
Using our refinement type system, we just need to solve constraints that contain only first-order terms.
Moreover, our refinement type system is more general than EHOL in the sense that it can reason about properties beyond the expected cost.
Our work also supports continuous distributions, which are not supported by \cite{avanzini2021}.
The CPS-based approach has the following advantages over other approaches for higher-order probabilistic programs \cite{beutner2021,avanzini2023}: (1) it provides a unified framework for various verification problems, and (2) it allows a reduction to general verification frameworks developed for non-probabilistic programs (e.g.\ CHC solving \cite{bjorner2015}).

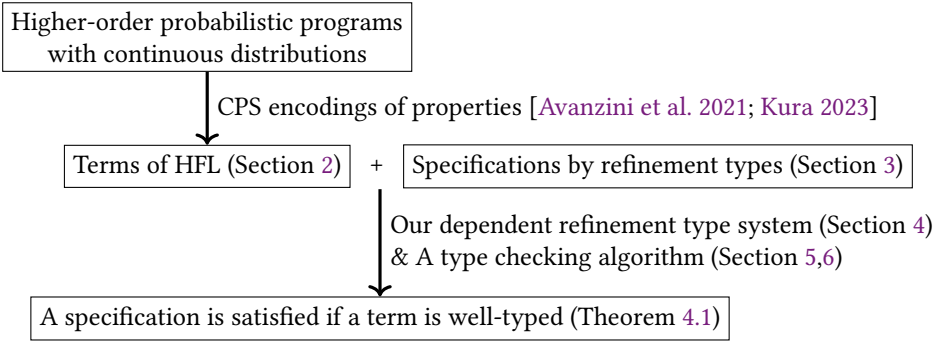
\begin{figure}
	\begin{tikzpicture}
		\node[align=center, draw] (prog) at (-3, 0.7) {Higher-order probabilistic programs \\ with continuous distributions};
		\node[draw] at (-3, -1) (hfl) {Terms of HFL (Section~\ref{sec:hfl})};
		\node[draw] at (3, -1) (spec) {Specifications by refinement types (Section~\ref{sec:instance})};
		\path (hfl) -- node[midway] (hflspec) {+} (spec);
		\node[draw] at (-0.7, -3) (ans) {A specification is satisfied if a term is well-typed (Theorem~\ref{thm:soundness})};
		\draw[->, very thick] (prog) -- node[midway, right] {CPS encodings of properties \cite{avanzini2021,kura2023}} (hfl);
		\draw[->, very thick] (-0.7, -1.3) -- node[midway, right, align=left] {Our dependent refinement type system (Section~\ref{sec:refinement}) \\ \& A type checking algorithm (Section~\ref{sec:type-check},\ref{sec:implementation})} (ans);
	\end{tikzpicture}
	\caption{Outline of our approach.}
	\label{fig:outline}
\end{figure}

Our approach enhances \emph{reusability} of both theory and implementation.
From a theoretical point of view, our design of the dependent refinement type system (Section~\ref{sec:refinement}) is guided by categorical semantics, which naturally yields a soundness theorem (Theorem~\ref{thm:soundness}) that holds for various situations.
This general design leads to novel automated verification methods for several properties (including the expected cost, the weakest pre-expectation, the cost moment, and the conditional weakest pre-expectation of probabilistic programs) while minimising the need for individual proofs for each property.
Although our paper specifically concentrates on probabilistic programs, our theoretical framework also goes beyond the verification of probabilistic programs (e.g.\ HFL model checking of event sequences \cite{kobayashi2018,kura2023}).
From the point of view of implementation, our approach can seamlessly integrate with existing type checkers for dependent refinement type systems designed for non-probabilistic programs.
This extension to probabilistic programs involves only two key modifications to constraint solvers: supporting admissible predicates for reasoning about fixed points (Section~\ref{subsec:admissible}) and supporting integration operators for reasoning about continuous distributions (Section~\ref{subsec:integration}).
These two modifications can be implemented as simple extensions of existing solvers, as we will explain in Section~\ref{sec:type-check}, and the remaining components of type checkers can be reused from those for non-probabilistic programs.
This simplicity has the advantage that we can easily incorporate advanced techniques and sophisticated constraint solvers developed for non-probabilistic programs.

\subsection{Refinement Types After CPS Transformations}
Our dependent refinement type system is intended to be used in combination with CPS-based methods to translate a given program to a term that describes a property of the program.
In the context of verification of probabilistic programs, \cite{avanzini2021} proposed a CPS transformation that translates a higher-order probabilistic program to a term of pure lambda calculus with fixed points such that the denotation of the term is equal to the expected cost of the given probabilistic program.

\newcommand{\coinflip}{\mathrm{coin}}
\begin{example}[the expected cost of coin flipping]\label{ex:coin-flip}
	For example, consider the following program that keeps flipping a fair coin until we get a head while counting the number of tails.
	\[ \letrec{\coinflip}{x}{\probbranch{(\coinflip\ ())^{\checkmark}}{1/2}{()}}{\coinflip\ ()} \]
	Here, $\probbranch{}{1/2}{}$ means probabilistic branching that takes each branch with probability $1/2$, and we consider the expected number of occurrences of $({-})^{\checkmark}$ in execution traces.
	Following \cite{avanzini2021}, we apply a CPS transformation and pass $\lambda y. 0$ as a continuation.
	Then, we obtain the following term, whose denotation is the expected cost of the probabilistic program above.
	\begin{equation}
		\letfix{\coinflip}{x\ k}{1/2 \cdot (1 + \coinflip\ ()\ k) + 1/2 \cdot k\ ()}{\coinflip\ ()\ (\lambda y. 0)}
		\label{eq:coin-flip-cpsed}
	\end{equation}
	Note that $\coinflip$ has type $\coinflip : \mathbf{unit} \to (\mathbf{unit} \to \ExtNonnegRealType) \to \ExtNonnegRealType$ where $\ExtNonnegRealType$ is a type of non-negative extended real numbers in $[0, \infty]$.
	Note also that in the above CPS transformation, $\oplus_p$ is mapped to $p \cdot ({-}) + (1 - p) \cdot ({-})$, and $({-})^{\checkmark}$ is mapped to $1 + ({-})$.
	Intuitively, the term~\eqref{eq:coin-flip-cpsed} represents the limit of the converging sequence below.
	\begin{equation}
		c_0 = 0, \quad c_1 = 1/2 \cdot (1 + c_0) + 1/2 \cdot 0 = 1/2, \quad c_2 = 1/2 \cdot (1 + c_1) + 1/2 \cdot 0 = 3/4, \quad \dots \to 1
	\end{equation}
	Technically, we get the converging sequence above because the interpretation of $\mathbf{let}\ \mathbf{fix}$ is different from the standard interpretation of recursion: $\mathbf{let}\ \mathbf{fix}$ is the least fixed point with respect to the standard order on $[0, \infty]$.
	\qed
\end{example}

Recently, \cite{kura2023} generalised this CPS translation to a wide class of weakest precondition transformers: we can translate a given higher-order program into a term of a generalised higher-order fixed point logic (HFL) so that the term describes a certain kind of properties of the given program.
In terms of verification of probabilistic programs, this generalisation specifically subsumes four properties: weakest pre-expectations, expected costs, moments of cost, and conditional weakest pre-expectations.
However, how to reason about HFL terms was out of the scope of these studies.

To automatically reason about HFL terms, we propose a dependent refinement type system for HFL.
For example, consider the term~\eqref{eq:coin-flip-cpsed} in Example~\ref{ex:coin-flip}.
As a specification, we may consider the following refinement type, which reads ``the expected cost $r$ satisfies $r \le 1$'' or more simply ``1 is an upper bound of the expected cost'', assuming that $\lambda y. 0 : \mathbf{unit} \to \{ r : \ExtNonnegRealType \mid r = 0 \}$ is passed as the second argument.
\begin{equation}
	\coinflip : (x : \mathbf{unit}) \to (k : \mathbf{unit} \to \{ r : \ExtNonnegRealType \mid r = 0 \}) \to \{ r : \ExtNonnegRealType \mid r \le 1 \}
	\label{eq:coin-flip-refinemnt-type}
\end{equation}

To verify that~\eqref{eq:coin-flip-refinemnt-type} is a correct refinement type for the term~\eqref{eq:coin-flip-cpsed}, we provide a set of typing rules for dependent refinement types.
We also prove that typing rules are sound.
By utilising mathematical abstraction via categorical semantics, the soundness theorem is proved in a general situation, subsuming any instance supported by the framework of \cite{kura2023}.
Therefore, we obtain verification methods for higher-order probabilistic programs simply as special cases.

Our mathematical abstraction makes it easier to extend the verification method limited to discrete distributions \cite{avanzini2021} to a method that allows continuous distributions too.
When we apply our approach to continuous distributions, we only need to check if a model for continuous distributions satisfies a certain set of axioms that make our soundness theorem hold.
In fact, there exists such a model: the category $\omegaQBS$ of $\omega$-quasi-Borel spaces \cite{vakar2019}.
This process is much easier than repeating the whole soundness proof for the case of continuous distributions.

\subsection{Key Modifications of Type Checkers and Constraints Solvers}
We implement a type checker for our dependent refinement type system as an extension of an existing type checker
\textsc{RCaml}\footnote{https://github.com/hiroshi-unno/coar}
for non-probabilistic programs.
There are two key modifications involved in this extension.

\subsubsection{Admissible Predicates}
As we explained in Example~\ref{ex:coin-flip}, the semantics of \emph{fixed points} $\mathbf{let}\ \mathbf{fix}$ in our HFL is given by a different order relation from the standard semantics of \emph{recursion}.
In the following, we distinguish \emph{fixed points} and \emph{recursion} as different concepts.
Because of the difference above, standard typing rules for recursion are unsound for fixed points (as we will explain in Section~\ref{sec:refinement-typing-rule}).
For example, the standard typing rule for partial correctness of recursion can derive the following type for Example~\ref{ex:coin-flip}.
\[ \coinflip : (x : \mathbf{unit}) \to (k : \mathbf{unit} \to \{ r : \ExtNonnegRealType \mid r = 0 \}) \to \{ r : \ExtNonnegRealType \mid r < 1 \} \]
However, this is incorrect because the true expected cost of Example~\ref{ex:coin-flip} is $r = 1$.

Following EHOL~\cite{avanzini2021}, we adapt a standard typing rule for recursion by imposing the admissibility condition on a predicate on the result type of a fixed point.
This adaptation is essential for obtaining a sound typing rule, especially when reasoning about quantitative properties expressed as the least fixed point.
Technically, a predicate $P \subseteq [0, \infty]$ is \emph{admissible} if the bottom element $0 \in [0, \infty]$ satisfies $P$ and the predicate is closed under the supremum $\sup_n x_n$ of any $\omega$-chain $x_0 \le x_1 \le \dots$ such that $x_n \in P$ for any $n$.
For example, \eqref{eq:coin-flip-refinemnt-type} is a correct type for Example~\ref{ex:coin-flip} because $r \le 1$ is admissible, but we cannot replace $r \le 1$ with $r < 1$ because $r < 1$ is not admissible (not closed under $\sup$).
Note that $r = 1$ is (unfortunately) not admissible either and thus cannot be used here.
In general, it is difficult to reason about the exact value or lower bounds in probabilistic program verification \cite{feng2023,hark2020,mciver2005,beutner2021}, which is why most of the studies (including ours) focus on upper bounds instead of the exact value of the expected cost.

To implement fixed points in our type checker, we extend our type checker and the constraint solver used in our type checker so that only admissible predicates are synthesised for fixed points.
This extension doesn't break the basic design of the existing implementation: what we need is (1) in the type checker, to implement the typing rule for fixed points and (2) in the constraint solver (specifically, a solver for constrained Horn clauses), to support a new type of predicate variables for admissible predicates, which we call \emph{admissible predicate variables}.
The former is rather straightforward.
For the latter, we first consider syntactic criteria to check whether a predicate is admissible or not.
Then, we restrict the search space for admissible predicates by the syntactic criteria.

The novelties here are (1) incorporating admissible predicates explicitly into a refinement type system and (2) supporting admissible predicates in a type checker.
Technically, ensuring admissibility is important for \emph{any} refinement type system when proving the soundness of typing rules for recursion or fixed points.
However, there hasn't been any refinement type system that requires admissibility explicitly because it is typically straightforward to ensure admissibility when reasoning about recursion, as we explain in Section~\ref{sec:refinement}.
For reasoning about fixed points, we propose a novel type checker that supports admissible predicates.
To the best of our knowledge, no other type checker of this nature exists.

\subsubsection{Integration Operators}
Although our dependent refinement type system is proved sound in a general situation, it is, to some extent, inevitable that we have to consider how to deal with basic (or primitive) operators, which may vary depending on programming languages targeted for verification.
Specifically, this is the case for integration operators for continuous distributions.
Since we consider probabilistic programs with continuous distributions (e.g.\ the uniform distribution over the unit interval $[0, 1]$), we have to include corresponding integration operators (e.g.\ $\int_0^1 ({-})\, \mathrm{d} x$) in our HFL so that we can reason about properties defined by expected values (see Section~\ref{sec:instance-continuous-distribution}).
However, integration operators are usually not supported by constraints solvers nor even considered in dependent refinement type systems.

To support integration operators, we provide a typing rule for an integration operator in Section~\ref{sec:typing-basic-operators}.
Then, we introduce another new type of predicate variables, which we call \emph{integrable predicate variables} so that constraints for integration operators can be expressed appropriately.
Finally, we extend our constraint solver to support integrable predicate variables.
Our extension is rather simple: we restrict templates for integrands to affine expressions so that we can easily compute integrals by linearity of integration.

\subsection{Summary of Contributions}
\begin{itemize}
	\item We provide a dependent refinement type system for a generalised HFL.
	We proved its soundness for a general class of models.
	Combined with CPS-based methods, our approach can verify the following properties of higher-order probabilistic programs with continuous distributions: weakest pre-expectations, expected costs, moments of cost, and conditional weakest pre-expectations.
	\item By introducing two new types of predicate variables (i.e.\ admissible and integrable predicate variables) to CHC, we present a reduction from the type-checking problem for our dependent refinement type system to constraint solving for the extended CHC.
	We also propose an extension of a CHC solver that supports admissible and integrable predicate variables by considering appropriate templates.
	A notable advantage of our approach lies in the ability to reuse most of the basic design of implementations except for the two types of predicate variables.
	This facilitates the seamless integration of advanced techniques developed for the verification of non-probabilistic programs.
	\item We implement a type checker and demonstrate its ability to verify various properties of benchmark programs.
	Our implementation solved 11 out of 14 benchmarks ranging over four kinds of verification problems for higher-order probabilistic programs.
	No other tool even supports as many kinds of benchmarks as ours.
	Since our approach does not depend on a choice of constraint solvers, there is considerable freedom in the choice of generic constraint solvers.
	Moreover, such constraint solvers need \emph{not} be tailored for probabilistic programs thanks to the CPS transformation.
	Therefore, future developments of generic constraint solvers will potentially improve experimental results.
\end{itemize}

\subsection{Organization of the Paper}
In Section~\ref{sec:hfl}, we review the definition of the generalised higher-order fixed point logic $\lambdaHFL$ \cite{kura2023}.
In Section~\ref{sec:instance}, we explain that CPS transformations (also proposed in \cite{kura2023}) from probabilistic programs to $\lambdaHFL$ can capture various interesting properties.
We also present our idea of using refinement types as specifications of probabilistic programs.
In Section~\ref{sec:refinement}, we present our dependent refinement type system and its soundness.
In Section~\ref{sec:type-check}, we explain a reduction from the type-checking problem to constraint solving where we introduce an extension of CHC.
In Section~\ref{sec:implementation}, we show how to implement a constraint solver by modifying an existing CHC solver and present experimental results.
We discuss related work in Section~\ref{sec:related-work}.
Section~\ref{sec:conclusion} contains conclusions and future work.

\section{A Generalised Higher Order Fixed Point Logic}\label{sec:hfl}

In this section, we review the definition of a generalised higher-order fixed point logic $\lambdaHFL$ \cite{kura2023} that can describe properties of higher-order programs.
We will explain later in Section~\ref{sec:instance} how to encode properties of a given program into a term of $\lambdaHFL$.

\begin{table}
	\caption{Truth values for reasoning about probabilistic programs.
	Truth values are equipped with partial orders so that we can define least fixed points.
	The superscript $({-})^{\op}$ indicates that we consider the opposite partial order (lfp w.r.t.\ the opposite order = gfp w.r.t.\ the original order).}
	\label{tab:truth-value}
	\begin{tabular}{l|c}
		Property & Truth values ($\answertype$) \\
		\hline
		Weakest pre-expectation \cite{mciver2005} & $[0, \infty]$ \\
		Expected cost \cite{kaminski2018} & $[0, \infty]$ \\
		Cost moment \cite{kura2019} & $[0, \infty]^n$ \\
		Conditional weakest pre-expectation \cite{olmedo2018} & $[0, \infty] \times [0, 1]^{\op}$
	\end{tabular}
\end{table}

\subsection{Syntax}
\begin{figure}
	\[ \sigma, \tau \coloneqq \answertype \mid b \mid 1 \mid \sigma \times \tau \mid 0 \mid \sigma + \tau \mid \sigma \to \answertype \qquad \qquad \Gamma \coloneqq \cdot \mid \Gamma, x : \tau \]
	\begin{align}
		M, N    & \coloneqq x & & \text{variable} &
		&\mid \mathbf{op}\ M && \text{basic operators} \\
		&\quad \mid () \quad\mid (M, N) & & \text{0-tuple, 2-tuple} &
		&\mid \pi_i\ M & & \text{projection ($i = 1, 2$)} \\
		&\quad \mid \delta(M) & & \text{nullary case analysis} &
		& \mid \iota_i\ M & & \text{injection ($i = 1, 2$)} \\
		&\quad \mid \caseexpr{M}{x_1 : \sigma_1}{N_1}{x_2 : \sigma_2}{N_2} & & \text{binary case analysis} &
		&\mid \lambda x : \sigma. M & & \text{lambda abstraction} \\
		&\quad \mid \fixpoint{(f : \sigma \to \answertype)}{M} & & \text{fixed point} &
		&\mid M\ N & & \text{function application}
	\end{align}
	\vspace{-1em}
	\Description{The definition of simple types, simple contexts, and terms of $\lambdaHFL$ is defined in BNF.}
	\caption{Simple types $\sigma, \tau$; simple contexts $\Gamma$; and terms $M, N$.
	In the definition above, $x$ ranges over variables, $b$ over base types in $\BaseTypes$, and $\mathbf{op}$ over basic operators in $\BasicOps$.
	We assume variables in a context are mutually distinct.}
	\label{fig:hfl-syntax}
\end{figure}

\begin{figure}
	\begin{mathpar}
		\inferrule[S-BasicOp]{
			\Gamma \vdash M : \mathrm{ar}(\mathbf{op})
		}{
			\Gamma \vdash \mathbf{op}\ M : \mathrm{car}(\mathbf{op})
		}
		\and
		\inferrule[S-Abs]{
			\Gamma, x : \sigma \vdash M : \answertype
		}{
			\Gamma \vdash \lambda x : \sigma. M : \sigma \to \answertype
		}
		\and
		\inferrule[S-Fix]{
			\Gamma, f : \sigma \to \answertype, \vdash M : \sigma \to \answertype
		}{
			\Gamma \vdash \fixpoint{(f : \sigma \to \answertype)}{M} : \sigma \to \answertype
		}
	\end{mathpar}
	\caption{Selected typing rules.}
	\label{fig:hfl-simple-typing}
\end{figure}

Following \cite{kura2023}, $\lambdaHFL$ is defined as a simply typed lambda calculus with fixed points.
We define \emph{simple types}, \emph{simple contexts}, and \emph{terms} of $\lambdaHFL$ as in Fig.~\ref{fig:hfl-syntax}.
In the definition of simple types, $\answertype$ is a type of truth values, which need not be boolean values.
The meaning of $\answertype$ depends on the verification problems we consider.
When verifying probabilistic programs, it is common to use quantitative values as ``truth values'' (see Table~\ref{tab:truth-value} and Section~\ref{sec:instance}).
For example, we use non-negative extended real numbers $[0, \infty]$ as truth values for expected cost analyses, in which case, $\lambdaHFL$ is essentially the same as the target language of \cite{avanzini2021}.
The unit type $1$ (or sometimes written as $\mathbf{unit}$) is inhabited by only one value, and the empty type $0$ is inhabited by no value.
We use the sum type $1 + 1$ as a type of boolean values where $\mathbf{true} \coloneqq \iota_1\ ()$ and $\mathbf{false} \coloneqq \iota_2\ ()$.
Although $\lambdaHFL$ is designed as a higher-order logic, we \emph{never} call a term of $\lambdaHFL$ a \emph{formula} to distinguish terms from first-order formulas defined in Section~\ref{sec:refinement}.

We have two parameters for syntax: a set $\BaseTypes$ of base types and a set $\BasicOps$ of basic operators.
Each element in $\BasicOps$ is a triple of the form $(\mathbf{op} : \mathrm{ar}(\mathbf{op}) \rightarrowtriangle \mathrm{car}(\mathbf{op})) \in \BasicOps$ where $\mathrm{ar}(\mathbf{op})$ and $\mathrm{car}(\mathbf{op})$ are types called an \emph{arity} and a \emph{coarity} of $\mathbf{op}$, respectively.
For example, we often include type $\mathbf{int}$ of integers and type $\mathbf{real}$ of real numbers as base types in $\BaseTypes$.
We can introduce constants of type $\mathbf{int}$ to $\lambdaHFL$ by adding a basic operator $n : 1 \rightarrowtriangle \mathbf{int}$ for any $n \in \mathbb{Z}$, and similarly for constants of other types.
Basic arithmetic operators like $({+}_{\mathbf{int}}) : \mathbf{int} \times \mathbf{int} \rightarrowtriangle \mathbf{int}$ and comparison operator like $({\le}_{\mathbf{int}}) : \mathbf{int} \times \mathbf{int} \rightarrowtriangle 1 + 1$ are also examples of basic operators.

The type $\answertype$ plays a special role in the typing system of $\lambdaHFL$.
\emph{Well-typed term} $\Gamma \vdash M : \tau$ are defined by rather standard typing rules except that the codomain of function types are restricted to $\answertype$ (see Fig.~\ref{fig:hfl-simple-typing}), which leads to a restriction of lambda abstractions and fixed points.
This is not too restrictive because we always get this form of lambda abstractions and fixed points by applying CPS transformations with $\answertype$ as an answer type.

We introduce several syntactic sugars as follows.
We define if expressions $\ifexpr{M}{N_1}{N_2} \coloneqq \caseexpr{M}{x_1}{N_1}{x_2}{N_2}$ by case analyses where $x_1$ and $x_2$ are fresh variables since $1 + 1$ is used as the type of boolean values.
For any constant $\mathbf{op} : 1 \rightarrowtriangle \sigma$, we usually omit argument $()$ and write $\mathbf{op} \coloneqq \mathbf{op}\ ()$.
We often use infix notation for arithmetic operators and comparison operators: more formally, if $\mathbf{op} : \sigma_1 \times \sigma_2 \rightarrowtriangle \tau$ is a basic operator where $\sigma_1, \sigma_2 \in \BaseTypes \cup \{ \answertype \}$, we define $M \mathrel{\mathbf{op}} N \coloneqq \mathbf{op}\ (M, N)$.
We define let-fix expressions by $\letfix{f}{x}{M}{N} \coloneqq (\lambda f. N)\ (\fixpoint{f}{\lambda x. M})$.

The restriction of function types forces us to use uncurried functions like $f : \mathbf{int} \times \mathbf{int} \to \answertype$ instead of curried one $f : \mathbf{int} \to (\mathbf{int} \to \answertype)$.
For convenience, we think of a curried function as a syntactic sugar for the corresponding uncurried function: $\lambda x. \lambda y. M \coloneqq \lambda (x, y). M \coloneqq \lambda z. M[\pi_1\ z/x, \pi_2\ z/y]$ where $M$ is a term of type $\answertype$ (use $\eta$-expansion to $M$ if needed).
Correspondingly, a partial application of $M : \sigma \times \tau \to \answertype$ to $N : \sigma$ is understood as $\lambda y : \tau. M\ (N, y) : \tau \to \answertype$.
We use a similar syntactic sugar for fixed points too.

\subsection{Denotational Semantics}

Technically, the denotational semantics for $\lambdaHFL$ is defined in a general way for a certain class of categorical models, which we call \emph{$\lambdaHFL$-models}.
Although this gives flexibility of choosing an appropriate $\lambdaHFL$-model for each problem of program verification, explaining $\lambdaHFL$-models in the most general form requires a certain level of familiarity with category theory.
For the sake of non-category theorists, we restrict our description in this section to a specific $\lambdaHFL$-model, namely, the category $\omegaCPO$ of $\omega$cpos and Scott-continuous functions, using as few terminology of category theory as possible.
This restriction to $\omegaCPO$ is just for simplicity of explanation: our soundness theorem in Section~\ref{sec:refinement} holds for any $\lambdaHFL$-model, including $\omegaQBS$ models used for continuous distributions in Section~\ref{sec:instance}.
Denotational semantics in general $\lambdaHFL$-models can be found in \referappendix{sec:detail-hfl}{A}.

We recall basic definitions of $\omega$cpos and Scott-continuous functions.
An \emph{$\omega$-chain} in a partially ordered set $(X, {\le})$ is a sequence $\{ x_n \in X \}_n$ of elements in $X$ such that $x_n \le x_{n + 1}$ for any $n$.
An \emph{$\omega$cpo} is a pair $(X, {\le}_X)$ of a set $X$ and a partial order ${\le}_X$ such that the supremum of any $\omega$-chain in $X$ exists.
Note that we often leave ${\le}_X$ implicit and write $X = (X, {\le}_X)$.
A \emph{Scott-continuous function} between two $\omega$cpos $(X, {\le}_X)$, $(Y, {\le}_Y)$ is a monotone function $f : X \to Y$ that preserves supremum of $\omega$-chains.
Let $\omegaCPO$ be the category of $\omega$cpos and Scott-continuous functions.
Recall that for any $\omega$cpo $(X, {\le}_X)$, $(Y, {\le}_Y)$, we have their product $(X, {\le}_X) \times (Y, {\le}_Y) = (X \times Y, {\le}_{X \times Y})$, coproduct (or sum) $(X, {\le}_X) + (Y, {\le}_Y) = (X + Y, {\le}_{X + Y})$, and function space $\exponential{(X, {\le}_X)}{(Y, {\le}_Y)} = (\{ f : X \to Y \mid \text{$f$ is Scott-continuous} \}, {\le}_{\exponential{X}{Y}})$ in $\omegaCPO$ where partial orders are defined by $(x, y) \le_{X \times Y} (x', y')$ iff $x \le_X x' \land y \le_Y y'$; $x \le_{X + Y} x'$ iff $x \le_X x' \lor x \le_Y x'$; and $f \le_{\exponential{X}{Y}} g$ iff $\forall x \in X, f(x) \le_Y g(x)$.

Now, we define the interpretation of types.
Suppose that interpretations of each base type $b \in \BaseTypes$ and the type $\answertype$ of truth values are given: let $\interpret{b}$ and $\interpret{\answertype}$ be $\omega$cpos that give an interpretation of $b$ and $\answertype$, respectively.
For example, we usually interpret $\mathbf{int}$ and $\mathbf{real}$ by $\omega$cpos $(\mathbb{Z}, {=})$ and $(\mathbb{R}, {=})$ equipped with discrete orders.
Then, we extend $\interpret{-}$ to the interpretation of types as follows.
\begin{gather}
	\interpret{b} \coloneqq \interpret{b} \qquad
	\interpret{\answertype} \coloneqq \interpret{\answertype} \qquad
	\interpret{0} \coloneqq (\emptyset, \emptyset) \qquad
	\interpret{\sigma + \tau} \coloneqq \interpret{\sigma} + \interpret{\tau} \\
	\interpret{1} \coloneqq (\{ \star \}, {=}) \qquad
	\interpret{\sigma \times \tau} \coloneqq \interpret{\sigma} \times \interpret{\tau} \qquad
	\interpret{\sigma \to \answertype} \coloneqq \exponential{\interpret{\sigma}}{\interpret{\answertype}}
\end{gather}

A context is interpreted by $\interpret{x_1 : \sigma_1, \dots, x_n : \sigma_n} \coloneqq \prod_{i = 1}^{n} \interpret{\sigma_i}$, or equivalently, by the set of mappings defined below.
Note that the interpretation $\interpret{\Gamma}$ of a context is an $\omega$cpo whose partial order is given by the pointwise order: $\gamma \le \gamma'$ if and only if $\gamma(x) \le \gamma'(x)$ for any $x$.
\[ \interpret{x_1 : \sigma_1, \dots, x_n : \sigma_n} \coloneqq \{ \gamma \mid \dom(\gamma) = \{ x_1, \dots, x_n \} \land \forall (x_i : \sigma_i), \gamma(x_i) \in \interpret{\sigma_i} \} \]
The empty mapping for the empty context is denoted by $\emptyenv$.

We define the interpretation of terms as follows.
Suppose that an interpretation of each basic operator in $\BasicOps$ is given: for each basic operator $\mathbf{op} : \sigma \rightarrowtriangle \tau$ in $\BasicOps$, let $\interpret{\mathbf{op}} : \interpret{\sigma} \to \interpret{\tau}$ be a Scott-continuous function that gives an interpretation of $\mathbf{op}$.
For any well-typed term $\Gamma \vdash M : \sigma$, the interpretation is defined as a Scott-continuous function $\interpret{\Gamma \vdash M : \sigma} : \interpret{\Gamma} \to \interpret{\sigma}$ (we often write $\interpret{M} = \interpret{\Gamma \vdash M : \sigma}$ for simplicity).
Most parts of the definition of $\interpret{M}$ are standard (see \referappendix{sec:detail-hfl}{A} for details).
The interpretation of $\Gamma \vdash \fixpoint{f}{M} : \sigma \to \answertype$ is defined by the least fixed point of $\interpret{M}(\gamma[f \mapsto ({-})]) : \exponential{\interpret{\sigma}}{\interpret{\answertype}} \to \exponential{\interpret{\sigma}}{\interpret{\answertype}}$ where $\gamma \in \interpret{\Gamma}$.
\[ \interpret{\fixpoint{(f : \sigma \to \answertype)}{M}}(\gamma) = \mathrm{lfp}^{\exponential{\interpret{\sigma}}{\interpret{\answertype}}} \interpret{M}(\gamma[f \mapsto ({-})]) \]
For example, if $\interpret{\answertype} = [0, 1]^{\op} = ([0, 1], {\ge})$, then a fixed point $\fixpoint{f}{M}$ is interpreted by the lfp with respect to $[0, 1]^{\op} = ([0, 1], {\ge})$, which is the gfp with respect to $[0, 1] = ([0, 1], {\le})$.

Note that how we interpret the type $\answertype$, base types in $\BaseTypes$, and basic operators in $\BasicOps$ is considered as a part of a $\lambdaHFL$-model.
Technically, a $\lambdaHFL$-model is defined as a combination of (1) a category $\category{C}$ with sufficient structures (product / coproduct / function space constructions and fixed points) and (2) an assignment of an interpretation to each base type, each basic operator, and $\answertype$ (see \referappendix{sec:detail-hfl}{A}).

\section{Encoding Various Problems to HFL Terms}\label{sec:instance}
In this section, we show several concrete verification problems for probabilistic programs.
We review CPS-based encodings proposed in \cite{kura2023} and exemplify how various properties of probabilistic programs can be translated to $\lambdaHFL$.
We also explain $\lambdaHFL$-models used for each verification problem and specifically how $\answertype$ is interpreted in each $\lambdaHFL$-model.
To motivate the dependent refinement type system in Section~\ref{sec:refinement}, we describe how refinement types are used to give specifications of probabilistic programs.

\subsection{A Brief Review of the CPS-Based Encodings}
\newcommand{\CPS}[1]{#1^{\sharp}}

The correctness of the CPS-based encodings below is guaranteed by the general framework proposed by \cite{kura2023}, which we review here to improve self-containedness.
The purpose of the framework is to obtain the weakest precondition of a given effectful program as a term of HFL.
This is achieved by considering a CPS transformation that translates an effectful program $M$ into a term of HFL $\CPS{M}$.
As a language for effectful programs, they consider a functional programming language with recursion and algebraic effects.
Intuitively, algebraic effects are primitive operations for causing computational effects.
Examples of algebraic effects include a probabilistic branching operation, an operation for sampling from a probability distribution, and the tick operator for increasing cost.
Note that although this paper focuses on probabilistic programs, the CPS-based encodings can be applied to more general algebraic effects, as detailed in \cite{kura2023}.

On the semantic side, a program $M$ is interpreted as a morphism $\interpret{M} : X \to T Y$ in a category $\category{C}$ where $T$ is a strong monad that represents computational effects.
Terms of HFL are interpreted in the same category $\category{C}$ as a pure simply typed lambda calculus as we explained in Section~\ref{sec:hfl}.
The weakest precondition for $f : X \to T Y$ is defined as follows.
Let $\answerobj \in \category{C}$ be an object that represents a set of truth values (e.g.\ $\answerobj = \{ 0, 1 \}$ or $\answerobj = [0, \infty]$).
Then, the postcondition is given by $p : Y \to \answerobj$.
Assuming that we have an Eilenberg--Moore algebra $\emalgsymbol : T \answerobj \to \answerobj$, the weakest precondition is defined by $\mathrm{wp}[f](p) = \emalgsymbol \comp T p \comp f : X \to \answerobj$.
It is known that many properties of programs can be expressed using weakest preconditions if we choose $\emalgsymbol : T \answerobj \to \answerobj$ appropriately \cite{aguirre2022}.
For example, the weakest pre-expectation of a probabilistic program $M$ is given by $\mathrm{wp}[\interpret{M}](p)$ where $T$ is the probability distribution monad, $\answerobj = [0, \infty]$, and $\emalgsymbol : T \answerobj \to \answerobj$ is the function defined by the expected value $\emalgsymbol(\mu) \coloneqq \mathbb{E}[\mu]$ of $\mu \in T [0, \infty]$.

Now, the CPS transformation gives the weakest precondition in the following sense.
\begin{theorem}[from \cite{kura2023}]\label{thm:cps-wpt}
	Assume that $\category{C}$ is a ``good'' model (which is almost the same as requiring $\category{C}$ to be a $\lambdaHFL$-model, but see \cite{kura2023} for details).
	For a well typed program $\Gamma \vdash M : \rho$ and a postcondition (a term of HFL) $x : \rho \vdash P : \answertype$, we have $\mathrm{wp}^{\emalgsymbol}[\interpret{M}] (\interpret{P}) = \interpret{\CPS{M}\ (\lambda x : \rho. P)}$ if types in $\Gamma$ and $\rho$ are constructed without $\to$.
	Here, $\Gamma \vdash \CPS{M} : (\rho \to \answertype) \to \answertype$ is the term obtained by applying the CPS transformation to $M$.
\end{theorem}
Note that the restriction on $\Gamma$ and $\rho$ does \emph{not} exclude the use of lambda abstraction or recursion in $M$.

Due to the nature of the CPS-based encoding, our method is inherently constrained to specifications presented as refinement types on the \emph{weakest precondition transformer} (WPT) of a given program.
Some properties are difficult to specify using WPTs although we can express many important properties using WPTs, as we will explain in the rest of this section.
For example, it is difficult to express non-functional, temporal liveness properties as WPTs, and thus, such properties are out of the scope of our approach.

\subsection{A Model for Higher Order Probabilistic Programs with Continuous Distributions}
Technically, the category $\omegaCPO$ of $\omega$cpos is not sufficient to model continuous distributions like uniform distributions and Gaussian distributions.
To interpret continuous distribution, we use the category $\omegaQBS$~\cite{vakar2019} of $\omega$qbses as a $\lambdaHFL$-model.
Intuitively, an \emph{$\omega$qbs} is a set $X$ equipped with both an $\omega$cpo structure $(X, {\le}_X)$ and a quasi-Borel-space structure $(X, M_X)$ in a compatible way (see \cite{vakar2019,heunen2017} for what $M_X$ means).
Quasi-Borel spaces are a generalisation of measurable spaces~\cite{heunen2017}.
Compared to the traditional measure-theoretic treatment of continuous distributions, quasi-Borel spaces have the advantage that quasi-Borel spaces are closed under the construction of function spaces: the function space $\exponential{X}{Y}$ between two $\omega$qbses $X, Y$ is again an $\omega$qbs.
Since $\omegaQBS$ has sufficient categorical structures as a $\lambdaHFL$-model, denotational semantics of $\lambdaHFL$ in $\omegaQBS$ is defined in the same way as the case of the $\omegaCPO$ model explained in Section~\ref{sec:hfl}, and more importantly, our soundness theorem in Section~\ref{sec:refinement} is also true for $\omegaQBS$.
To keep our explanation simple, we often leave quasi-Borel-space structures implicit in this section and just write $\omega$cpo structures even when we use $\omegaQBS$ as a model (see \referappendix{sec:detail-hfl}{A} for details).

\subsection{Expected Cost}\label{subsec:instance-expected-cost}
We consider the problem of estimating the expected cost of higher-order probabilistic programs.
This problem is studied in \cite{avanzini2021}.
Our cost model is given by explicit \emph{tick operators} $({-})^{\checkmark}$, which increment the accumulated cost and do nothing else.
We consider the expected number of tick operators used until a given program terminates.
Given a probabilistic program, we can translate it to a term $M$ of $\lambdaHFL$ whose interpretation $\interpret{M}$ is the expected cost \cite{avanzini2021}.
Specifically, we use a CPS transformation that maps each occurrence of the tick operator $({-})^{\checkmark}$ to $1 + ({-})$ and each probabilistic branching $\probbranch{}{p}{}$ to $p \cdot ({-}) + (1 - p) \cdot ({-})$ (see Example~\ref{ex:coin-flip}).

Terms of $\lambdaHFL$ obtained from the CPS transformation above is interpreted in $\omegaQBS$.
The type $\answertype$ of truth values is interpreted by $[0, \infty]$ with the standard order whereas base types are interpreted by sets with discrete orders (and with the standard $\sigma$-algebra structures).
\begin{gather}
	\interpret{\answertype} = ([0, \infty], {\le}) \qquad
	\interpret{\mathbf{int}} = (\mathbb{Z}, {=}) \qquad
	\interpret{\mathbf{real}} = (\mathbb{R}, {=})
\end{gather}

The remaining problem here is how to reason about the term of $\lambdaHFL$ obtained from the CPS transformation above.
We can give a specification about the expected cost as a refinement type for the term of $\lambdaHFL$.
For example, the specification for Example~\ref{ex:coin-flip} can be given as~\eqref{eq:coin-flip-refinemnt-type}, which reads ``an upper bound of the expected cost is 1''.
To verify the correctness of such refinement types, we will provide a dependent refinement type system in Section~\ref{sec:refinement}.

\subsection{Continuous Distributions}\label{sec:instance-continuous-distribution}
Consider the expected cost analysis for a probabilistic program with continuous distributions.
To reason about continuous distribution, we add integration operators to $\lambdaHFL$ as a basic operator.
For simplicity, we focus on the uniform distribution over the unit interval $[0, 1]$.
Then, the integration operator for the uniform distribution is given as follows.
\begin{equation}
	\mathbf{unif} : (\mathbf{real} \to \answertype) \rightarrowtriangle \answertype \qquad\qquad \interpret{\mathbf{unif}}(f) = \int_0^1 f(x)\, \mathrm{d} x
	\label{eq:unif-integration-operator}
\end{equation}

\begin{example}\label{ex:random-walk-unif}
	Consider the expected cost of a continuous variant of random walks.
	\[ \letrec{\mathrm{rw}}{x}{\ifexpr{x \ge 0}{(y \leftarrow \mathbf{uniform}; (\mathrm{rw}\ (x + 3 \cdot y - 2))^{\checkmark})}{()}}{\mathrm{rw}\ 1} \]
	Here, $\mathbf{uniform}$ is the uniform distribution over $[0, 1]$
	The argument $x$ of the function $\mathrm{rw} : \mathbf{real} \to \mathbf{unit}$ is the current position, and $x = 1$ is the initial position.
	For each step, the program samples $y$ from the uniform distribution over $[0, 1]$, transforms $y$ into a sample $y' = 3 \cdot y - 2$ from the uniform distribution over $[-2, 1]$, and then adds $y'$ to the current position $x$ until $x$ becomes negative.
	To reason about the expected cost, we apply a CPS transformation and obtain the following.
	\begin{equation}
		\letfix{\mathrm{rw}}{x\ k}{
		\ifexpr{x \ge 0}{
			\mathbf{unif}\ (\lambda y. 1 + \mathrm{rw}\ (x + 3 \cdot y - 2)\ k)
		}{k\ ()}
	}{\mathrm{rw}\ 1\ (\lambda r. 0)}
	\label{eq:random-walk-unif-cpsed}
	\end{equation}
	Here, the uniform distribution $\mathbf{uniform}$ is translated to the corresponding integration operator $\mathbf{unif}$.
	The interpretation of~\eqref{eq:random-walk-unif-cpsed} gives the expected cost of the random walk.
	\qed
\end{example}

The expected cost of ``$\mathrm{rw}\ x$'' in Example~\ref{ex:random-walk-unif} is upper-bounded by $2 \cdot x + 4$.
Using refinement types, we can specify this as follows.
\begin{equation}
	\mathrm{rw} : (x : \{ x : \mathbf{real} \mid x \ge -2 \}) \to (k : \mathbf{unit} \to \{ r : \answertype \mid r = 0 \}) \to \{ r : \answertype \mid r \le |2 \cdot x + 4| \}
	\label{eq:random-walk-unif-cpsed-type}
\end{equation}
Here, $|{-}| : \mathbf{real} \rightarrowtriangle \answertype$ is a basic operator that returns the absolute value of a real number, which is used here just to ensure that the type of $|2 \cdot x + 4|$ is $\answertype$.

\subsection{Weakest Pre-Expectation}\label{subsec:instance-weakest-preexpectation}
The weakest pre-expectation for probabilistic programs \cite{mciver2005} is a quantitative generalisation of the weakest precondition for non-probabilistic programs.
Instead of using boolean predicates as pre-/post-conditions, the weakest pre-expectation transformer uses $[0, \infty]$-valued predicates as \emph{pre-/post-expectations}.
Intuitively, the weakest pre-expectation $\mathrm{wp}[\mathsf{prog}](\mathsf{post})$ is the expected value of the given post-expectation $\mathsf{post}$ with respect to the probability distribution of outputs of the program $\mathsf{prog}$.
For example, termination probability is a special case of the weakest pre-expectation: if the post-expectation is the constant $1$, then the weakest pre-expectation is the termination probability of a given program.

The weakest pre-expectation can be generalised to higher-order probabilistic programs.
In fact, we obtain the weakest pre-expectation transformer by just ignoring tick operators when applying the CPS transformation for expected cost analyses in Section~\ref{subsec:instance-expected-cost}.
The term obtained from the CPS transformation takes a post-expectation as a continuation and returns the value of the weakest pre-expectation.
We use the same $\lambdaHFL$-model as expected cost analyses (Section~\ref{subsec:instance-expected-cost}).
Specifically, we use $\interpret{\answertype} = ([0, \infty], {\le})$ as a set of quantitative truth values.

\newcommand{\recthree}{\mathrm{rec3}}
\begin{example}\label{ex:rec3}
	Consider the termination probability of the following program \cite{olmedo2016}.
	\[ \letrec{\recthree}{x}{\probbranch{()}{1/2}{(\recthree\ (); \recthree\ (); \recthree\ ())}}{\recthree\ ()} \]
	By applying the CPS transformation and passing $\lambda y. 1$ as a continuation, we obtain the following term in $\lambdaHFL$ whose interpretation is the termination probability $p \in \interpret{\answertype}$.
	\[ \letfix{\recthree}{x\ k}{1/2 \cdot k\ () + 1/2 \cdot \recthree\ ()\ (\lambda y. \recthree\ ()\ (\lambda y. \recthree\ ()\ k))}{\recthree\ ()\ (\lambda y. 1)} \]
	Here, we assume that we have as basic operations binary arithmetic operators $({+}), (\cdot) : \answertype \times \answertype \rightarrowtriangle \answertype$ and constants $1/2, 1 : \mathbf{unit} \rightarrowtriangle \answertype$.
	\qed
\end{example}

It is known that $(\sqrt{5} - 1) / 2$ gives an upper bound of the termination probability of Example~\ref{ex:rec3}.
Using refinement types, we can specify this as follows.
\[ \recthree : (x : \mathbf{unit}) \to (k : \mathbf{unit} \to \{ v : \answertype \mid r = 1 \}) \to \{ v : \answertype \mid r \le (\sqrt{5} - 1) / 2 \} \]

\subsection{Cost Moment}\label{subsec:instance-cost-moment}
The cost moment analysis \cite{kura2019} of a probabilistic program is a generalisation of the expected cost analysis.
In expected cost analyses, the aim is to estimate the expected value (i.e.\ the first moment $\mathbb{E}[C]$) of the cost $C$ of a probabilistic program.
On the other hand, the aim of cost moment analyses is to estimate higher moments $\mathbb{E}[C^n]$ of the cost $C$.

Even when we are interested only in the $n$-th moment $\mathbb{E}[C^n]$ of the cost, it is convenient to compute all the moments from the first moment to the $n$-th moment simultaneously $(\mathbb{E}[C], \dots, \mathbb{E}[C^n])$ as pointed out in \cite{kura2019}.
Thus, we define a $\lambdaHFL$-model for cost moment analyses by the category $\omegaQBS$ together with the following interpretation of $\answertype$.
\[ \interpret{\answertype} = ([0, \infty]^n, {\le}) \qquad\qquad \text{where ${\le}$ is the component-wise order.} \]
Intuitively, a value $(v_1, \dots, v_n) \in \interpret{\answertype}$ represents a tuple of moments of cost: $v_1$ is the first moment (i.e.\ the expected cost), and $v_2$ is the second moment, and so on.

A CPS transformation for cost moment analyses is given by a simple modification to expected cost analyses \cite{kura2023}.
The tick operator $({-})^{\checkmark}$ is mapped to the \emph{elapse function} $1 \oplus ({-}) : \answertype \rightarrowtriangle \answertype$ instead of $1 + ({-}) : \answertype \rightarrowtriangle \answertype$.
Intuitively, the elapse function $1 \oplus (\mathbb{E}[C], \dots, \mathbb{E}[C^n])$ computes the binomial expansion of $(\mathbb{E}[1 + C], \dots, \mathbb{E}[(1 + C)^n])$, which motivates the following definition: the $k$-th component of $\interpret{1 \oplus ({-})}(v_1, \dots, v_n)$ is defined by $1 + \sum_{i=1}^k \binom{k}{i} v_i$

\begin{example}\label{ex:coin-flip-moment}
	Consider estimating the second moment of the cost of the coin flipping in Example~\ref{ex:coin-flip}.
	We obtain the following term by the CPS transformation:
	\begin{equation}
		\letfix{\coinflip_2}{x\ k}{1/2 \cdot (1 \oplus \coinflip_2\ ()\ k) + 1/2 \cdot k\ ()}{\coinflip_2\ ()\ (\lambda r. (0, 0))}
		\label{eq:coin-flip-moment}
	\end{equation}
	where $1/2 \cdot ({-}) : \answertype \rightarrowtriangle \answertype$ is the component-wise multiplication.
	The interpretation of~\eqref{eq:coin-flip-moment} is the pair of the expected cost and the second moment of the cost.
	The meaning of the elapse function $1 \oplus ({-})$ becomes clearer if we (informally) decompose $\answertype$ to $\ExtNonnegRealType \times \ExtNonnegRealType$ as follows.
	\begin{align}
		&\mathbf{let}\ \mathbf{fix}\ \coinflip_2\ x\ k = \mathbf{let}\ (r_{11}, r_{12}) = \coinflip_2\ ()\ k\ \mathbf{in}\quad \mathbf{let}\ (r_{21}, r_{22}) = k\ ()\ \mathbf{in} \\
		& (1/2 \cdot (1 + r_{11}) + 1/2 \cdot r_{21},\quad 1/2 \cdot (1 + 2 \cdot r_{11} + r_{12}) + 1/2 \cdot r_{22})\ \mathbf{in} \quad \coinflip_2\ ()\ (\lambda r. (0, 0))
		\tag*{\qed}
	\end{align}
\end{example}

To verify that $3$ is an upper bound of the second moment of the cost of Example~\ref{ex:coin-flip-moment}, we can write the following specification using refinement types.
\[ \coinflip_2 : (x : \mathbf{unit}) \to (k : \mathbf{unit} \to \{ r : \answertype \mid r = (0, 0) \}) \to \{ r : \answertype \mid r \le (\infty, 3) \} \]

\subsection{Conditional Weakest Pre-Expectation}\label{subsec:instance-conditioning}
We consider reasoning about higher-order probabilistic programs with conditioning.
Conditioning is used in probabilistic programs to model posterior probabilities.
Here, we focus on \emph{hard conditioning} $\mathsf{observe}(b)$ for boolean values $b$ (as opposed to \emph{soft conditioning} $\mathsf{score}(r)$ for non-negative real weight $r$).
Intuitively, a probabilistic program with conditioning $\mathsf{observe}(b)$ is interpreted as a sub-probability distribution from which all runs that violate the condition $b$ are excluded.
To reason about probabilistic programs with conditioning, the conditional weakest pre-expectation transformer is introduced in \cite{olmedo2018} for the probabilistic guarded command language (pGCL) with conditioning.
The conditional weakest pre-expectation is defined by the weakest pre-expectation normalised (divided) by the probability of all valid runs.

We can extend conditional weakest pre-expectations to higher-order programs by a CPS transformation as follows.
First note that without loss of generality, we can assume that the form of conditioning in a probabilistic program is either $\mathsf{observe}(\mathbf{true})$ or $\mathsf{observe}(\mathbf{false})$ because we can rewrite $\mathsf{observe}(P)$ to $\ifexpr{P}{\mathsf{observe}(\mathbf{true})}{\mathsf{observe}(\mathbf{false})}$.
Then, the CPS transformation for conditional weakest pre-expectation maps $\mathsf{observe}(\mathbf{true})$ and $\mathsf{observe}(\mathbf{false})$ to $1 \cdot ({-})$ and $0 \cdot ({-})$, respectively.
We will show an example in Example~\ref{ex:three-coin}.

As a $\lambdaHFL$-model, we use $\omegaQBS$ together with the following interpretation of $\answertype$.
\[ \interpret{\answertype} = [0, \infty] \times [0, 1]^{\op} = ([0, \infty], {\le}) \times ([0, 1], {\ge}) \]
The definition of ``truth values'' above follows \cite{olmedo2018}.
When computing conditional weakest pre-expectation, it is convenient to separately compute unnormalised weakest pre-expectations for valid runs and the probability of all valid runs (i.e.\ runs such that $b$ in $\mathsf{observe}(b)$ is always evaluated as true).
The first component $[0, \infty]$ represents the unnormalised weakest pre-expectation for valid runs, and the second component $[0, 1]^{\op}$ represents the probability of all valid runs.
We use the opposite order for the second component $[0, 1]^{\op}$ because the probability of all valid runs is defined using greatest fixed points.
Thus, the conditional weakest pre-expectation is given by normalisation $v_1 / v_2$ where $(v_1, v_2) \in \interpret{\answertype}$.

\newcommand{\threecoin}{\mathrm{tc}}
\begin{example}\label{ex:three-coin}
	Consider the conditional weakest pre-expectation of the following program \cite[Example~4.4]{olmedo2018}, which flips three coins repeatedly.
	\begin{align}
		&\mathbf{let}\ \mathrm{prob}_{1/2}\ x = \probbranch{\mathbf{true}}{1/2}{\mathbf{false}}\ \mathbf{in} \\
		&\mathbf{let}\ \mathbf{rec}\ \threecoin\ m = \mathbf{let}\ b_1 = \mathrm{prob}_{1/2}\ ()\ \mathbf{in}\ \mathbf{let}\ b_2 = \mathrm{prob}_{1/2}\ ()\ \mathbf{in}\ \mathbf{let}\ b_3 = \mathrm{prob}_{1/2}\ ()\ \mathbf{in} \\
		&\qquad\qquad \mathbf{observe}(\lnot b_1 \lor \lnot b_2 \lor \lnot b_3);\quad \ifexpr{b_1 \lor b_2 \lor b_3}{\threecoin\ (m + 1)}{m + 1}\ \mathbf{in}
		\qquad \threecoin\ 0
	\end{align}
	By applying a CPS transformation and simplifying the result, we obtain the following term of $\lambdaHFL$.
	Here, we use $\lambda m. [m = N]$ as a post-expectation where $N$ is an integer and $[m = N] \coloneqq \ifexpr{m = N}{1}{0}$ is the Iverson bracket.
	\[ M_{\mathrm{cwp}} \quad\coloneqq\quad \letfix{\threecoin}{m\ k}{1/8 \cdot 0 + 6/8 \cdot \threecoin\ (m + 1)\ k + 1/8 \cdot k\ (m + 1)}{\threecoin\ 0\ (\lambda m. ([m = N], 1))} \]
	The conditional weakest pre-expectation is obtained as $\interpret{M_{\mathrm{cwp}}}_1 / \interpret{M_{\mathrm{cwp}}}_2$ where $\interpret{M_{\mathrm{cwp}}}_i$ is the $i$-th component of $\interpret{M_{\mathrm{cwp}}}$.
	\qed
\end{example}

Using refinement types, we can give the following specification for Example~\ref{ex:three-coin}.
\begin{align}
	\threecoin : (m : \mathbf{int}) &\to (k : (m' : \mathbf{int}) \to \{ r : \answertype \mid r = ([m' = N], 1) \}) \\
	&\to \{ (r_1, r_2) : \answertype \mid r_1 \le \ifexpr{m \ge N}{0}{1/6 \cdot (3/4)^{N - m}} \land r_2 \ge 1/2 \}
\end{align}
Note that this specification implies that the conditional weakest pre-expectation is upper-bounded by $r_1 / r_2 \le 1/3 \cdot (3/4)^{N - m}$.


\section{Dependent Refinement Type System}\label{sec:refinement}
In this section, we propose a novel refinement type system for $\lambdaHFL$.

\subsection{Formulas}
\subsubsection{Syntax}

Let $\AtomicPreds$ be a set of atomic predicates and $\mathrm{ar}(\mathbf{ap})$ be a type, called an \emph{arity}, for each $\mathbf{ap} \in \AtomicPreds$.
We assume that the arity $\mathrm{ar}(\mathbf{ap})$ is a finite product of base types or $\answertype$, that is, there exists $n$ and $\sigma_1, \dots \sigma_n \in \BaseTypes \cup \{ \answertype \}$ such that $\mathrm{ar}(\mathbf{ap}) = \sigma_1 \times \dots \times \sigma_n$.
We often write $\mathbf{ap} : \sigma$ to indicate that the arity of $\mathbf{ap}$ is $\sigma$.
We assume that $\AtomicPreds$ always contains the order relation $({\le}_{\answertype}) : \answertype \times \answertype$ for $\answertype$ and the equality relation $({=}_{b}) : b \times b$ for each base type $b \in \BaseTypes$.

\emph{Formulas} are defined by the following syntax.
\[ \psi, \phi \coloneqq \mathbf{ap}(M) \mid \top \mid \bot \mid \psi \implies \phi \mid \psi \land \phi \mid \psi \lor \phi \qquad \text{where $\mathbf{ap} \in \AtomicPreds$} \]

We define a \emph{well-formed formula} $\Gamma \vdash \phi$ inductively: for atomic predicates, $\Gamma \vdash \mathbf{ap}(M)$ is well-formed if $\Gamma \vdash M : \mathrm{ar}(\mathbf{ap})$ is well-typed, and for other cases, $\Gamma \vdash \phi$ is well-formed if each sub-formula is well-formed.
For example, if $\Gamma \vdash M : b$ and $\Gamma \vdash N : b$ are well-typed term, then $\Gamma \vdash ({=}_{b})\ (M, N)$ is a well-formed formula.
For better readability, we often use infix notation and omit type annotations, like $M = N$ instead of $({=}_{b})\ (M, N)$ and $M \le N$ instead of $({\le}_{\answertype})\ (M, N)$.

\subsubsection{Semantics}
Suppose that a $\lambdaHFL$-model and an interpretation of atomic predicates in $\AtomicPreds$ are given: for any $\mathbf{ap} : \sigma$, let $\interpret{\mathbf{ap}} \subseteq \interpret{\sigma}$ be a subset that gives an interpretation of $\mathbf{ap}$.
We assume that we always interpret ${\le}_{\answertype}$ by $\interpret{{\le}_{\answertype}} \coloneqq \{ (x, y) \in \interpret{\answertype} \times \interpret{\answertype} \mid x \le_{\interpret{\answertype}} y \} \subseteq \interpret{\answertype \times \answertype}$ and ${=}_b$ by $\interpret{{=}_b} \coloneqq \{ (x, x) \in \interpret{b} \times \interpret{b} \mid x \in \interpret{b} \} \subseteq \interpret{b \times b}$.
A well-formed formula $\Gamma \vdash \phi$ is interpreted as the subset $\interpret{\phi} \subseteq \interpret{\Gamma}$ defined below.
We say $\gamma \in \interpret{\Gamma}$ \emph{satisfies} $\phi$ if $\gamma \in \interpret{\phi}$.
\begin{align}
	\interpret{\top} &\coloneqq \interpret{\Gamma} &
	\interpret{\psi \land \phi} &\coloneqq \interpret{\psi} \cap \interpret{\phi} &
	\interpret{\mathbf{ar}(M)} &\coloneqq \{ \gamma \in \interpret{\Gamma} \mid \interpret{M}(\gamma) \in \interpret{\mathbf{ap}} \} \\
	\interpret{\bot} &\coloneqq \emptyset &
	\interpret{\psi \lor \phi} &\coloneqq \interpret{\psi} \cup \interpret{\phi} &
	\interpret{\psi \implies \phi} &\coloneqq \{ \gamma \in \interpret{\Gamma} \mid \gamma \in \interpret{\psi} \implies \gamma \in \interpret{\phi} \}
\end{align}

\subsection{Refinement Types}

\subsubsection{Syntax}
\emph{Refinement types} $\dot{\sigma}, \dot{\tau}$ are defined as follows.
\begin{align}
	\dot{\sigma}, \dot{\tau} \coloneqq \{ v : b \mid \phi \} \mid \{ v : \answertype \mid \phi \} \mid \{ v : 1 \mid \phi \} \mid (x : \dot{\sigma}) \times \dot{\tau} \mid 0 \mid \dot{\sigma} + \dot{\tau} \mid (x : \dot{\sigma}) \to \{ v : \answertype \mid \phi \}
\end{align}
The \emph{underlying type} $\underlying{\dot{\sigma}}$ of a refinement type $\dot{\sigma}$ is a simple type obtained by erasing formulas.
For example, we have $\underlying{\{ v : \sigma \mid \phi \}} = \sigma$ and $\underlying{(x : \dot{\sigma}) \to \dot{\tau}} = \underlying{\dot{\sigma}} \to \underlying{\dot{\tau}}$.
The \emph{underlying context} $\underlying{\dot{\Gamma}}$ is defined by $\underlying{x_1 : \dot{\sigma}_1, \dots, x_n : \dot{\sigma}_n} \coloneqq x_1 : \underlying{\dot{\sigma}_1}, \dots, x_n : \underlying{\dot{\sigma}_n}$.
For any $\sigma \in \BaseTypes \cup \{ 1, \answertype \}$, we often write $\sigma \coloneqq \{ v : \sigma \mid \top \}$.
For any $\sigma_1, \sigma_2 \in \BaseTypes \cup \{ 1, \answertype \}$, we define $\{ (v_1, v_2) : \sigma_1 \times \sigma_2 \mid \phi \} \coloneqq (v_1 : \sigma_1) \times \{ v_2 : \sigma_2 \mid \phi \}$.
For $n > 2$, we define $\{ (v_1, \dots, v_n) : \sigma_1 \times \dots \times \sigma_n \mid \phi \}$ in the same way.

\emph{Refinement contexts} $\dot{\Gamma}$ are defined by a list of pairs of a variable and a refinement type: $\dot{\Gamma} \coloneqq \cdot \mid \dot{\Gamma}, x : \dot{\sigma}$.
We assume variables in a context are mutually distinct.
We define \emph{well-formed contexts} $\vdash \dot{\Gamma}$ and \emph{well-formed types} $\dot{\Gamma} \vdash \dot{\sigma}$ inductively so that any formula in $\dot{\Gamma}$ and $\dot{\sigma}$ is well-formed (\referappendix{sec:detail-refinement}{B}).
Note that we write a dot on top of each meta-variable for types $\sigma, \tau$ and contexts $\Gamma$ to distinguish refinement types from simple types in Section~\ref{sec:hfl}.

\subsubsection{Semantics}

Intuitively, a refinement type represents a subset of its underlying type.
Taking this into account, we define the interpretation of refinement contexts and refinement types as follows.
For each well-formed context $\vdash \dot{\Gamma}$, the interpretation $\interpret{\dot{\Gamma}}$ is given by a pair $(X, P)$ where $X = \interpret{\underlying{\dot{\Gamma}}}$ is the interpretation of the underlying context of $\dot{\Gamma}$, and $P \subseteq X$ is a subset.
Intuitively, $P$ is the set of all $\gamma \in \interpret{\underlying{\dot{\Gamma}}}$ that satisfy all formulas in $\dot{\Gamma}$.
For each well-formed type $\dot{\Gamma} \vdash \dot{\sigma}$, the interpretation $\interpret{\dot{\Gamma} \vdash \dot{\sigma}}$ is given by a quadruple $(X, Y, P, Q)$ where $(X, P) = \interpret{\dot{\Gamma}}$ is the interpretation of $\dot{\Gamma}$, $Y = \interpret{\underlying{\dot{\sigma}}}$ is the interpretation of the underlying type of $\dot{\sigma}$, and $Q \subseteq X \times Y$ is a subset such that $(x, y) \in Q$ implies $x \in P$.
Examples of the interpretation are given below.
\begin{align}
	\intertext{Let \hspace{1em} $\dot{\Gamma}_1 \ \coloneqq\ x : \{ x : \mathbf{int} \mid 0 \le x \}$ \hspace{1em} and \hspace{1em} $\dot{\Gamma}_2 \ \coloneqq\ x : \{ x : \mathbf{int} \mid 0 \le x \}, y : \{ y : \mathbf{int} \mid x \le y \}$.}
	&\interpret{\dot{\Gamma}_1} &&=\quad (\mathbb{Z},\quad \{ x \in \mathbb{Z} \mid 0 \le x \}) \\
	&\interpret{\dot{\Gamma}_1 \vdash \{ y : \mathbf{int} \mid x \le y \}} &&=\quad (\mathbb{Z},\quad \mathbb{Z},\quad \{ x \in \mathbb{Z} \mid 0 \le x \},\quad \{ (x, y) \in \mathbb{Z} \times \mathbb{Z} \mid 0 \le x \land x \le y \}) \\
	&\interpret{\dot{\Gamma}_2} &&=\quad (\mathbb{Z} \times \mathbb{Z},\quad \{ (x, y) \in \mathbb{Z} \times \mathbb{Z} \mid 0 \le x \land x \le y \})
\end{align}
We often write $\interpret{-}_i$ for the $i$-th component of $\interpret{-}$.
Technically, this interpretation is based on a categorical construction of models of dependent refinement type systems~\cite{kura2021} (see \referappendixnocite{sec:detail-refinement}{B} for details).
Note that if $(X, Y, P, Q) = \interpret{\dot{\Gamma} \vdash \dot{\sigma}}$, then we have $\interpret{\dot{\Gamma}, x : \dot{\sigma}} = (X \times Y, Q)$.

Using our denotational semantics, a (semantic) subtyping relation can be defined as follows.
A type $\dot{\Gamma} \vdash \dot{\sigma}$ is a \emph{semantic subtype} of $\dot{\Gamma} \vdash \dot{\tau}$ if $\interpret{\dot{\Gamma}, x : \dot{\sigma}}_1 = \interpret{\dot{\Gamma}, x : \dot{\tau}}_1$ and $\interpret{\dot{\Gamma}, x : \dot{\sigma}}_2 \subseteq \interpret{\dot{\Gamma}, x : \dot{\tau}}_2$.

\subsection{Typing Rules}\label{sec:refinement-typing-rule}
\begin{figure}
	\begin{mathpar}
		\small
		\inferrule[R-App]{
			\dot{\Gamma} \vdash M : (x : \dot{\sigma}) \to \{ v : \answertype \mid \phi \} \\
			\dot{\Gamma} \vdash N : \dot{\sigma}
		}{
			\dot{\Gamma} \vdash M\ N : \{ v : \answertype \mid \phi[N/x] \}
		}
		\and
		\inferrule[R-Abs]{
			\dot{\Gamma}, x : \dot{\sigma} \vdash M : \{ v : \answertype \mid \phi \}
		}{
			\dot{\Gamma} \vdash \lambda x : \dot{\sigma}. M : (x : \dot{\sigma}) \to \{ v : \answertype \mid \phi \}
		}
		\and
		\inferrule[R-Unit]{
			\vdash \dot{\Gamma}
		}{
			\dot{\Gamma} \vdash () : \{ v : 1 \mid \top \}
		}
		\and
		\inferrule[R-Pair]{
			\dot{\Gamma} \vdash M : \dot{\sigma} \\
			\dot{\Gamma} \vdash N : \dot{\tau}[M / x]
		}{
			\dot{\Gamma} \vdash (M, N) : (x : \dot{\sigma}) \times \dot{\tau}
		}
		\and
		\inferrule[R-Fst]{
			\dot{\Gamma} \vdash M : (x : \dot{\sigma}) \times \dot{\tau}
		}{
			\dot{\Gamma} \vdash \pi_1\ M : \dot{\sigma}
		}
		\and
		\inferrule[R-Snd]{
			\dot{\Gamma} \vdash M : (x : \dot{\sigma}) \times \dot{\tau}
		}{
			\dot{\Gamma} \vdash \pi_2\ M : \dot{\tau}[\pi_1\ M / x]
		}
		\and
		\inferrule[R-Case0]{
			\dot{\Gamma} \vdash M : 0 \\
			\dot{\Gamma} \vdash \dot{\tau}
		}{
			\dot{\Gamma} \vdash \delta(M) : \dot{\tau}
		}
		\and
		\inferrule[R-Inj]{
			i \in \{ 1, 2 \} \\
			\dot{\Gamma} \vdash M : \dot{\sigma}_i \\
			\dot{\Gamma} \vdash \dot{\sigma}_{3-i}
		}{
			\dot{\Gamma} \vdash \iota_i\ M : \dot{\sigma}_1 + \dot{\sigma}_2
		}
		\and
		\inferrule[R-Case2]{
			\dot{\Gamma}, z : \dot{\sigma}_1 + \dot{\sigma}_2 \vdash \dot{\tau} \\
			\dot{\Gamma} \vdash M : \dot{\sigma}_1 + \dot{\sigma}_2 \\
			\dot{\Gamma}, x_1 : \dot{\sigma}_1 \vdash N_1 : \dot{\tau}[\iota_1\ x_1/z] \\
			\dot{\Gamma}, x_2 : \dot{\sigma}_2 \vdash N_2 : \dot{\tau}[\iota_2\ x_2/z]
		}{
			\dot{\Gamma} \vdash \caseexpr{M}{x_1}{N_1}{x_2}{N_2} : \dot{\tau}[M/z]
		}
		\and
		\inferrule[R-VarRefine]{
			\vdash \dot{\Gamma} \\
			(x : \{ v : b \mid \phi \}) \in \dot{\Gamma}
		}{
			\dot{\Gamma} \vdash x : \{ v : b \mid v = x \}
		}
		\and
		\inferrule[R-Var]{
			\vdash \dot{\Gamma} \\
			(x : \dot{\sigma}) \in \dot{\Gamma}
		}{
			\dot{\Gamma} \vdash x : \dot{\sigma}
		}
		\and
		\inferrule[R-Sub]{
			\dot{\Gamma} \vdash M : \dot{\sigma} \\
			\dot{\Gamma} \vdash \dot{\sigma} <: \dot{\tau}
		}{
			\dot{\Gamma} \vdash M : \dot{\tau}
		}
		\and
		\inferrule[R-Fix]{
			\dot{\Gamma}, f : (x : \dot{\sigma}) \to \{ v : \answertype \mid \phi \} \vdash M : (x : \dot{\sigma}) \to \{ v : \answertype \mid \phi \} \\
			\underlying{\dot{\Gamma}}, x : \underlying{\dot{\sigma}}, v : \answertype \vdash \phi \\
			\text{$\phi$ is admissible at $v$}
		}{
			\dot{\Gamma} \vdash \fixpoint{f}{M} : (x : \dot{\sigma}) \to \{ v : \answertype \mid \phi \}
		}
		\and
		\inferrule[R-BasicOp]{
			\dot{\Gamma} \vdash M : \dot{\sigma} \\
			\dot{\Gamma} \vdash \dot{\tau} \\
			\mathrm{ar}(\mathbf{op}) = \underlying{\dot{\sigma}} \\
			\mathrm{car}(\mathbf{op}) = \underlying{\dot{\tau}} \\
			\forall \gamma \in \interpret{\dot{\Gamma}, v : \dot{\sigma}}_2, \gamma[v \mapsto a(\mathbf{op})(\gamma(v))] \in \interpret{\dot{\Gamma}, v : \dot{\tau}}_2
		}{
			\dot{\Gamma} \vdash \mathbf{op}(M) : \dot{\tau}
		}
	\end{mathpar}
	\caption{Typing rules for $\dot{\Gamma} \vdash M : \dot{\sigma}$.}
	\label{fig:typing-refinement}
\end{figure}

Fig.~\ref{fig:typing-refinement} shows typing rules for dependent refinement types.
Here, a judgement for well-typed-ness has the form $\dot{\Gamma} \vdash M : \dot{\sigma}$.
We can use ordinary typing rules for dependent refinement types in most cases except for \textsc{R-Fix} and \textsc{R-BasicOp}.
\textsc{R-Fix} and \textsc{R-BasicOp} contain semantic conditions in their premises, which will be explained in Section~\ref{sec:typing-fixed-point} and Section~\ref{sec:typing-basic-operators}, respectively.
In \textsc{R-Sub}, we use a subtyping relation $\dot{\Gamma} \vdash \dot{\sigma} <: \dot{\tau}$.
Since derivation rules for the subtyping relation are standard, they are omitted here (see \referappendixnocite{sec:detail-refinement}{B} for details).
The soundness of the typing rules is stated as follows.

\begin{theorem}[soundness]\label{thm:soundness}
	For any $\lambdaHFL$-model, typing rules in Fig.~\ref{fig:typing-refinement} are sound: if $\dot{\Gamma} \vdash M : \dot{\sigma}$ is well-typed, then for any $\gamma \in \interpret{\dot{\Gamma}}_2$, we have $(\gamma, \interpret{M}(\gamma)) \in \interpret{\dot{\Gamma}, r : \dot{\sigma}}_2$ where $\interpret{\dot{\Gamma}}_2$ is the second component of $\interpret{\dot{\Gamma}}$.
\end{theorem}
\begin{proof}
	By induction of derivation of $\dot{\Gamma} \vdash M : \dot{\sigma}$.
	We use a general categorical construction of \cite{kura2021}.
	For the case of \textsc{R-Fix}, see Lemma~\ref{lem:r-fix-sound}.
	Other cases are proved in \referappendixnocite{sec:detail-proof}{C}.
\end{proof}

The following corollary states the soundness in a more intuitive way.
\begin{corollary}\label{cor:soundness}
	For any $\lambdaHFL$-model, if $\vdash M : \{ v : \sigma \mid \phi \}$ is well-typed, then $\vdash \phi[M/v]$ is true.
	\qed
\end{corollary}

Since Corollary~\ref{cor:soundness} hold for any $\lambdaHFL$-model, we obtain soundness for each instance in Section~\ref{sec:instance}.
\begin{corollary}[soundness of weakest pre-expectations and expected costs]
	Consider the $\lambdaHFL$-model for weakest pre-expectations and expected costs.
	For any $u \in \interpret{\answertype} = ([0, \infty], {\le})$,
	\begin{equation}
		\text{if} \quad \vdash M : \{ v : \answertype \mid v \le u \} \quad \text{is well-typed, then} \quad \interpret{M}(\emptyenv) \le u.
		\tag*{\qed}
	\end{equation}
\end{corollary}

\begin{corollary}[soundness of cost moment analyses]
	Consider the $\lambdaHFL$-model for cost moment analyses.
	For any $(u_1, \dots, u_n) \in \interpret{\answertype} = ([0, \infty]^n, {\le})$,
	\begin{equation}
		\text{if} \quad \vdash M : \{ v : \answertype \mid v \le (u_1, \dots, u_n) \} \quad \text{is well-typed, then} \quad \interpret{M}(\emptyenv) \le (u_1, \dots, u_n).
		\tag*{\qed}
	\end{equation}
\end{corollary}

\begin{corollary}[soundness of the conditional weakest pre-expectations]
	Consider the $\lambdaHFL$-model for conditional weakest pre-expectations.
	For any $(u_1, u_2) \in [0, \infty] \times [0, 1]^{\op}$,
	\[ \text{if} \quad \vdash M : \{ v : \answertype \mid v \le_{\answertype} (u_1, u_2) \} \quad \text{is well-typed, then} \quad \interpret{M}_1(\emptyenv) \le u_1 \land \interpret{M}_2(\emptyenv) \ge u_2 \]
	where $\interpret{M}_i(\emptyenv)$ is the $i$-th component of $\interpret{M}(\emptyenv) \in \interpret{\answertype} = [0, \infty] \times [0, 1]^{\op}$.
	\qed
\end{corollary}

\subsubsection{Unsoundness of the Standard Rule for Recursion}
\begin{figure}
	\begin{mathpar}
		\inferrule[Fix-Unsound]{
			\Gamma, f : (x : \sigma) \to \{ v : \ExtNonnegRealType \mid \phi \} \vdash M : (x : \sigma) \to \{ v : \ExtNonnegRealType \mid \phi \}
		}{
			\Gamma \vdash \fixpoint{f}{M} : (x : \sigma) \to \{ v : \ExtNonnegRealType \mid \phi \}
		}
	\end{mathpar}
	\caption{The standard typing rule for (the partial correctness) of recursion is unsound for fixed points.}
	\label{fig:unsound-typing-fixed-point}
\end{figure}

The standard typing rule for recursion (\textsc{Fix-Unsound} in Fig.~\ref{fig:unsound-typing-fixed-point}) is unsound for fixed points. 
To give a counterexample, consider the term~\eqref{eq:coin-flip-cpsed} in Example~\ref{ex:coin-flip}.
We consider the following type where $\phi(r)$ is a predicate on the expected cost $r$.
\begin{equation}
	\coinflip : (x : \mathbf{unit}) \to (k : \mathbf{unit} \to \{ r : \ExtNonnegRealType \mid r = 0 \}) \to \{ r : \ExtNonnegRealType \mid \phi(r) \}
	\label{eq:coin-flip-refinemnt-type-general}
\end{equation}
If we use \textsc{Fix-Unsound}, we can derive $\phi(r) = r < 1$ and $\phi(r) = \bot$ since $\phi(r)$ implies $\phi(1/2 \cdot (1 + r) + 1/2 \cdot 0)$ in both cases.
However, both of them are incorrect because the true expected cost $r = 1$ doesn't satisfy neither $r < 1$ nor $\bot$.

Technically, the unsoundness of \textsc{Fix-Unsound} is because the semantics of $\mathbf{fix}$ is different from usual semantics of recursion: $\mathbf{fix}$ is defined by the least fixed point with respect to the standard order on $[0, \infty]$ whereas usual recursion is the least fixed point with respect to $X_{\bot} = X \cup \{ \bot \}$ for some $X$ where $\bot$ represents divergence and the partial order on $X_{\bot}$ is defined by (1) $\bot$ is the least element and (2) restriction of the partial order on $X$ is the discrete order.
Therefore, the typing rule for the latter semantics is unsound for the former.

\subsubsection{A Sound Typing Rule for Fixed Points}\label{sec:typing-fixed-point}
To fix the issue of \textsc{Fix-Unsound}, we take into account \emph{admissible subsets} (\emph{admissible predicates}) of $\omega$cpos and consider the typing rule \textsc{R-Fix}.
This approach is an adaptation of EHOL~\cite{avanzini2021} to a more automatable method, namely, a dependent refinement type system.
Furthermore, we prove that \textsc{R-Fix} is sound for any $\lambdaHFL$-models.
Thus, our dependent refinement type system is applicable to all verification problems listed in Section~\ref{sec:instance}.

Admissible subsets have a closure property with respect to least fixed points.
Recall that the least fixed point $\mathrm{lfp}\ f$ in $\omega$cpos is given by the supremum of an $\omega$-chain: assuming that an $\omega$cpo $X$ has a least element $\bot_X \in X$, we have $\mathrm{lfp}\ f = \sup_n f^n(\bot_X)$ for any Scott-continuous function $f : X \to X$.
Suppose that $S \subseteq X$ is a subset such that $f : X \to X$ maps elements in $S$ to elements in $S$.
Even in this situation, $\mathrm{lfp}\ f \in S$ is not necessarily true in general.
To ensure $\mathrm{lfp}\ f \in S$, we need the notion of admissible subsets.
\begin{definition}[admissible subset]
	Let $X$ be an $\omega$cpo with a least element $\bot_X \in X$.
	A subset $S \subseteq X$ is \emph{admissible} if (i) $\bot_X \in S$ and (ii) $S$ is closed under supremum of $\omega$-chains.
\end{definition}

\begin{example}\label{ex:admissible-in-ext-nonneg-real}
	Consider the $\omega$cpo $([0, \infty], {\le})$.
	For any $a \in [0, \infty]$, $\{ x \in [0, \infty] \mid x \le a \} \subseteq [0, \infty]$ is admissible.
	On the other hand, $\{ x \in [0, \infty] \mid x < a \} \subseteq [0, \infty]$ is \emph{not} admissible because this subset is not closed under supremum of $\omega$-chains.
	For any $a > 0$, $\{ x \in [0, \infty] \mid a \le x \} \subseteq [0, \infty]$ is \emph{not} admissible because the least element $0 \in [0, \infty]$ is not in this subset.
	\qed
\end{example}
Example~\ref{ex:admissible-in-ext-nonneg-real} suggests why $\phi(r) = r < 1$ should \emph{not} be derived for Example~\ref{ex:coin-flip} and~\eqref{eq:coin-flip-refinemnt-type-general}.
Recall that the interpretation of the fixed point in~\eqref{eq:coin-flip-cpsed} is given by the supremum of the $\omega$-chain $0 \le 1/2 \le 3/4 \le \dots \to 1$.
Although each element of the $\omega$-chain satisfies $\phi(r) = r < 1$, their supremum doesn't because $\phi(r) = r < 1$ is not admissible.

\begin{example}\label{ex:admissible-lifting-discrete}
	Consider the $\omega$cpo $X_{\bot} = (X \cup \{ \bot \}, {\le})$ where we have $x \le y$ if and only if $x = \bot$.
	For any subset $S \subseteq X$, $S \cup \{ \bot \} \subseteq X_{\bot}$ is admissible.
	\qed
\end{example}
Example~\ref{ex:admissible-lifting-discrete} is why the standard rule for partial correctness (\text{Fix-Unsound}) of recursion does not impose admissibility explicitly.
The semantics of recursion is usually defined using the $\omega$cpo $X_{\bot}$, in which case, the admissibility of predicates is guaranteed almost for free.

Recall that a formula $\Gamma \vdash \phi$ is interpreted as a subset of $\interpret{\Gamma} = \prod_i \interpret{\sigma_i}$.
For each variable $(x_i : \sigma_i) \in \Gamma$, $\interpret{\phi} \subseteq \interpret{\Gamma}$ induces a subset of $\interpret{\sigma_i}$ by fixing values assigned to other variables in $\Gamma$.
We define the admissibility of formulas by considering the admissibility of such subsets.
\begin{definition}[admissible formula]
	A well-formed formula $\Gamma \vdash \phi$ is \emph{admissible at variable $x$} if $(x : \answertype) \in \Gamma$ and for any $\gamma \in \interpret{\Gamma}$, $\{ v \in \interpret{\answertype} \mid \gamma[x \mapsto v] \in \interpret{\Gamma} \} \subseteq \interpret{\answertype}$ is admissible.
	Here, $\gamma[x \mapsto v]$ is the mapping such that $\gamma[x \mapsto v](x) = v$ and $\gamma[x \mapsto v](y) = \gamma(y)$ if $x \neq y$.
\end{definition}

\begin{lemma}\label{lem:r-fix-sound}
	\textsc{R-Fix} is sound in the following sense.
	Let $\dot{\tau} = (x : \dot{\sigma}) \to \{ v : \answertype \mid \phi \}$.
	If $M$ satisfies $(\gamma, \interpret{M}(\gamma)) \in \interpret{\dot{\Gamma}, f : \dot{\tau}, f' : \dot{\tau}}_2$ for any $\gamma \in \interpret{\dot{\Gamma}, f : \dot{\tau}}_2$ and $\underlying{\dot{\Gamma}}, x : \underlying{\dot{\sigma}}, v : \answertype \vdash \phi$ is admissible at $v$, then we have $(\gamma, \interpret{\fixpoint{f}{M}}(\gamma)) \in \interpret{\dot{\Gamma}, f : \dot{\tau}}_2$ for any $\gamma \in \interpret{\dot{\Gamma}}_2$.
\end{lemma}
\begin{proof}
	This is a part of the induction in Theorem~\ref{thm:soundness}.
	We sketch the proof (see \referappendixnocite{sec:detail-proof}{C} for details).
	Since $\underlying{\dot{\Gamma}}, x : \underlying{\dot{\sigma}}, v : \answertype \vdash \phi$ is admissible at $v$, we can prove that the forth component of $\interpret{\dot{\Gamma} \vdash (x : \dot{\sigma}) \to \{ v : \answertype \mid \phi \}}$ is also ``admissible'', which precisely means that for any $\gamma \in \interpret{\dot{\Gamma}}_2$,
	\[ \{ v : \interpret{\underlying{\dot{\sigma}}} \to \interpret{\answertype} \mid \gamma[f \mapsto v] \in \interpret{\dot{\Gamma} \vdash (x : \dot{\sigma}) \to \{ v : \answertype \mid \phi \}}_4 \} \subseteq \exponential{\interpret{\underlying{\dot{\sigma}}}}{\interpret{\answertype}} \]
	is admissible where $\interpret{-}_i$ is the $i$-th component of $\interpret{-}$.
	By definition of the interpretation, we have
	$\interpret{\fixpoint{f}{M}}(\gamma) = \sup_n F^n(\bot)$
	where $F \coloneqq \interpret{M}(\gamma[f \mapsto ({-})])$.
	By the induction hypothesis and by the admissibility of $\interpret{\dot{\Gamma} \vdash (x : \dot{\sigma}) \to \{ v : \answertype \mid \phi \}}$, we conclude $(\gamma, \interpret{\fixpoint{f}{M}}(\gamma)) \in \interpret{\dot{\Gamma}, f : (x : \dot{\sigma}) \to \{ v : \answertype \mid \phi \}}_2$ for any $\gamma \in \interpret{\dot{\Gamma}}_2$.
\end{proof}

To implement $\textsc{R-Fix}$ in a type checker, we need to check the admissibility of formulas \emph{syntactically}, which will be discussed later in Section~\ref{subsec:admissible}.

\subsubsection{Typing Rules for Basic Operators}\label{sec:typing-basic-operators}
\textsc{R-BasicOp} gives a general typing rule for basic operators.
One of its premises is a semantic condition $\forall \gamma \in \interpret{\dot{\Gamma}, v : \dot{\sigma}}_2, \gamma[v \mapsto a(\mathbf{op})(\gamma(v))] \in \interpret{\dot{\Gamma}, v : \dot{\tau}}_2$, which intuitively means that if an argument of $a(\mathbf{op})$ satisfies the precondition in $\dot{\sigma}$, then $a(\mathbf{op})$ returns a value that satisfies the postcondition in $\dot{\tau}$.
Although \textsc{R-BasicOp} is a natural and general rule, it is usually convenient to specialise it so that we can avoid the semantic condition in an implementation of a type checker.
Fortunately, rather common typing rules shown in Fig.~\ref{fig:basic-operator-typing} can be proved sound in any $\lambdaHFL$-model, and they cover most of basic operators that we need for verifying probabilistic programs.

\begin{figure}
	\begin{mathpar}
		\small
		\inferrule[R-BasicConst]{
			\mathbf{op} : 1 \to \tau \\
			\tau \in \BaseTypes \cup \{ \answertype \}
		}{
			\dot{\Gamma} \vdash \mathbf{op} : \{ v : \tau \mid v =_{\tau} \mathbf{op} \}
		}
		\and
		\inferrule[R-BasicSimp]{
			\mathbf{op} : \sigma_1 \times \dots \times \sigma_n \to \tau \\
			\sigma_1, \dots, \sigma_n, \tau \in \BaseTypes \cup \{ \answertype \} \\\\
			\dot{\Gamma} \vdash M : \{ (v_1, \dots, v_n) : \sigma_1 \times \dots \times \sigma_n \mid \phi[\mathbf{op}(v_1, \dots, v_n) / v] \}
		}{
			\dot{\Gamma} \vdash \mathbf{op}(M) : \{ v : \tau \mid \phi \}
		}
		\and
		\inferrule[R-BasicBool]{
			\mathbf{op} : \sigma_1 \times \dots \times \sigma_n \to 1 + 1 \\
			\sigma_1, \dots, \sigma_n \in \BaseTypes \cup \{ \answertype \} \\
			v_1 : \sigma_1, \dots, v_n : \sigma_n \vdash \psi_t \\
			\forall \gamma \in \interpret{(v_1, \dots, v_n) : \sigma_1 \times \dots \times \sigma_n \vdash \psi_t}, a(\mathbf{op})(\gamma) = \iota_1\ () \\
			v_1 : \sigma_1, \dots, v_n : \sigma_n \vdash \psi_f \\
			\forall \gamma \in \interpret{(v_1, \dots, v_n) : \sigma_1 \times \dots \times \sigma_n \vdash \psi_f}, a(\mathbf{op})(\gamma) = \iota_2\ () \\
			\dot{\Gamma} \vdash M : \{ (v_1, \dots, v_n) : \sigma_1 \times \dots \times \sigma_n \mid \psi_t \land \phi_t[()/v] \lor \psi_f \land \phi_f[()/v] \}
		}{
			\dot{\Gamma} \vdash \mathbf{op}(M) : \{ v : 1 \mid \phi_t \} + \{ v : 1 \mid \phi_f \}
		}
	\end{mathpar}
	\caption{Three common patterns of typing rules for basic operators.}
	\label{fig:basic-operator-typing}
\end{figure}

In Fig.~\ref{fig:basic-operator-typing}, the first rule \textsc{R-BasicConst} is a specialised rule for constants $\mathbf{op} : 1 \rightarrowtriangle \tau$ in $\BasicOps$ where $\tau \in \BaseTypes \cup \{ \answertype \}$.
The second rule \textsc{R-BasicSimp} gives a typing rule for basic operators whose arity and coarity are given by $\mathbf{op} : \sigma_1 \times \dots \times \sigma_n \rightarrowtriangle \tau$ where $\sigma_1, \dots, \sigma_n, \tau \in \BaseTypes \cup \{ \answertype \}$.
Note that many binary arithmetic operators conform to this pattern.
The last rule \textsc{R-BasicBool} is a specialised rule for basic operators that return boolean values of type $1 + 1$.
\textsc{R-BasicBool} still contains semantic conditions, but they are much easier to handle than the general rule \textsc{R-BasicOp}.
The semantic conditions for \textsc{R-BasicBool} means that $\psi_t$ and $\psi_f$ are conditions on $\mathrm{ar}(\mathbf{op})$ under which $\mathbf{op}$ returns $\mathbf{true}$ and $\mathbf{false}$, respectively.
For example, \textsc{R-LeqInt} in Fig.~\ref{fig:basic-operator-typing-specialised} is a more specialised rule for a basic operator $(({\le_{\mathbf{int}}}) : \mathbf{int} \times \mathbf{int} \rightarrowtriangle 1 + 1) \in \BasicOps$, in which we use $\psi_t(v_1, v_2) \coloneqq v_1 \le v_2$ and $\psi_f(v_1, v_2) \coloneqq v_1 \le v_2 \implies \bot$, assuming that we have a corresponding atomic predicate $({\le}_{\mathbf{int}}) : \mathbf{int} \times \mathbf{int}$ in $\AtomicPreds$.

\begin{proposition}
	Typing rules in Fig.~\ref{fig:basic-operator-typing} are sound in any $\lambdaHFL$-model.
	\qed
\end{proposition}

\begin{figure}
	\begin{mathpar}
		\small
		\inferrule[R-LeqInt]{
			\dot{\Gamma} \vdash (M, N) : \{ (v_1, v_2) : \mathbf{int} \times \mathbf{int} \mid v_1 \le v_2 \land \phi_t[()/v] \lor (v_1 \le v_2 \implies \bot) \land \phi_f[()/v] \}
		}{
			\dot{\Gamma} \vdash M \le_{\mathbf{int}} N : \{ v : 1 \mid \phi_t \} + \{ v : 1 \mid \phi_f \}
		}
		\and
		\inferrule[R-Unif]{
			\underlying{\dot{\Gamma}} \vdash N : \mathbf{real} \to \answertype \\
			\dot{\Gamma} \vdash M : (x : \{ x : \mathbf{real} \mid 0 \le x \land x \le 1 \}) \to \{ v : \answertype \mid v \le N\ x \}
		}{
			\dot{\Gamma} \vdash \mathbf{unif}(M) : \{ v : \answertype \mid v \le \mathbf{unif}(N) \}
		}
	\end{mathpar}
	\caption{Specialised typing rules for basic operators.}
	\label{fig:basic-operator-typing-specialised}
\end{figure}

A notable exception to three common patterns in Fig.~\ref{fig:basic-operator-typing} is the integration operator for a continuous distribution.
For example, to reason about uniform distribution over the unit interval $[0, 1]$, we consider an integration operator $\mathbf{unif} : (\mathbf{real} \to \answertype) \rightarrowtriangle \answertype$ defined in~\eqref{eq:unif-integration-operator} whose semantics is given by $f \mapsto \int_0^1 f(x)\, \mathrm{d} x$.
Since the arity of $\mathbf{unif}$ is a function type $\mathbf{real} \to \answertype$, we cannot apply any of Fig.~\ref{fig:basic-operator-typing}.
In such a situation, we need to design a typing rule on a case-by-case basis.
The general rule \textsc{R-BasicOp} remains sound for all basic operators including $\mathbf{unif}$, but \textsc{R-BasicOp} is not very convenient in practice.
We provide a typing rule \textsc{R-Unif} for $\mathbf{unif}$ in Fig.~\ref{fig:basic-operator-typing-specialised}.
\textsc{R-Unif} asserts: if a term $N$ of type $\mathbf{real} \to \answertype$ is an upper bound of $M$, then $\mathbf{unif}(M)$ is upper bounded by $\mathbf{unif}(N)$.
This is sound by the monotonicity of integration.
Here, note that we don't have to restrict ourselves to the uniform distribution.
It is easy to consider similar typing rules for other probability distributions and to prove soundness.

\begin{proposition}
	\textsc{R-Unif} is sound for the $\lambdaHFL$-model for weakest pre-expectations, expected costs, cost moments, and conditional weakest pre-expectations.
\end{proposition}
\begin{proof}
	By the monotonicity of $\interpret{\mathbf{unif}}$.
	See \referappendix{sec:detail-proof}{C} for details.
\end{proof}

\section{Type Checking Algorithm}\label{sec:type-check}
\newcommand{\CHCAdmInt}{$\mathbf{CHC}[\mathrm{adm}, \int]$}

As we have seen in Section~\ref{sec:instance}, many verification problems for probabilistic programs can be expressed as type-checking problems for our dependent refinement type system.
Given a term $M$ of $\lambdaHFL$, a refinement context $\dot{\Gamma}$, and a refinement type $\dot{\sigma}$ for the term, the \emph{type-checking problem} is the problem of deciding whether $\dot{\Gamma} \vdash M : \dot{\sigma}$ is well-typed.

In this section, we explain a reduction from type-checking problems to constraint solving, which is based on the standard type checking algorithm \cite{rondon2008,unno2009} for refinement type systems.
The main difference from the standard one is that we consider Constrained Horn Clauses (CHC) constraints extended with two new types of predicate variables: one is for admissible predicates, and the other is for integration operators.
We call this \CHCAdmInt.
We generate a set of constraints whose satisfiability implies the well-typedness of a given term $\dot{\Gamma} \vdash M : \dot{\sigma}$.
Note that specifications by refinement types are only required for top-level declarations because our type checker can synthesise inductive invariants automatically.
Implementing a constraint solver for the extended CHC will be explained later in Section~\ref{sec:implementation}.

\subsection{Reduction to CHC Constraints}

The workflow of type checking is given as follows.
Given a term of $\lambdaHFL$, we first apply the Hindley--Milner type inference algorithm \cite{damas1982} to obtain simple types for each sub-terms of $M$.
Then, we generate templates of refinement types by replacing each occurrence of $\sigma \in \BaseTypes \cup \{ 1, \answertype \}$ in simple types with the refinement type $\{ x : \sigma \mid P(\tilde{y}) \}$ where $P$ is a (fresh) predicate variable and $\tilde{y}$ is a list of variables visible from the current scope.
We generate \CHCAdmInt-constraints on predicate variables by applying typing rules of the dependent refinement type system.
Those \CHCAdmInt-constraints have three types of predicate variables: ordinary predicate variables, \emph{admissible predicate variables}, and \emph{integrable predicate variables}.
We use ordinary predicate variables for most of the typing rules, just like the standard type-checking algorithm \cite{rondon2008,unno2009}.
Admissible predicate variables and integrable predicate variables are used for \textsc{R-Fix} and \textsc{R-Unif}.
Once we obtain CHC constraints, we solve them using a CHC solver extended for \CHCAdmInt.
If a solution is found, then $\vdash M : \dot{\sigma}$ is well-typed.
This reduction to \CHCAdmInt-constraints generates a set of CHC constraints whose size grows linearly with respect to the size of a given program.
\footnote{Our current implementation of the constraint generator performs a preprocessing step that eliminates as many redundant predicate variables as possible to assist the backend CHC solver.
Although the original CHC constraints are of linear size, this preprocessing can cause an exponential blowup of CHC constraints in the worst case.
This issue is left for future work because the experiments (Section~\ref{sec:implementation}) showed that we can still solve many benchmarks.}

We explain more about the new types of predicate variables.
A predicate variable in CHC is denoted by $P(\tilde{x})$ where $\tilde{x}$ is a set of variables on which $P$ depends.
An \emph{admissible predicate variable} $P(v; \tilde{x})$ is defined as a predicate variable $P(v, \tilde{x})$ that must be instantiated by a predicate $\phi(y, \tilde{x})$ (i.e.\ a formula with free variables $y, \tilde{x}$) that is admissible at $y$.
Since \textsc{R-Fix} requires the admissibility of the predicate on the codomain of a fixed point $\fixpoint{f}{M} : \sigma \to \answertype$, we use an admissible predicate variable when generating constraints using \textsc{R-Fix}.
An \emph{integrable predicate variable} $P(v; y; \tilde{x})$ is a predicate variable $P(v, y, \tilde{x})$ that must be instantiated by a formula of the form $v \le N\ y$ (recall \textsc{R-Unif}) where free variables in $N$ must be in $\tilde{x}$.
For each integrable predicate variable $P(v; y; \tilde{x})$, we have an associated predicate variable $\mathbf{Integ}_{\mathbf{unif}}(P)(v; \tilde{x})$.
Whenever $P(v; y; \tilde{x})$ is instantiated to $v \le N\ y$, $\mathbf{Integ}_{\mathbf{unif}}(P)(v; \tilde{x})$ is instantiated to $v \le \mathbf{unif}(N)$ at the same time.
We use an integrable predicate variable when generating CHC constraints using \textsc{R-Unif}.

\begin{example}\label{ex:constraint-generation-simple}
	We explain the CHC constraint generation using a simple example.
	Consider type-checking the following term.\footnote{This term has a subterm of type $\mathbf{int} \to \mathbf{int}$, which is, strictly speaking, not allowed in our refinement type system.
	However, we are allowing such function types only in this example for illustrative purposes.}
	\begin{equation}
		(\lambda x. x + 1)\ 42 \quad:\quad \{y : \mathbf{int} \mid y \ge 0\}
		\label{eq:ex:constraint-generation}
	\end{equation}

	First, we apply the Hindley--Milner type inference, which yields the simple type for each subtem of \eqref{eq:ex:constraint-generation}.
	For example, $(\lambda x. x + 1)\ 42$ has type $\mathbf{int}$ and $\lambda x. x + 1$ has type $\mathbf{int} \to \mathbf{int}$.
	Then, we use a procedure $\mathbf{ConstGen}$ for constraint generation, which takes a refinement context $\dot{\Gamma}$, a term $M$, and a type annotation $\dot{\sigma}$; and returns CHC constraints.
	The procedure $\mathbf{ConstGen}$ recursively applies typing rules for our refinement type system.
	In this example, we invoke $\mathbf{ConstGen}$ with the following arguments.
	\[ \mathbf{ConstGen}\big({\cdot},\quad (\lambda x. x + 1)\ 42,\quad \{y : \mathbf{int} \mid y \ge 0\}\big) \qquad \text{where $\cdot$ is the empty (refinement) context.} \]
	To compute this, \textsc{R-App} should be applied first.
	Since $\mathbf{ConstGen}$ needs to guess a refinement type of $\lambda x. x + 1$, $\mathbf{ConstGen}$ generates fresh predicate variables and builds a refinement-type template $(x : \{ x : \mathbf{int} \mid P_1(x) \}) \to \{ y : \mathbf{int} \mid P_2(x, y) \}$.
	For the function application to be well-typed, the following constraint \eqref{eq:const-gen} should be satisfied where $\mathbf{SubType}(\dot{\Gamma}, \dot{\sigma}, \dot{\tau})$ is a procedure for generating constraints for the subtyping relation $\dot{\Gamma} \vdash \dot{\sigma} <: \dot{\tau}$.
	\begin{equation}
		\begin{aligned}
			&\mathbf{ConstGen}\big({\cdot},\ \lambda x. x + 1,\ (x : \{ x : \mathbf{int} \mid P_1(x) \}) \to \{ y : \mathbf{int} \mid P_2(x, y) \}\big) \\
			&\cup \mathbf{ConstGen}\big({\cdot},\ 42,\ \{ x : \mathbf{int} \mid P_1(x) \}\big) \mathrel{\cup} \mathbf{SubType}({\cdot},\ \{ y : \mathbf{int} \mid P_2(42, y) \},\ \{y : \mathbf{int} \mid y \ge 0\})
		\end{aligned}
		\label{eq:const-gen}
	\end{equation}
	In other words, $\mathbf{ConstGen}$ computes constraints using \textsc{R-App-ConstGen} below.
	\begin{mathpar}
		\inferrule[R-App-ConstGen]{
			\text{Let $\dot{\sigma} \to \dot{\tau}_1$ be a refinement-type template obtained from the (simple) type of $M$.} \\
			\dot{\Gamma} \vdash M : \dot{\sigma} \to \dot{\tau}_1 \\
			\dot{\Gamma} \vdash N : \dot{\sigma} \\
			\dot{\Gamma} \vdash \dot{\tau}_1[N/x] <: \dot{\tau}_2
		}{
			\dot{\Gamma} \vdash M\ N : \dot{\tau}_2
		}
	\end{mathpar}
	The rest of the constraint generation proceeds as follows.
	The first component of \eqref{eq:const-gen} is computed by applying \textsc{R-Abs} and then applying \textsc{R-BasicSimp} for ${+} : \mathbf{int} \times \mathbf{int} \to \mathbf{int}$, which yields $P_1(x) \implies P_2(x, x + 1)$.
	By \textsc{R-BasicConst}, the second component of \eqref{eq:const-gen} is $x = 42 \implies P_1(x)$.
	By the definition of subtyping relation, the third component of \eqref{eq:const-gen} is $P_2(42, y) \implies y \ge 0$.
	As a result, we get the following CHC constraints.
	\[ P_1(x) \implies P_2(x, x + 1) \qquad x = 42 \implies P_1(x) \qquad P_2(42, y) \implies y \ge 0 \]
	Since the CHC constraints are satisfiable (for example, let $P_1(x) = (x = 42)$ and $P_2(x, y) = (y \ge 0)$), we conclude that \eqref{eq:ex:constraint-generation} is well-typed.
	\qed
\end{example}

\begin{figure}
	\begin{mathpar}
		\inferrule[R-Fix-ConstGen]{
			\tau \to \answertype = \mathrm{SType}(\fixpoint{f}{M}) \\
			\text{Let $\dot{\tau}$ be a refinement-type template obtained from $\tau$.} \\
			\text{Let $P^{\mathrm{adm}}$ be a fresh \emph{admissible} predicate variable.} \\
			\dot{\Gamma}, f : (x : \dot{\tau}) \to \{ v : \answertype \mid P^{\mathrm{adm}}(v; \mathrm{vars}(\dot{\Gamma}), x) \} \vdash M : (x : \dot{\tau}) \to \{ v : \answertype \mid P^{\mathrm{adm}}(v; \mathrm{vars}(\dot{\Gamma}), x) \} \\
			\Gamma \vdash (x : \dot{\tau}) \to \{ v : \answertype \mid P^{\mathrm{adm}}(v; \mathrm{vars}(\dot{\Gamma}), x) \} <: (x : \dot{\sigma}) \to \{ v : \answertype \mid \phi \}
		}{
			\dot{\Gamma} \vdash \fixpoint{f}{M} : (x : \dot{\sigma}) \to \{ v : \answertype \mid \phi \}
		}
		\and
		\inferrule[R-Unif-ConstGen]{
			\text{Let $P^{\mathrm{int}}$ be a fresh integration predicate variable.} \\
			\dot{\Gamma} \vdash M : (x : \{ x : \mathbf{real} \mid 0 \le x \land x \le 1 \}) \to \{ v : \answertype \mid P^{\mathrm{int}}(v; x; \mathrm{vars}(\dot{\Gamma})) \} \\
			\dot{\Gamma} \vdash \{ v : \answertype \mid \mathbf{Integ}_{\mathbf{unif}}(P^{\mathrm{int}})(v; \mathrm{vars}(\dot{\Gamma})) \} <: \{ v : \answertype \mid \phi \}
		}{
			\dot{\Gamma} \vdash \mathbf{unif}(M) : \{ v : \answertype \mid \phi \}
		}
	\end{mathpar}
	\caption{Selected rules for constraint generation. We assume that the arguments of predicate variables are of type $\mathbf{int}$, $\mathbf{real}$, or $\answertype$, and that variables that are not typed by those types are implicitly removed from the arguments of predicate variables.}
	\label{fig:rule-constgen}
\end{figure}

\begin{example}\label{ex:constraint-generation-random-walk-unif}
	We explain how constraints are generated for \textsc{R-Fix} and \textsc{R-Unif}.
	Consider type-checking the following term, which is taken from the random walk example in \eqref{eq:random-walk-unif-cpsed}.
	\begin{equation}
		\fixpoint{\mathrm{rw}}{(\lambda\ x\ k. \ifexpr{x \ge 0}{
			\mathbf{unif}\ (\lambda y. 1 + \mathrm{rw}\ (x + 3 \cdot y - 2)\ k)
		}{k\ ()})}
		\label{eq:rw-unif-fix}
	\end{equation}
	The type annotation is given as follows, which is essentially the same as \eqref{eq:random-walk-unif-cpsed-type}.
	\begin{equation}
		\mathrm{rw} : (x : \{ x : \mathbf{real} \mid x \ge -2 \}) \to (k : (u : \mathbf{unit}) \to \{ r : \answertype \mid r = 0 \}) \to \{ r : \answertype \mid r \le |2 \cdot x + 4| \}
	\end{equation}
	Similarly to Example~\ref{ex:constraint-generation-simple}, we apply the Hindley--Milner type inference to infer simple types and then apply the procedure $\mathbf{ConstGen}$ for constraint generation to \eqref{eq:rw-unif-fix}.
	Let's take a closer look at the constraint generation.
	$\mathbf{ConstGen}$ first applies \textsc{R-Fix-ConstGen} in Fig.~\ref{fig:rule-constgen} to handle the fixed point.
	\textsc{R-Fix-ConstGen} introduces the refinement-type template \eqref{eq:rw-unif-refinement-template} for $\mathrm{rw}$ where $P^{\mathrm{adm}}_{\mathrm{rw}}$ is an admissible predicate variable.
	\begin{equation}
		\begin{aligned}
			(x : \{ x : \mathbf{real} \mid P_x(x) \}) &\to (k : (u : \{u : \mathbf{unit} \mid P_u(x)\}) \to \{ r : \answertype \mid P_k(x, r) \}) \\
			&\to \{ r : \answertype \mid P^{\mathrm{adm}}_{\mathrm{rw}}(r; x) \}
		\end{aligned}
		\label{eq:rw-unif-refinement-template}
	\end{equation}
	\textsc{R-Fix-ConstGen} also requires that \eqref{eq:rw-unif-refinement-template} must be a subtype of \eqref{eq:random-walk-unif-cpsed-type}.
	After handling the fixed point, $\mathbf{ConstGen}$ handles lambda abstraction and if-then-else, which is standard and straightforward.
	In the then clause, we have an integration operator, to which $\mathbf{ConstGen}$ applies \textsc{R-Unif-ConstGen} (Fig.~\ref{fig:rule-constgen}).
	\textsc{R-Unif-ConstGen} introduces a fresh integration predicate variable $P^{\mathrm{int}}$ for the codomain type of the integrand.
	After finishing the constraint generation and simplification, we get the following CHC constraints.
	\begin{gather}
		\left. \begin{aligned}
			x \ge -2 &\implies P_x(x) \\
			x \ge -2 \land P_{u}(x) \land r = 0 &\implies P_{k}(x, r) \\
			x \ge -2 \land P^{\mathrm{adm}}_{\mathrm{rw}}(r; x) &\implies r \le |2 \cdot x + 4|
		\end{aligned} \qquad\qquad\right] \quad\text{\eqref{eq:rw-unif-refinement-template} is a subtype of \eqref{eq:random-walk-unif-cpsed-type}} \\[1ex]
		\left. \begin{aligned}
			P_x(x) \land x \ge 0 \land \mathbf{Integ}_{\mathbf{unif}}(P^{\mathrm{int}})(r; x) &\implies P^{\mathrm{adm}}_{\mathrm{rw}}(r; x) \\
			P_x(x) \land x \ge 0 \land 0 \le y \land y \le 1 &\implies P_x(x + 3 \cdot y - 2) \\
			P_x(x) \land x \ge 0 \land 0 \le y \land y \le 1 \land P^{\mathrm{adm}}_{\mathrm{rw}}(r; x + 3 \cdot y - 2) &\implies P^{\mathrm{int}}(r + 1; y; x)
		\end{aligned} \quad\right] \ \text{then clause} \\[1ex]
		\left. \begin{aligned}
			P_x(x) \land x < 0 &\implies P_{u}(x) \\
			P_x(x) \land x < 0 \land P_{k}(x, r) &\implies P^{\mathrm{adm}}_{\mathrm{rw}}(r; x)
		\end{aligned} \qquad\qquad\right] \quad\text{else clause}
	\end{gather}
	These constraints are satisfiable: Let \eqref{eq:rw-unif-refinement-template} be equal to \eqref{eq:random-walk-unif-cpsed-type} and $P^{\mathrm{int}}(r; y; x) = (r \le 2 \cdot x + 6 \cdot y + 1)$.
	Note that under this assignment, we have $\mathbf{Integ}_{\mathbf{unif}}(P^{\mathrm{int}})(r; x) = (r \le \int_0^1 2 \cdot x + 6 \cdot y + 1 \,\mathrm{d}y) = (r \le 2 \cdot x + 4)$.
	\qed
\end{example}

\subsection{General Syntactic Conditions for Admissible Formulas}
We provide general syntactic conditions for admissible formulas in Fig.~\ref{fig:admissible-formula}.
These rules give us hints when we implement a constraint solver for CHC with admissible predicates.
We consider a judgement of the form $\Gamma \vdash \mathrm{adm}(v, \phi)$, which means $\phi$ is admissible at $v$, and provide derivation rules in Fig.~\ref{fig:admissible-formula}.
Those rules are sound, as stated below.

\begin{figure}
	\[ (\Gamma, v : \tau) \setminus v \coloneqq \Gamma \qquad (\Gamma, x : \tau) \setminus v \coloneqq (\Gamma \setminus v), x : \tau \]
	\begin{mathpar}
		\inferrule[Adm-Leq]{
			(v : \answertype) \in \Gamma \\	
			\Gamma \vdash M : \answertype
		}{
			\Gamma \vdash \mathrm{adm}(v, v \le_{\answertype} M)
		}
		\and
		\inferrule[Adm-Imp]{
			(\Gamma \setminus v) \vdash \phi \\
			\Gamma \vdash \mathrm{adm}(v, \psi)
		}{
			\Gamma \vdash \mathrm{adm}(v, \phi \implies \psi)
		}
		\and
		\inferrule[Adm-And]{
			\Gamma \vdash \mathrm{adm}(v, \phi) \\
			\Gamma \vdash \mathrm{adm}(v, \psi)
		}{
			\Gamma \vdash \mathrm{adm}(v, \phi \land \psi)
		}
		\and
		\inferrule[Adm-Or]{
			\Gamma \vdash \mathrm{adm}(v, \phi) \\
			\Gamma \vdash \mathrm{adm}(v, \psi)
		}{
			\Gamma \vdash \mathrm{adm}(v, \phi \lor \psi)
		}
	\end{mathpar}
	\caption{Syntactic rules to guarantee admissibility of formulas.}
	\label{fig:admissible-formula}
\end{figure}

\begin{theorem}
	If \hspace{0.5ex}$\Gamma \vdash \mathrm{adm}(v, \phi)$, then $\Gamma \vdash \phi$ is well-formed is admissible at $v$ for any $\lambdaHFL$-model.
	\qed
\end{theorem}

Note that for $v, M : \answertype$, $\Gamma \vdash \mathrm{adm}(v, v < M)$ and $\Gamma \vdash \mathrm{adm}(v, M \le v)$ are not derivable because $v < M$ and $M \le v$ are not admissible in general (see Example~\ref{ex:admissible-in-ext-nonneg-real} for counterexamples).
This poses a difficulty in reasoning about lower bounds of the least fixed point in our dependent refinement type system.
In general, reasoning about lower bounds for probabilistic programs is more difficult than upper bounds.
There are a few papers \cite{beutner2021,feng2023,hark2020,mciver2005} on reasoning about lower bounds, but it is not straightforward to combine these methods with our general framework.
We would like to tackle this in future work.
Note that we can reason about lower bounds if an HFL term does not contain least fixed points.


\section{Implementation and Experiments}\label{sec:implementation}

Following Section~\ref{sec:type-check}, our type checker is implemented as an extension of
\textsc{RCaml},
in which
\textsc{PCSat}
is used as a CHC solver.
The CHC solver is based on a method called counterexample-guided inductive synthesis (CEGIS) with linear templates and extended to \CHCAdmInt-constraints.
First, we will explain how we extend the constraint solver to admissible predicate variables and integrable predicate variables.
Then, we will present and discuss experimental results.

\subsection{CEGIS-Based CHC Solving}
CEGIS \cite{solar-lezama2006} is a method for constraint solving and consists of two components: a synthesiser and a validator.
The synthesiser guesses a solution, and the validator checks whether a candidate solution is a genuine solution.
Given a set of constraints (e.g.\ CHC constraints), the synthesiser and the validator iteratively interact with each other and gradually improve candidate solutions by accumulating counterexamples found by the validator.

The synthesiser of our CHC solver guesses a candidate solution using (affine) templates.
The coefficients of templates are determined using counterexamples.
The validator of our CHC solver checks a candidate solution using an SMT solver.
If a candidate solution is not a genuine solution, the SMT solver gives new counterexamples, which are used in the next iteration of CEGIS.

\subsection{Supporting Admissible Predicate Variables}\label{subsec:admissible}
The basic strategy for supporting new types of predicate variables is to appropriately restrict the search space of solutions.
In our CHC solver, a search space for predicate variables is defined by a \emph{template}, which is a set of formulas that contain unknown parameters.
To support admissible predicate variables, we define \emph{admissible templates} such that any instantiation of unknown parameters gives a formula that is syntactically admissible.
Since the interpretation of $\answertype$ varies for each problem in Section~\ref{sec:instance}, different problems need different admissible templates.
We design and implement admissible templates for the problems considered in this paper using the syntactic rules for admissible predicates in Fig.~\ref{fig:admissible-formula} as a guide because those syntactic rules are sound for any $\lambdaHFL$-model (i.e., for any interpretation of $\answertype$).

\subsubsection{An Admissible Template for Weakest Pre-Expectations and Expected Costs}
Let $P_{\mathrm{adm}}(v; \tilde{x})$ be an admissible predicate variable where $v$ is a variable at which $P_{\mathrm{adm}}(v; \tilde{x})$ is expected to be admissible and $\tilde{x} = \{ x_1, \dots, x_m \}$ is a set of variables.
We assume that each variable is assigned a type and that the type of each variable belongs to $\{ \mathbf{int}, \mathbf{real}, \answertype \}$.
This ensures that templates can be handled by SMT-based constraint solvers.
In the implementation, this assumption can be satisfied by simply ignoring variables that are not typed as $\mathbf{int}$, $\mathbf{real}$, or $\answertype$.

For weakest pre-expectations and expected costs ($\interpret{\answertype} = [0, \infty]$), the search space for $P_{\mathrm{adm}}(v; \tilde{x})$ is defined by the following admissible template.
\begin{equation}
	d \cdot v \le e \qquad \text{where} \qquad e \quad\coloneqq\quad \ifexpr{\phi}{e_1}{e_2} \quad\mid\quad |c_0 + \sum_{i=1}^m c_i \cdot x_i|
	\label{eq:admissible-template-wp-ec}
\end{equation}
Here, $d \in \mathbb{Z}$ ($d \ge 0$) and $c_i \in \mathbb{Z}$ ($i = 0, \dots, m$) are unknown parameters; and $\phi$ is a conjunction of affine inequations over $\tilde{x}$.
Note that we are essentially using rational numbers as unknown parameters since $d$ serves as the denominator.
Furthermore, $d = 0$ corresponds to the case of $v \le \infty$.
We assume that we have type conversion operators like $\mathbf{int} \rightarrowtriangle \mathbf{real}$ and the absolute-value operator $|{-}| : \mathbf{real} \rightarrowtriangle \answertype$ as basic operators; and type conversions are implicitly used in $c_0 + \sum_{i=1}^m c_i \cdot x_i$.
Precisely speaking, $\phi$ is not a formula defined in Section~\ref{sec:refinement} but a term of $\lambdaHFL$ of type $1 + 1$.
However, we don't have to worry too much about it because such a term is easily definable in $\lambdaHFL$.

\begin{lemma}\label{lem:admissible-template-wp-ec}
	Any instantiation of the admissible template~\eqref{eq:admissible-template-wp-ec} is a formula admissible at $v$.
\end{lemma}
\begin{proof}
	By induction on the definition of the admissible template.
	For the base case, $d \cdot v \le |c_0 + \sum_{i=1}^m c_i \cdot x_i|$ is admissible by \textsc{Adm-Leq} in Fig.~\ref{fig:admissible-formula}.
	For the step case, $d \cdot v \le \ifexpr{\phi}{e_1}{e_2}$ is admissible because the equivalent formula $(\phi \implies d \cdot v \le e_1) \land (\lnot \phi \implies  d \cdot v \le e_2)$ is admissible by \textsc{Adm-And} and \textsc{Adm-Imp}.
\end{proof}

\subsubsection{An Admissible Template for Cost Moment}
Consider the case of cost moment analyses ($\interpret{\answertype} = [0, \infty]^n$).
Similarly to the case of expected cost analyses, we define the following admissible template for $P_{\mathrm{adm}}(v; \tilde{x}) = P_{\mathrm{adm}}((v_1, \dots, v_n); \tilde{x})$.
\[ \bigwedge_{i = 1}^n d_i \cdot v_i \le e_i \qquad \text{where} \qquad e \quad\coloneqq\quad \ifexpr{\phi}{e_1}{e_2} \quad\mid\quad |c_0 + \sum_{i=1}^m c_i \cdot x_i| \]
Here, we decompose $v : \answertype$ to a tuple of non-negative extended real numbers $(v_1, \dots, v_n) : (\ExtNonnegRealType)^n$.
Similarly to Lemma~\ref{lem:admissible-template-wp-ec}, any instantiation of the admissible template above is an admissible formula at $v$.

\subsubsection{An Admissible Template for Conditional Weakest Pre-Expectation}
Consider the case of conditional weakest pre-expectation ($\interpret{\answertype} = [0, \infty] \times [0, 1]^{\op}$).
We define an admissible template for $P_{\mathrm{adm}}((v_1, v_2); \tilde{x})$ as follows.
\[ d_1 \cdot v_1 \le e_1 \land d_2 \cdot v_2 \ge \min \{ d_2, e_2 \} \qquad \text{where} \qquad e \quad\coloneqq\quad \ifexpr{\phi}{e_1}{e_2} \quad\mid\quad |c_0 + \sum_{i=1}^m c_i \cdot x_i| \]
Note that $d_2 \cdot v_2 \ge \min \{ d_2, e_2 \}$ gives a \emph{lower bound} of $v_2$ because $(v_1, v_2) \le_{\answertype} (u_1, u_2)$ is defined by $v_1 \le u_1$ and $v_2 \ge u_2$.
Note also that the right-hand side of $d_2 \cdot v_2 \ge \min \{ d_2, e_2 \}$ ensures the range of $\min \{ 1, e_2 / d_2 \}$ is subsumed by $[0, 1]$ if $d_2 > 0$.
In other words, the admissible template is designed so that we have $(e_1 / d_1, \min \{ 1, e_2 / d_2 \}) : \answertype$.

\subsection{Supporting Integrable Predicate Variables}\label{subsec:integration}
Our implementation uses the following templates for integrable predicate variables.
Recall that an integrable variable $P(v; y; \tilde{x})$ has to be instantiated as $v \le N\ y$.
We consider restricting the form of $N : \mathbf{real} \to \answertype$ to the affine expression $N = \lambda y. c \cdot y + N'$ where $N'$ is an affine expression over $\tilde{x}$.
That is, the template for $P(v; y; \tilde{x})$ is given by $v \le c \cdot y + N'$.
Here, note that $y$ does not occur in $N'$.
Since $\mathbf{unif}(N)$ can be computed as $1/2 \cdot c + N'$ by the linearity of integration, we can instantiate $\mathbf{Integ}_{\mathbf{unif}}(P)(v; \tilde{x})$ as $v \le 1/2 \cdot c + N'$.

\subsection{Results and Discussion}

Table~\ref{tab:experiment} shows the result of our experiments.
Most of the benchmarks are collected from past papers on verification of probabilistic programs \cite{olmedo2018,olmedo2016,avanzini2021}.
Specifications for benchmarks are annotated by hand as refinement types.
Note that annotations are required for only top-level declarations, and safe inductive invariants for the least fixed points of $\lambdaHFL$ are automatically synthesised by the constraint solver.
Note also that our analysis is modular rather than whole-program because our implementation analyses a benchmark program function by function.
Our benchmarks contain four verification problems for higher-order probabilistic programs discussed in this paper.
The results were obtained on 12th Gen Intel(R) Core(TM) i7-1270P   2.20 GHz with 32 GB of memory.

The aim of the evaluation is to demonstrate that by instantiating our general framework with a specific refinement type checker and a CHC solver, we can indeed build a verifier that is effective for
(a) 4 problem instances on cost moment analysis and conditional weakest pre-expectation of recursive probabilistic programs that cannot be handled by existing methods and
(b) 10 problem instances on weakest pre-expectation and expected cost analysis of (higher-order) recursive probabilistic programs that few existing methods can handle.
These problem instances shown in \referappendixnocite{sec:benchmark}{D} are rather small, but they were selected because they pose challenges for existing automated methods, and our tool has been able to verify most of them fully automatically.

\begin{table}
	\caption{Experimental results. Benchmarks are listed in \protect\referappendixnocite{sec:benchmark}{D}. The timeout is set to 300 seconds.}
	\label{tab:experiment}
	\small
	\begin{tabular}{clc}
		Problem & Benchmark & Time (sec) \\
		\hline
		\multirow{3}{*}{Weakest pre-expectation} & \texttt{lics16\_rec3} (Example~\ref{ex:rec3}) & timeout \\
		& \texttt{lics16\_rec3\_ghost} & 1.270 \\
		& \texttt{lics16\_coins} & 3.110 \\
		\hline
		\multirow{7}{*}{Expected cost analysis} & \texttt{random\_walk} & 2.761 \\
		& \texttt{random\_walk\_unif} (Example~\ref{ex:random-walk-unif}) & 7.508 \\
		& \texttt{coin\_flip} (Example~\ref{ex:coin-flip}) & 0.718 \\
		& \texttt{coin\_flip\_unif} & 0.884 \\
		& \texttt{icfp21\_walk} & 3.532 \\
		& \texttt{icfp21\_coupons} & timeout \\
		& \texttt{lics16\_fact} & 3.383 \\
		\hline
		\multirow{2}{*}{Cost moment analysis} & \texttt{coin\_flip\_ord2} (Example~\ref{ex:coin-flip-moment}) & 1.135 \\
		& \texttt{coin\_flip\_ord3} & 4.040 \\
		\hline
		\multirow{2}{*}{Conditional weakest pre-expectation} & \texttt{toplas18\_ex4.4} (Example~\ref{ex:three-coin}) & timeout \\
		& \texttt{two\_coin\_conditioning} & 1.079
	\end{tabular}
\end{table}

The experimental result (Table~\ref{tab:experiment}) shows that our implementation successfully solved most of the benchmarks even though we implemented it as a simple extension of the existing type checker for non-probabilistic programs.
While some benchmarks could be solved by the current implementation, it is worth noting that this is not because we consider the verification of probabilistic programs.
Rather, this is due to common issues with ordinary dependent refinement type systems for non-probabilistic programs.
Therefore, by incorporating advanced techniques for ordinary dependent refinement type systems, our implementation can potentially overcome these issues.

For example, our implementation could not solve \texttt{lics16\_rec3}.
When reasoning about $\recthree : \mathbf{unit} \to (\mathbf{unit} \to \answertype) \to \answertype$, the predicate on the result type $\answertype$ should be able to depend on the result of the second argument of $\recthree$, which has a function type $\mathbf{unit} \to \answertype$.
However, most implementations of dependent refinement type systems struggle to handle such situations effectively.
One possible way to overcome this problem is to add an extra ghost parameter and define $\recthree : \mathbf{unit} \to [a : \answertype] \to (\mathbf{unit} \to \{ v : \answertype \mid r \le a \}) \to \{ v : \answertype \mid r \le c \cdot a \}$ as follows.
\begin{align}
	&\letfix{\recthree}{x\ [a]\ k}{1/2 \cdot k\ () + 1/2 \cdot \recthree\ ()\ (\lambda y. \recthree\ ()\ [c^2 \cdot a]\ (\lambda y. \recthree\ ()\ [c \cdot a]\ k))}{\\&\qquad\qquad\qquad\qquad\qquad\recthree\ ()\ [a]\ (\lambda y. 1)} \label{eq:rec3-ghost}
\end{align}

Here, we indicate the extra parameter $a$ by square brackets $[{-}]$, and $c$ is a constant such that $(\sqrt{5} - 1) / 2 \le c \le 1$.
Our implementation can verify that \eqref{eq:rec3-ghost} is well-typed (\texttt{lics16\_rec3\_ghost} in Table~\ref{tab:experiment}).
This idea is proposed in \cite{unno2013}, in which they proved that the relative completeness of the refinement-type-based verification for \emph{non-probabilistic} higher-order programs can be achieved by inserting extra ghost parameters at appropriate positions.
We conjecture that a similar result holds for the probabilistic case, but we leave it for future work.

Our implementation also suffered from the limitations about affine templates for predicate variables (\texttt{toplas18\_ex4.4}).
A possible remedy for this issue is to use other constraint solvers that support more expressive templates, e.g., CEGIS for polynomial invariants \cite{sharma2013}, a synthesis method based on polynomial templates \cite{chatterjee2020}, and a method based on recurrence relations \cite{kincaid2017}.
\begin{remark}
	The restriction of integrable predicate variables to affine templates above arises from the following two requirements.
	(1) The integration of templates must be computable.
	(2) The class of templates must be solvable in constraint solvers.
	Affine templates satisfy both requirements:
	We can easily integrate affine templates by substituting the expected value, and affine constraints are easy to solve.
	If we consider piecewise affine templates, it would not be difficult to handle them in our constraint solver, but computing integral would be difficult.
	On the other hand, computing the integration of polynomial templates is easy because we only have to substitute $n$-th moment $\mathbb{E}[x^n]$ for $x^n$, but our current constraint solver has difficulty in solving polynomial constraints.
\end{remark}

There are a few existing tools that aim at the verification of higher-order probabilistic programs.
It is difficult to make a fair comparison because those tools have more or less different problem settings from ours.
However, we would like to note that none of them supports all the benchmarks listed in Table~\ref{tab:experiment}.
\cite{beutner2021} proposed a verification tool for almost sure termination (AST).
Note that the AST verification is a special case of the lower-bound verification of the weakest pre-expectation, and the positive-AST (AST within finite expected runtime) verification is the upper-bound verification of the expected cost.
The input format of their tool is rather restrictive compared to ours: input programs must be of the form $\fixpoint{f}{\lambda x : \mathbf{real}. M}$.
Their tool doesn't accept nested recursion or recursive functions that have more than one argument.
As a result, their tool accepts only four of the benchmarks in Table~\ref{tab:experiment} (\texttt{coin\_flip}, \texttt{lics16\_fact}, \texttt{random\_walk}, \texttt{random\_walk\_unif}) and was able to prove AST of only one of them (\texttt{coin\_flip} in 0.156 sec).
\cite{avanzini2023} proposed a verification tool for the upper bounds of the weakest pre-expectation and the expected cost of \emph{imperative} probabilistic programs with recursion.
Their tool accepts five of the benchmarks in Table~\ref{tab:experiment} (\texttt{coin\_flip}, \texttt{lics16\_coins}, \texttt{lics16\_fact}, \texttt{lics16\_rec3}, \texttt{random\_walk}) and was able to solve four of them (\texttt{coin\_flip} in 0.025 sec, \texttt{lics16\_coins} in 0.038 sec, \texttt{lics16\_fact} in 0.035 sec, \texttt{random\_walk} in 0.036 sec).


\section{Related Work}\label{sec:related-work}

\subsection{Weakest Pre-Expectations and Expected Costs}
The \emph{weakest pre-expectation transformer} \cite{mciver2005,olmedo2016} is a generalisation of the weakest precondition transformer \cite{dijkstra1975}.
This notion is used to verify properties such as termination probabilities and probabilistic invariants \cite{bao2022,batz2023a} for imperative probabilistic programming languages.
The \emph{expected runtime transformer} \cite{kaminski2018} is a similar notion proposed for verification of expected costs.
Concerning expected cost analyses, the notion of \emph{ranking supermartingales} \cite{chakarov2013} has also been studied and applied to automatic verification of almost sure termination.
It is known that ranking supermartingales give upper bounds of the expected cost and that this can be understood order-theoretically \cite{takisaka2018}.
\emph{Ranking supermartingales for higher moments} are proposed \cite{kura2019} as an extension of ranking supermartingales that gives upper bounds of the higher moments of runtime (i.e.\ the expected value of \emph{powers} of runtime).
By exploiting the relationship between ranking supermartingales and the expected runtime transformer, a cost moment transformer is obtained as a generalisation of the expected runtime transformer \cite{aguirre2022}.
The \emph{conditional weakest pre-expectation transformer} \cite{olmedo2018} is another extension of the weakest pre-expectation transformer for probabilistic programs with conditioning.
Studies listed above are targeted at imperative probabilistic programs, and higher-order probabilistic programs are out of scope.

\subsection{Verification of Higher-Order Probabilistic Programs}
\cite{avanzini2021} proposed a CPS transformation for the expected cost analysis for higher-order probabilistic programs.
Their CPS transformation translates a probabilistic program into a pure term of simply typed lambda calculus with fixed points whose denotation gives the expected cost.
They also provided a program logic EHOL to reason about CPS-transformed programs.
In contrast to our approach, they did not introduce an algorithm for automated verification, nor did they incorporate continuous distributions into their language.

There are several program logics proposed for verifying higher-order probabilistic programs \cite{sato2019a,aguirre2021a,aguirre2017}.
Some of them have been implemented on interactive theorem provers \cite{hirata2022}, but these approaches largely depend on human intervention.

Probabilistic extensions of higher-order model checking have been studied.
A probabilistic extension of higher-order recursion schemes is studied in \cite{kobayashi2020}.
They provide theoretical results on the decidability and hardness of model checking problems and experimental results.
A probabilistic extension of higher-order fixed point logic (PHFL) is proposed in \cite{mitani2020}.
It is worth mentioning that our $\lambdaHFL$ and their PHFL share most of the syntax of HFL.

As for automated verification, \cite{beutner2021} proposed algorithms to verify lower bounds of termination probability and almost sure termination for higher-order probabilistic programs with continuous distributions.
Although their language incorporates soft-conditioning, what they verify differs from conditional weakest pre-expectations considered in our work and \cite{olmedo2018} because their semantics essentially ignores soft-conditioning.
For example, ``score(0); diverge'' is equivalent to ``diverge'' in their semantics where ``diverge'' is a diverging program, whereas the conditional weakest pre-expectations for these two programs are different.
\cite{wang2020a} studied type-based amortised expected cost analysis for the expected cost of higher-order probabilistic programs.
However, their type system does not support continuous distributions.

\subsection{Verification of Higher-Order Non-Probabilistic Programs}
Dijkstra monads enable specifying and verifying programs via weakest precondition transformers.
In \cite{ahman2017}, a CPS transformation and refinement types are applied for automated verification of higher-order programs.
However, our approach and theirs are different.
We apply a CPS transformation to a program to be verified to obtain the weakest precondition transformer and then apply refinement types to reason about the weakest precondition transformer.
In contrast, \cite{ahman2017} applies a CPS transformation to a monad to define a Dijkstra monad and then uses the Dijkstra monad to verify programs.
Their work does consider refinement types and Dijkstra monads at the same time, but these are used as independent mechanisms.
In fact, \cite{maillard2019} studies Dijkstra monads in a setting without refinement types.
Another notable difference is that the method proposed in \cite{ahman2017} cannot be applied to probabilistic programs.
This is because their soundness is proved using deterministic program semantics.

\cite{katsura2020} proposed a refinement type system for a higher-order fixed-point logic $\nu\mathbf{HFL}$.
Their $\nu\mathbf{HFL}$ can be understood as a special case of our $\lambdaHFL$.
If we interpret $\answertype$ as $\{ \mathbf{true}, \mathbf{false} \}$ and choose a set $\BasicOps$ of basic operators appropriately, we get the quantifier-free fragment of $\nu\mathbf{HFL}$.
Moreover, we get the full $\nu\mathbf{HFL}$ by adding universal quantifiers to our $\lambdaHFL$, which is a straightforward extension.
We note that their refinement type system is developed for non-probabilistic programs and cannot be applied to the problems that we considered in this paper.

\section{Conclusions and Future Work}\label{sec:conclusion}

We proposed a dependent refinement type system for a generalised higher-order fixed point logic.
Combined with CPS-based encodings of properties of higher-order probabilistic programs, our approach provided an algorithm to verify several properties of probabilistic programs with continuous distributions and (hard) conditioning.
Our approach can seamlessly integrate with verification techniques developed for non-probabilistic programs.
We implemented our approach and demonstrated its ability.

In future work, we would like to improve our implementation by combining advanced techniques for refinement types and CHC solving.
For example, in Section~\ref{sec:instance}, we gave specifications of probabilistic programs using concrete upper bounds provided by hand.
To avoid specifying concrete upper bounds, it might be possible to use techniques like CHC optimisation \cite{gu2023} and refinement type optimization \cite{hashimoto2015}.
We are also interested in automated inference of ghost parameters \cite{unno2013} because we had to specify ghost parameters manually to solve Example~\ref{ex:rec3} in the current implementation.
On the theoretical side, our dependent refinement type system currently has the following limitations: we cannot reason about lower bounds of least fixed points (dually upper bounds of greatest fixed points).
Verifying lower bounds for probabilistic programs is studied in \cite{beutner2021,hark2020,feng2023,mciver2005}.
We would like to study a category-theoretic abstraction of these studies and combine it with our type system.
Making full use of the generality of our approach is another direction of future work.
For example, the combination of nondeterminism and probability is common in the verification of probabilistic programs.
Therefore, we would like to investigate an appropriate $\lambdaHFL$-model for this situation.
Expected cost analyses for quantum programs \cite{avanzini2022,liu2022} could be another interesting application of our approach.

\begin{acks}
	We are grateful to the anonymous reviewers for their insightful comments.
	This work was supported by JST ACT-X Grant Number JPMJAX2104; JSPS Overseas Research Fellowships; and JSPS KAKENHI
	Grant Numbers JP20H04162, JP20H05703, JP22H03564, and JP24H00699.
\end{acks}

\bibliographystyle{ACM-Reference-Format}
\bibliography{refinement-cps}

\ifthenelse{\boolean{longversion}}{
	\clearpage
	\appendix
	
\section{Details of HFL}\label{sec:detail-hfl}
Our $\lambda_{\mathbf{HFL}}$ is designed based on \cite{kura2023}.

Let $\BaseTypes$ be a set of base types and $\BasicOps$ be a set of typed signatures of basic operations.

\subsection{Typing Rules}
\begin{mathpar}
	\inferrule[S-Var]{
		(x : \sigma) \in \Gamma
	}{
		\Gamma \vdash x : \sigma
	}
	\and
	\inferrule[S-BasicOp]{
		\Gamma \vdash M : \mathrm{ar}(\mathbf{op})
	}{
		\Gamma \vdash \mathbf{op}\ M : \mathrm{car}(\mathbf{op})
	}
	\and
	\inferrule[S-Unit]{ }{
		\Gamma \vdash () : 1
	}
	\and
	\inferrule[S-Pair]{
		\Gamma \vdash M : \sigma \\
		\Gamma \vdash N : \tau
	}{
		\Gamma \vdash (M, N) : \sigma \times \tau
	}
	\and
	\inferrule[S-Fst]{
		\Gamma \vdash M : \sigma_1 \times \sigma_2
	}{
		\Gamma \vdash \pi_1\ M : \sigma_1
	}
	\and
	\inferrule[S-Snd]{
		\Gamma \vdash M : \sigma_1 \times \sigma_2
	}{
		\Gamma \vdash \pi_2\ M : \sigma_2
	}
	\and
	\inferrule[S-Case0]{
		\Gamma \vdash M : 0
	}{
		\Gamma \vdash \delta(M) : \tau
	}
	\and
	\inferrule[S-Inj]{
		\Gamma \vdash M : \sigma_i
	}{
		\Gamma \vdash \iota_i\ M : \sigma_1 + \sigma_2
	}
	\and
	\inferrule[S-Case2]{
		\Gamma \vdash M : \sigma_1 + \sigma_2 \\
		\Gamma, x_1 : \sigma_1 \vdash M_1 : \tau \\
		\Gamma, x_2 : \sigma_2 \vdash M_2 : \tau
	}{
		\Gamma \vdash \caseexpr{M}{x_1 : \sigma_1}{M_1}{x_2 : \sigma_2}{M_2} : \tau
	}
	\and
	\inferrule[S-Abs]{
		\Gamma, x : \sigma \vdash M : \answertype
	}{
		\Gamma \vdash \lambda x : \sigma. M : \sigma \to \answertype
	}
	\and
	\inferrule[S-App]{
		\Gamma \vdash M : \sigma \to \answertype \\
		\Gamma \vdash N : \sigma
	}{
		\Gamma \vdash M\ N : \answertype
	}
	\and
	\inferrule[S-Fix]{
		\Gamma, f : \sigma \to \answertype, \vdash M : \sigma \to \answertype
	}{
		\Gamma \vdash \fixpoint{f}{M} : \sigma \to \answertype
	}
\end{mathpar}

\subsection{Semantics}

Let $\category{C}$ be an $\omegaCPO$-enriched cartesian closed category.
The underlying ($\Set$-enriched) category of $\category{C}$ is denoted by $\category{C}_0$.

\begin{definition}
	An object $A \in \category{C}$ has a \emph{bottom element} if there exists $\bot^A : 1 \to A$ such that for any $f : 1 \to A$, we have $\bot^A \le f$ with respect to the partial order on $\category{C}(1, A)$.
\end{definition}

\begin{definition}
	We say $\category{C}$ \emph{admits a lifting strong monad} if
	\begin{itemize}
		\item the algebraic theory defined by one constant (nullary operation) $\bot$ and one inequation $\bot \le x$ has free algebras, that is, the following universal property holds: for any $X \in \category{C}$, there exists $X_{\bot} \in \category{C}$ and $\eta_X : X \to X_{\bot}$ such that if $f : X \to A$ is such that $A \in \category{C}$ has a bottom element, then there exists a unique $\overline{f} : X_{\bot} \to A$ such that $\overline{f} \comp \eta_X = f$, and
		\item the induced lifting monad $({-})_{\bot}$ is strong.
	\end{itemize}
\end{definition}
Note that an Eilenberg--Moore algebra of the lifting monad $({-})_{\bot}$ is just an object $A \in \category{C}$ with a bottom element.
Note also that a lifting monad defines a uniform $({-})_{\bot}$-fixed-point operator $({-})^{\dagger} : \category{C}_0(X_{\bot}, X_{\bot}) \to \category{C}_0(1, X_{\bot})$ for any $X \in \category{C}$.
\[ f^{\dagger} : 1 \to X_{\bot} \quad\coloneqq\quad \sup_n f^n \comp \bot^{X_{\bot}} \qquad \text{for any $f : X_{\bot} \to X_{\bot}$} \]
It is known that a uniform $T$-fixed-point operator uniquely extends to a parameterised uniform $T$-fixed-point operator $({-})^{\dagger} : \category{C}_0(X \times A, A) \to \category{C}_0(X, A)$ for any Eilenberg--Moore $T$-algebra $\emalgsymbol : T A \to A$.

\begin{definition}[$\lambdaHFL$-model]
	A \emph{$\lambdaHFL$-model} is a tuple $\mathcal{A} = (\category{C}, A, a, \answerobj)$ where
	\begin{itemize}
		\item $\category{C}$ is an $\omegaCPO$-enriched bicartesian closed category that admits a lifting monad.
		\item $A : \BaseTypes \to \category{C}_0$ is an interpretation of each base type.
		\item For each $\mathbf{op} : \mathrm{ar}(\mathbf{op}) \rightarrowtriangle \mathrm{car}(\mathbf{op})$, $a(\mathbf{op}) : \mathcal{A}\interpret{\mathrm{ar}(\mathbf{op})} \to \mathcal{A}\interpret{\mathrm{car}(\mathbf{op})}$ is an interpretation of $\mathbf{op}$.
		Here, $\mathcal{A}\interpret{-}$ is the interpretation of types defined in Definition~\ref{def:interpretation-type}.
		\item $\answerobj \in \category{C}$ has a bottom element $\bot^{\answerobj} : 1 \to \answerobj$.
	\end{itemize}
\end{definition}
Note that the definition above is simplified compared to \cite{kura2023}: we don't require a pseudo-lifting strong monad $T$ on $\category{C}$ and an EM $T$-algebra on $\answerobj$.
This is because we unified effect-free constants and modal operators as basic operators.
If we have a strong monad from $({-})_{\bot}$ to $T$, then an EM $T$-algebra on $\answerobj$ induces an EM $({-})_{\bot}$-algebra on $\answerobj$, and both the uniform $T$-fixed-point operator and the uniform $({-})_{\bot}$-fixed-point operator give the same interpretation of fixed points.

\begin{definition}[interpretation of types]\label{def:interpretation-type}
	For each type $\sigma$, its interpretation $\mathcal{A}\interpret{\sigma}$ is defined as follows.
	\begin{gather}
		\mathcal{A}\interpret{b} \coloneqq A b \qquad
		\mathcal{A}\interpret{\answertype} \coloneqq \answerobj \qquad
		\mathcal{A}\interpret{0} \coloneqq 0 \qquad
		\mathcal{A}\interpret{\sigma + \tau} \coloneqq \mathcal{A}\interpret{\sigma} + \mathcal{A}\interpret{\tau} \\
		\mathcal{A}\interpret{1} \coloneqq 1 \qquad
		\mathcal{A}\interpret{\sigma \times \tau} \coloneqq \mathcal{A}\interpret{\sigma} \times \mathcal{A}\interpret{\tau} \qquad
		\mathcal{A}\interpret{\sigma \to \answertype} \coloneqq \exponential{\mathcal{A}\interpret{\sigma}}{\mathcal{A}\interpret{\answertype}}
	\end{gather}
\end{definition}

\begin{definition}[interpretation of contexts]
	The interpretation $\mathcal{A}\interpret{\Gamma}$ of a context is defined as follows.
	\[ \mathcal{A}\interpret{\cdot} \coloneqq 1 \qquad \mathcal{A}\interpret{\Gamma, x : \sigma} \coloneqq \mathcal{A}\interpret{\Gamma} \times \mathcal{A}\interpret{\sigma} \]
\end{definition}

\begin{definition}[interpretation of terms]
	For any well-typed term $\Gamma \vdash M : \sigma$, its interpretation $\mathcal{A}\interpret{M} = \mathcal{A}\interpret{\Gamma \vdash M : \sigma} : \mathcal{A}\interpret{\Gamma} \to \mathcal{A}\interpret{\sigma}$ is defined as follows.
	\begin{gather}
		\mathcal{A} \interpret{\Gamma, x : \sigma \vdash y : \tau} \coloneqq \begin{cases}
			\mathcal{A} \interpret{\Gamma \vdash y : \tau} \comp \pi_1 & x \neq y \\
			\pi_2 & x = y
		\end{cases} \\
		\mathcal{A} \interpret{\mathbf{op}\ M} \coloneqq a(\mathbf{op}) \comp \mathcal{A} \interpret{M} \qquad
		\mathcal{A} \interpret{()} \coloneqq {!} \qquad
		\mathcal{A} \interpret{(M_1, M_2)} \coloneqq \tupling{\mathcal{A} \interpret{M_1}}{\mathcal{A} \interpret{M_2}} \\
		\mathcal{A} \interpret{\pi_i\ M} \coloneqq \pi_i \comp \mathcal{A} \interpret{M} \qquad
		\mathcal{A} \interpret{\delta(M)} \coloneqq {?} \comp \mathcal{A} \interpret{M} \qquad
		\mathcal{A} \interpret{\iota_i\ M} \coloneqq \iota_i \comp \mathcal{A} \interpret{M} \\
		\mathcal{A} \interpret{\caseexpr{M}{x_1}{M_1}{x_2}{M_2}} \coloneqq [\mathcal{A} \interpret{M_1}, \mathcal{A} \interpret{M_2}] \comp [\identity{} \times \iota_1, \identity{} \times \iota_2]^{-1} \comp \tupling{\identity{}}{\mathcal{A} \interpret{M}} \\
		\mathcal{A} \interpret{\lambda x. M} \coloneqq \Lambda (\mathcal{A} \interpret{M}) \qquad
		\mathcal{A} \interpret{M\ N} \coloneqq \mathbf{ev} \comp \tupling{\mathcal{A} \interpret{M}}{\mathcal{A} \interpret{N}} \qquad
		\mathcal{A} \interpret{\fixpoint{f}{M}} \coloneqq (\mathcal{A} \interpret{M})^{\dagger}
	\end{gather}
\end{definition}


\section{Details of Refinement Type System}\label{sec:detail-refinement}

\subsection{Formulas}

\subsubsection{Well-formed formulas}
We define \emph{well-formed formulas} $\Gamma \vdash \phi$ as follows.
\begin{mathpar}
	\inferrule{
		\mathbf{ap} : \sigma \\
		\Gamma \vdash M : \sigma
	}{
		\Gamma \vdash \mathbf{ap}(M)
	}
	\and
	\inferrule{ }{
		\Gamma \vdash \top
	}
	\and
	\inferrule{ }{
		\Gamma \vdash \bot
	}
	\and
	\inferrule{
		\Gamma \vdash \psi \\
		\Gamma \vdash \phi
	}{
		\Gamma \vdash \psi \implies \phi
	}
	\and
	\inferrule{
		\Gamma \vdash \psi \\
		\Gamma \vdash \phi
	}{
		\Gamma \vdash \psi \land \phi
	}
	\and
	\inferrule{
		\Gamma \vdash \psi \\
		\Gamma \vdash \phi
	}{
		\Gamma \vdash \psi \lor \phi
	}
\end{mathpar}

\subsubsection{Semantics}
We assume that readers are familiar with fibrations.
See~\cite{} for the introduction to fibrations.

Notations:
Let $p : \category{E} \to \category{B}$ be a fibration.
\begin{itemize}
	\item $\category{E}_I$ is the fibre category over $I \in \category{B}$.
	\item The reindexing functor along $u : I \to J$ is denoted by $u^{*} : \category{E}_J \to \category{E}_I$.
	\item A cartesian lifting over $u : I \to J$ is denoted by $\overline{u}(Y) : u^{*} Y \to Y$ where $Y \in \category{E}_J$.
\end{itemize}

We define a fibration $p : \category{P} \to \category{C}$ by the change-of-base construction.
\begin{center}
	\begin{tikzcd}
		\category{P} \ar[d, "p"] \ar[r] \ar[rd, phantom, very near start, "\lrcorner"] & \mathbf{Sub}(\Set) \ar[d, "\mathbf{sub}_{\category{C}}"] \\
		\category{C} \ar[r, "{\category{C}(1, {-})}"] & \Set
	\end{tikzcd}
\end{center}
An object in $\category{P}$ is a pair $(I, P)$ such that $I \in \category{C}$ and $P \subseteq \category{C}(1, I)$ and a morphism $f : (I, P) \dotTo (J, Q)$ is a morphism $f : I \to J$ such that for any $x \in P$, $f \comp x \in Q$.
We sometimes omit $I$ and write $P = (I, P)$ when $I$ is clear from the context.

\begin{lemma}
	$p : \category{P} \to \category{C}$ is a bifibration with a fibred bi-cc structure and simple products/coproducts.
\end{lemma}
\begin{proof}
	The subobject fibration $\mathbf{Sub}(\Set) \to \Set$ is a bifibration with a fibred bi-cc structure and simple products/coproducts, and $U$ preserves finite products (see~\cite[4.3.1]{hermida1993} and \cite[Lemma~1.9.14]{jacobs2001}).
\end{proof}
Notation:
\begin{itemize}
	\item Fibred initial/terminal object functor are $\bot, \top : \category{C} \to \category{P}$.
	\item A fibred product, a fibred coproduct, and a fibred exponential for $X, Y \in \category{P}_I$ are denoted by $X \land Y$, $X \lor Y$, $X \implies Y$.
	\item For simple products and coproducts, we write $\forall, \exists : \category{P}_{X \times Y} \to \category{P}_{X}$.
\end{itemize}

\begin{definition}[interpretation of formulas]
	Let $\mathcal{P}$ be an interpretation of atomic predicates such that $\mathcal{P}(\mathbf{ap}) \in \category{P}_{\mathcal{A}\interpret{\mathrm{ar}(\mathbf{ap})}}$ (i.e.\ $\mathcal{P}(\mathbf{ap}) \subseteq \category{C}(1, \mathcal{A}\interpret{\mathrm{ar}(\mathbf{ap})})$) for any atomic predicate $\mathbf{ap} : \mathrm{ar}(\mathbf{ap}) \in \AtomicPreds$.
	For any well-formed formula $\Gamma \vdash \phi$, we define $\mathcal{A}^{\mathcal{P}}\interpret{\Gamma \vdash \phi}$ as follows.
	\begin{gather}
		\mathcal{A}^{\mathcal{P}}\interpret{\Gamma \vdash \mathbf{ap}(M)} \coloneqq (\mathcal{A}\interpret{M})^{*} \mathcal{P}(\mathbf{ap}) = \{ \gamma : 1 \to \mathcal{A}\interpret{\Gamma} \mid \mathcal{A}\interpret{M} \comp \gamma \in \mathcal{P}(\mathbf{ap}) \} \\
		\mathcal{A}^{\mathcal{P}}\interpret{\Gamma \vdash \top} \coloneqq \top \mathcal{A}\interpret{\Gamma} \qquad
		\mathcal{A}^{\mathcal{P}}\interpret{\Gamma \vdash \bot} \coloneqq \bot \mathcal{A}\interpret{\Gamma} \\
		\mathcal{A}^{\mathcal{P}}\interpret{\Gamma \vdash \psi \land \phi} \coloneqq \mathcal{A}\interpret{\Gamma \vdash \psi} \land \mathcal{A}\interpret{\Gamma \vdash \phi} \qquad
		\mathcal{A}^{\mathcal{P}}\interpret{\Gamma \vdash \psi \lor \phi} \coloneqq \mathcal{A}\interpret{\Gamma \vdash \psi} \lor \mathcal{A}\interpret{\Gamma \vdash \phi} \\
		\mathcal{A}^{\mathcal{P}}\interpret{\Gamma \vdash \psi \implies \phi} \coloneqq \mathcal{A}\interpret{\Gamma \vdash \psi} \implies \mathcal{A}\interpret{\Gamma \vdash \phi}
	\end{gather}
\end{definition}

\subsection{Well-Formed Contexts/Types}

We define well-formed contexts $\vdash \dot{\Gamma}$ and well-formed types $\dot{\Gamma} \vdash \dot{\sigma}$ inductively so that any formula in $\dot{\Gamma}$ and $\dot{\sigma}$ is well-formed.
\begin{mathpar}
	\inferrule{ }{
		\vdash \cdot
	}
	\and
	\inferrule{
		\dot{\Gamma} \vdash \dot{\sigma}
	}{
		\vdash \dot{\Gamma}, x : \dot{\sigma}
	}
\end{mathpar}

\begin{mathpar}
	\inferrule{
		\vdash \dot{\Gamma} \\
		\underlying{\dot{\Gamma}}, v : \sigma \vdash \phi \\
		\sigma \in \BaseTypes \cup \{ \answertype, 1 \}
	}{
		\dot{\Gamma} \vdash \{ v : \sigma \mid \phi \}
	}
	\and
	\inferrule{
		\dot{\Gamma}, x : \dot{\sigma} \vdash \{ v : \answertype \mid \phi \}
	}{
		\dot{\Gamma} \vdash (x : \dot{\sigma}) \to \{ v : \answertype \mid \phi \}
	}
	\and
	\inferrule{
		\dot{\Gamma}, x : \dot{\sigma} \vdash \dot{\tau}
	}{
		\dot{\Gamma} \vdash (x : \dot{\sigma}) \times \dot{\tau}
	}
	\and
	\inferrule{ }{
		\dot{\Gamma} \vdash 0
	}
	\and
	\inferrule{
		\dot{\Gamma} \vdash \dot{\sigma} \\
		\dot{\Gamma} \vdash \dot{\tau}
	}{
		\dot{\Gamma} \vdash \dot{\sigma} + \dot{\tau}
	}
\end{mathpar}

\subsection{Semantics of Well-Formed Contexts/Types}

Let $s_{\category{C}} : s(\category{C}) \to \category{C}$ be the simple fibration on $\category{C}$ (see~\cite{} for the definition).
Since $\category{C}$ is cartesian closed, $s_{\category{C}} : s(\category{C}) \to \category{C}$ is a CCompC~\cite[Theorem~10.5.5]{jacobs2001}.
Since $\category{C}$ is a bi-ccc, $s_{\category{C}} : s(\category{C}) \to \category{C}$ has strong fibred coproducts.

More concretely, $s_{\category{C}} : s(\category{C}) \to \category{C}$ has the following structures.
\begin{itemize}
	\item An object in $s(\category{C})$ is a pair $(I, X)$ where $I, X \in \category{C}$.
	\item A morphism in $s(\category{C})$ is a pair $(u, f) : (I, X) \to (J, Y)$ where $u : I \to J$ and $f : I \times X \to Y$.
	\item Cartesian liftings are given as follows.
	\begin{center}
		\begin{tikzcd}
			& u^{*}(J, X) = (I, X) \ar[r, "{(u, \pi_2)}"] \ar[d, swap, "{u^{*}(\identity{}, f) = (\identity{}, f \comp (u \times \identity{}))}"] & (J, X) \ar[d, "{(\identity{}, f)}"] \\
			s(\category{C}) \ar[d, "s_{\category{C}}"] & u^{*}(J, Y) = (I, Y) \ar[r, "{(u, \pi_2)}"] & (J, Y) \\
			\category{C} & I \ar[r, "u"] & J
		\end{tikzcd}
	\end{center}
	\item Fibred terminal object functor $1 : \category{C} \to s(\category{C})$.
	\[ 1 I = (I, 1) \]
	\item Comprehension functor $\{ {-} \} : s(\category{C}) \to \category{C}$.
	\[ \{ (I, X) \} = I \times X \qquad \category{C}(I, \{ (J, Y) \}) \cong s_{\category{C}}(1 I, (J, Y)) \]
	\item Product $\prod_{(I, X)} : (s_{\category{C}})_{I \times X} \to (s_{\category{C}})_{I}$.
	\[ \prod_{(I, X)} (I \times X, Y) = (I, \exponential{X}{Y}) \qquad \prod_{(I, X)} (\identity{I \times X}, f) = (\identity{}, \Lambda (f \comp \tupling{\pi_1 \times \identity{}}{\eval \comp (\pi_2 \times \identity{})})) \]
	\[ (s_{\category{C}})_{I \times X}(\pi_1^{*} (I, Z), (I \times X, Y)) \cong (s_{\category{C}})_{I}((I, Z), \prod_{(I, X)} (I \times X, Y)) \]
	\[ \eta = \Lambda(\pi_2 \comp \pi_1) : (I, Y) \to (I, \exponential{X}{Y}) \]
	\[ \epsilon = (\identity{}, \eval \comp \braiding \comp (\pi_2 \times \identity{})) : (I \times X, \exponential{X}{Y}) \dotTo (I \times X, Y) \]
	\item Coproduct
	\[ \coprod_{(I, X)} (I \times X, Y) = (I, X \times Y) \qquad \coprod_{(I, X)} (\identity{I \times X}, f) = (\identity{I}, ) \]
	\[ (s_{\category{C}})_{I \times X}((I \times X, Y), \pi_1^{*} (I, Z)) \cong (s_{\category{C}})_{I}(\coprod_{(I, X)} (I \times X, Y), (I, Z)) \]
	\[ \eta = (\identity{}, \pi_2 \times \identity{}) : (I \times X, Y) \to (I \times X, X \times Y) \qquad \epsilon = (\identity{}, \pi_2 \comp \pi_2) : (I, X \times Y) \to (I, Y) \]
	\[ \kappa = \associator : (I \times X) \times Y \to I \times (X \times Y) \qquad \mathbf{fst} = (\identity{}, \pi_1 \comp \pi_2) : (I, X \times Y) \to (I, X) \]
	\item Strong fibred initial objects $0 : \category{C} \to s(\category{C})$.
	\[ 0 I = (I, 0) \]
	\[ {?}_{(I, X)} = (\identity{}, {?} \comp \pi_2) : (I, 0) \to (I, X) \]
	\item Strong fibred binary coproducts
	\[ (I, X) + (I, Y) = (I, X + Y) \]
	\[ \iota_1 = (\identity{}, \iota_1 \comp \pi_2) : (I, X) \to (I, X + Y) \qquad \iota_2 = (\identity{}, \iota_2 \comp \pi_2) : (I, X) \to (I, X + Y) \]
	The functor
	\[ \tupling{(\identity{} \times \iota_1)^{*}}{(\identity{} \times \iota_2)^{*}} : s(\category{C})_{I \times (X + Y)} \to s(\category{C})_{I \times X} \times s(\category{C})_{I \times Y} \]
	is fully faithful.
\end{itemize}

\newcommand{\refinetotal}{\{ s(\category{C}) \mid \category{P} \}}
\newcommand{\refinefib}{\{ s_{\category{C}} \mid p \}}

By applying the construction in~\cite{kura2021}, we obtain a model of a dependent refinement type system.
\begin{lemma}
	We have a SCCompC $\refinefib : \refinetotal \to \category{P}$ with strong fibred coproducts.
	We also have the following morphism of SCCompCs.
	\begin{center}
		\begin{tikzcd}
			\refinetotal \ar[r, "u"] \ar[d, "\refinefib"] & s(\category{C}) \ar[d, "s_{\category{C}}"] \\
			\category{P} \ar[r, "p"] & \category{C}
		\end{tikzcd}
	\end{center}
\end{lemma}
\begin{proof}
	The side condition for the existence of strong (binary) coproducts is satisfied as follows.
	\begin{itemize}
		\item $s(\category{C}) \to \category{C}$ is a CCompC with strong fibred coproducts.
		\item For any $u : I \to I'$ in $\category{C}$ and $(I', X), (I', Y) \in s(\category{C})_{I'}$, the following two squares are pullbacks.
		\begin{center}
			\begin{tikzcd}
				I \times X \ar[r, "\identity{} \times \iota_1"] \ar[d, "u \times \identity{}"] & I \times (X + Y) \ar[d, "u \times \identity{}"] & I \times Y \ar[l, "\identity{} \times \iota_2"] \ar[d, "u \times \identity{}"] \\
				I' \times X \ar[r, "\identity{} \times \iota_1"] & I' \times (X + Y) & I' \times Y \ar[l, "\identity{} \times \iota_2"]
			\end{tikzcd}
		\end{center}
		\item $p : \category{P} \to \category{C}$ is a fibred bi-ccc and thus satisfies Frobenius by~\cite[Remark~4.5.7]{hermida1993}.
		\item $p : \category{P} \to \category{C}$ is a cofibration.
	\end{itemize}
\end{proof}

More concretely, $\refinefib : \refinetotal \to \category{P}$ has the following structures.
\begin{itemize}
	\item An object $((I, X), P, Q) \in \{ s(\category{C}) \mid \category{P} \}$ consists of $(I, X) \in s(\category{C})$, $P \subseteq \category{C}(1, I)$, and $Q \subseteq \category{C}(1, I \times X)$ such that if $\tupling{i}{x} \in Q$, then $i \in P$.
	\item A morphism in $\{ s(\category{C}) \mid \category{P} \}$ (denoted by $(u, f) : ((I, X), P, Q) \dotTo ((J, Y), R, S)$) is a morphism $(u, f) : (I, X) \to (J, Y)$ in $s(\category{C})$ such that $u : P \dotTo R$ and $\tupling{u \comp \pi_1}{f} : Q \dotTo S$ in $\category{P}$, that is, for any $i \in P$, $u \comp i \in R$ and for any $\tupling{i}{x} \in Q$, $\tupling{u \comp i}{f \comp \tupling{i}{x}} \in R$.
	\item Reindexing functor:
	\begin{center}
		\begin{tikzcd}
			((I, X), P, \pi_{(I, X)}^{*} P \land (u \times \identity{})^{*} R) \ar[r, "{(u, \pi_2)}"] \ar[d, "{(\identity{}, f \comp (u \times \identity{}))}"] & ((J, X), Q, R) \ar[d, "{(\identity{}, f)}"] \\
			((I, Y), P, \pi_{(I, Y)}^{*} P \land (u \times \identity{})^{*} S) \ar[r, "{(u, \pi_2)}"] & ((J, Y), Q, S) \\
			(I, P) \ar[r, "u"] & (J, Q)
		\end{tikzcd}
	\end{center}
	\item Terminal object functor
	\[ \dot{1} (I, P) = ((I, 1), P, \pi_1^{*} P) \]
	where $\pi_1^{*} P = \{ \tupling{i}{\identity{}} : 1 \to I \times 1 \mid i \in P \}$
	\item Comprehension functor
	\[ \{ ((I, X), P, Q) \} = (I \times X, Q) \]
	\item Product
	\[ \dot{\prod}_{((I, X), P, Q)} (I \times X, Y, Q, R) = ((I, \exponential{X}{Y}), P, \text{(deferred to \eqref{eq:prod_pred})}) \]
	where the last component is
	\begin{align}
		&\{ \tupling{i}{f} : 1 \to I \times (\exponential{X}{Y}) \mid i \in P \land \\
		&\qquad \forall x : 1 \to X, \tupling{i}{x} \in Q \implies \tupling{\tupling{i}{x}}{\eval \comp \tupling{f}{x}} \in R \} \label{eq:prod_pred}
	\end{align}
	\item Coproduct
	\[ \dot{\coprod}_{((I, X), P, Q)} (I \times X, Y, Q, R) = ((I, X \times Y), P, (\associator^{-1})^{*} R) \]
	where $(\associator^{-1})^{*} R = \{ \tupling{i}{\tupling{x}{y}} : 1 \to I \times (X \times Y) \mid \tupling{\tupling{i}{x}}{y} \in R \}$
	\item Fibred binary coproduct
	\begin{align}
		&((I, X), P, Q) \dotPlus ((I, Y), P, R) \\
		&= ((I, X + Y), P, \{ (\identity{I} \times \fstcoproj) \comp f \mid f \in Q \} \cup \{ (\identity{I} \times \sndcoproj) \comp g \mid g \in R \})
	\end{align}
	\item Fibred initial object
	\[ \dot{0} (I, P) = ((I, 0), P, \emptyset) \]
\end{itemize}
Note that morphisms in $\refinetotal$ and $\category{P}$ are denoted by $(u, f) : ((I, X), P, Q) \dotTo ((J, Y), R, S)$ and $u : (I, P) \dotTo (J, R)$, which indicate that they are unique morphisms over $(u, f) : (I, X) \dotTo (J, Y)$ in $s(\category{C})$ and $u : I \dotTo J$ in $\category{C}$, respectively.

\begin{definition}[interpretation of well-formed contexts/types]
	A well-formed context $\vdash \dot{\Gamma}$ is interpreted as $\mathcal{A}^{\mathcal{P}}\interpret{\dot{\Gamma}} \in \category{P}_{\mathcal{A}\interpret{\underlying{\dot{\Gamma}}}}$.
	\begin{gather}
		\interpret{\cdot} \coloneqq \top 1 \in \category{P}_{1} \qquad
		\interpret{\dot{\Gamma}, x : \dot{\sigma}} \coloneqq \{ \interpret{\dot{\Gamma} \vdash \dot{\sigma}} \} \in \category{P}_{\mathcal{A}\interpret{\underlying{\dot{\Gamma}}} \times \mathcal{A}\interpret{\underlying{\dot{\sigma}}}}
	\end{gather}
	A well-formed type $\dot{\Gamma} \vdash \dot{\sigma}$ is interpreted as $\mathcal{A}^{\mathcal{P}}\interpret{\dot{\Gamma} \vdash \dot{\sigma}} \in \{ s(\category{C}) \mid \category{P} \}_{\mathcal{A}^{\mathcal{P}}\interpret{\dot{\Gamma}}}$ such that $u \mathcal{A}^{\mathcal{P}}\interpret{\dot{\Gamma} \vdash \dot{\sigma}} = (\mathcal{A}\interpret{\underlying{\dot{\Gamma}}}, \mathcal{A}\interpret{\underlying{\dot{\sigma}}})$.
	\begin{gather}
		\mathcal{A}^{\mathcal{P}}\interpret{\dot{\Gamma} \vdash \{ v : \sigma \mid \phi \}} \coloneqq ((\mathcal{A}\interpret{\underlying{\dot{\Gamma}}}, \mathcal{A}\interpret{\sigma}), \mathcal{A}^{\mathcal{P}}\interpret{\dot{\Gamma}}, \pi_1^{*} \mathcal{A}^{\mathcal{P}}\interpret{\dot{\Gamma}} \land \mathcal{A}^{\mathcal{P}}\interpret{\underlying{\dot{\Gamma}}, v : \sigma \vdash \phi}) \\
		\mathcal{A}^{\mathcal{P}}\interpret{\dot{\Gamma} \vdash (x : \dot{\sigma}) \to \{ v : \answertype \mid \phi \}} \coloneqq \dot{\prod}_{\mathcal{A}^{\mathcal{P}}\interpret{\dot{\Gamma} \vdash \dot{\sigma}}} \mathcal{A}^{\mathcal{P}}\interpret{\dot{\Gamma}, x : \dot{\sigma} \vdash \{ v : \answertype \mid \phi \}} \\
		\mathcal{A}^{\mathcal{P}}\interpret{\dot{\Gamma} \vdash (x : \dot{\sigma}) \times \dot{\tau}} \coloneqq \dot{\coprod}_{\mathcal{A}^{\mathcal{P}}\interpret{\dot{\Gamma} \vdash \dot{\sigma}}} \mathcal{A}^{\mathcal{P}}\interpret{\dot{\Gamma}, x : \dot{\sigma} \vdash \dot{\tau}} \\
		\mathcal{A}^{\mathcal{P}}\interpret{\dot{\Gamma} \vdash 0} \coloneqq \dot{0} \mathcal{A}^{\mathcal{P}}\interpret{\dot{\Gamma}} \qquad
		\mathcal{A}^{\mathcal{P}}\interpret{\dot{\Gamma} \vdash \dot{\sigma} + \dot{\tau}} \coloneqq \mathcal{A}^{\mathcal{P}}\interpret{\dot{\Gamma} \vdash \dot{\sigma}} \dotPlus \mathcal{A}^{\mathcal{P}}\interpret{\dot{\Gamma} \vdash \dot{\tau}}
	\end{gather}
\end{definition}

\subsection{Subtyping relation}

We define $\dot{\Gamma} \mid v : \sigma. \phi \vDash \psi$ by
\[ \pi_1^{*} \mathcal{A}^{\mathcal{P}}\interpret{\dot{\Gamma}} \land \mathcal{A}^{\mathcal{P}}\interpret{\underlying{\dot{\Gamma}}, v : \sigma \vdash \phi} \le \mathcal{A}^{\mathcal{P}}\interpret{\underlying{\dot{\Gamma}}, v : \sigma \vdash \psi}. \]
Intuitively, this means that if formulas in $\dot{\Gamma}$ and $\phi$ are true, then $\psi$ is also true.

We define a subtyping relation $\dot{\Gamma} \vdash \dot{\sigma} <: \dot{\tau}$.
\begin{mathpar}
	\inferrule[Sub-Refl]{
		\dot{\Gamma} \vdash \dot{\sigma}
	}{
		\dot{\Gamma} \vdash \dot{\sigma} <: \dot{\sigma}
	}
	\and
	\inferrule[Sub-Trans]{
		\dot{\Gamma} \vdash \dot{\sigma}_1 <: \dot{\sigma}_2 \\
		\dot{\Gamma} \vdash \dot{\sigma}_2 <: \dot{\sigma}_3
	}{
		\dot{\Gamma} \vdash \dot{\sigma}_1 <: \dot{\sigma}_3
	}
	\and
	\inferrule[Sub-Refine]{
		\dot{\Gamma} \vdash v : \sigma. \phi \vDash \psi \\
		\sigma \in \BaseTypes \cup \{ \answertype, 1 \}
	}{
		\dot{\Gamma} \vdash \{ v : \sigma \mid \phi \} <: \{ v : \sigma \mid \psi \}
	}
	\and
	\inferrule[Sub-DFun]{
		\dot{\Gamma} \vdash \dot{\sigma}_2 <: \dot{\sigma}_1 \\
		\dot{\Gamma}, x : \dot{\sigma}_1 \vdash \tau_1 \\
		\dot{\Gamma}, x : \dot{\sigma}_2 \vdash \dot{\tau}_1 <: \dot{\tau}_2
	}{
		\dot{\Gamma} \vdash (x : \dot{\sigma}_1) \to \dot{\tau}_1 <: (x : \dot{\sigma}_2) \to \dot{\tau}_2
	}
	\and
	\inferrule[Sub-DPair]{
		\dot{\Gamma} \vdash \dot{\sigma}_1 <: \dot{\sigma}_2 \\
		\dot{\Gamma}, x : \dot{\sigma}_2 \vdash \tau_2 \\
		\dot{\Gamma}, x : \dot{\sigma}_1 \vdash \dot{\tau}_1 <: \dot{\tau}_2
	}{
		\dot{\Gamma} \vdash (x : \dot{\sigma}_1) \times \dot{\tau}_1 <: (x : \dot{\sigma}_2) \times \dot{\tau}_2
	}
	\and
	\inferrule[Sub-Sum]{
		\dot{\Gamma} \vdash \dot{\sigma}_1 <: \dot{\sigma}_2 \\
		\dot{\Gamma} \vdash \dot{\tau}_1 <: \dot{\tau}_2
	}{
		\dot{\Gamma} \vdash \dot{\sigma}_1 + \dot{\tau}_1 <: \dot{\sigma}_2 + \dot{\tau}_2
	}
\end{mathpar}

\subsection{Typing rules}
A well-typed term $\dot{\Gamma} \vdash M : \dot{\sigma}$ is defined as follows.
\begin{mathpar}
	\inferrule[R-Sub]{
		\dot{\Gamma} \vdash M : \dot{\sigma} \\
		\dot{\Gamma} \vdash \dot{\sigma} <: \dot{\tau}
	}{
		\dot{\Gamma} \vdash M : \dot{\tau}
	}
	\and
	\inferrule[R-App]{
		\dot{\Gamma} \vdash M : (x : \dot{\sigma}) \to \{ v : \answertype \mid \phi \} \\
		\dot{\Gamma} \vdash N : \dot{\sigma}
	}{
		\dot{\Gamma} \vdash M\ N : \{ v : \answertype \mid \phi[N/x] \}
	}
	\and
	\inferrule[R-Abs]{
		\dot{\Gamma}, x : \dot{\sigma} \vdash M : \{ v : \answertype \mid \phi \}
	}{
		\dot{\Gamma} \vdash \lambda x : \dot{\sigma}. M : (x : \dot{\sigma}) \to \{ v : \answertype \mid \phi \}
	}
	\and
	\inferrule[R-Unit]{
		\vdash \dot{\Gamma}
	}{
		\dot{\Gamma} \vdash () : \{ v : 1 \mid \top \}
	}
	\and
	\inferrule[R-Pair]{
		\dot{\Gamma} \vdash M : \dot{\sigma} \\
		\dot{\Gamma} \vdash N : \dot{\tau}[M / x]
	}{
		\dot{\Gamma} \vdash (M, N) : (x : \dot{\sigma}) \times \dot{\tau}
	}
	\and
	\inferrule[R-Fst]{
		\dot{\Gamma} \vdash M : (x : \dot{\sigma}) \times \dot{\tau}
	}{
		\dot{\Gamma} \vdash \pi_1\ M : \dot{\sigma}
	}
	\and
	\inferrule[R-Snd]{
		\dot{\Gamma} \vdash M : (x : \dot{\sigma}) \times \dot{\tau}
	}{
		\dot{\Gamma} \vdash \pi_2\ M : \dot{\tau}[\pi_1\ M / x]
	}
	\and
	\inferrule[R-Case0]{
		\dot{\Gamma} \vdash M : 0 \\
		\dot{\Gamma} \vdash \dot{\tau}
	}{
		\dot{\Gamma} \vdash \delta(M) : \dot{\tau}
	}
	\and
	\inferrule[R-Inj]{
		\dot{\Gamma} \vdash M : \dot{\sigma}_i \\
		\dot{\Gamma} \vdash \dot{\sigma}_{3-i}
	}{
		\dot{\Gamma} \vdash \iota_i\ M : \dot{\sigma}_1 + \dot{\sigma}_2
	}
	\and
	\inferrule[R-Case2]{
		\dot{\Gamma}, z : \dot{\sigma}_1 + \dot{\sigma}_2 \vdash \dot{\tau} \\
		\dot{\Gamma} \vdash M : \dot{\sigma}_1 + \dot{\sigma}_2 \\
		\dot{\Gamma}, x_1 : \dot{\sigma}_1 \vdash N_1 : \dot{\tau}[\iota_1\ x_1/z] \\
		\dot{\Gamma}, x_2 : \dot{\sigma}_2 \vdash N_2 : \dot{\tau}[\iota_2\ x_2/z]
	}{
		\dot{\Gamma} \vdash \caseexpr{M}{x_1}{N_1}{x_2}{N_2} : \dot{\tau}[M/z]
	}
	\and
	\inferrule[R-VarRefine]{
		\vdash \dot{\Gamma} \\
		(x : \{ v : b \mid \phi \}) \in \dot{\Gamma}
	}{
		\dot{\Gamma} \vdash x : \{ v : b \mid v = x \}
	}
	\and
	\inferrule[R-Var]{
		\vdash \dot{\Gamma} \\
		(x : \dot{\sigma}) \in \dot{\Gamma}
	}{
		\dot{\Gamma} \vdash x : \dot{\sigma}
	}
	\and
	\inferrule[R-Fix]{
		\dot{\Gamma}, f : (x : \dot{\sigma}) \to \{ v : \answertype \mid \phi \} \vdash M : (x : \dot{\sigma}) \to \{ v : \answertype \mid \phi \} \\
		\underlying{\dot{\Gamma}}, x : \underlying{\dot{\sigma}}, v : \answertype \vdash \phi \\
		\text{$\phi$ is admissible at $v$}
	}{
		\dot{\Gamma} \vdash \fixpoint{f}{M} : (x : \dot{\sigma}) \to \{ v : \answertype \mid \phi \}
	}
	\and
	\inferrule[R-BasicOp]{
		\dot{\Gamma} \vdash M : \dot{\sigma} \\
		\dot{\Gamma} \vdash \dot{\tau} \\
		\mathrm{ar}(\mathbf{op}) = \underlying{\dot{\sigma}} \\
		\mathrm{car}(\mathbf{op}) = \underlying{\dot{\tau}} \\
		(\identity{} \times a(\mathbf{op})) : \mathcal{A}^{\mathcal{P}}\interpret{\dot{\Gamma}, v : \dot{\sigma}} \dotTo \mathcal{A}^{\mathcal{P}}\interpret{\dot{\Gamma}, v : \dot{\tau}}
	}{
		\dot{\Gamma} \vdash \mathbf{op}(M) : \dot{\tau}
	}
\end{mathpar}
\begin{mathpar}
	\inferrule[R-BasicConst]{
		\mathbf{op} : 1 \to \tau \\
		\tau \in \BaseTypes \cup \{ \answertype \}
	}{
		\dot{\Gamma} \vdash \mathbf{op} : \{ v : \tau \mid v =_{\tau} \mathbf{op} \}
	}
	\and
	\inferrule[R-BasicSimple]{
		\mathbf{op} : \sigma_1 \times \dots \times \sigma_n \to \tau \\
		\sigma_1, \dots, \sigma_n, \tau \in \BaseTypes \cup \{ \answertype \} \\
		\dot{\Gamma} \vdash M : \{ (v_1, \dots, v_n) : \sigma_1 \times \dots \times \sigma_n \mid \phi[\mathbf{op}(v_1, \dots, v_n) / v] \}
	}{
		\dot{\Gamma} \vdash \mathbf{op}(M) : \{ v : \tau \mid \phi \}
	}
	\and
	\inferrule[R-BasicBool]{
		\mathbf{op} : \sigma_1 \times \dots \times \sigma_n \to 1 + 1 \\
		\sigma_1, \dots, \sigma_n \in \BaseTypes \cup \{ \answertype \} \\
		v_1 : \sigma_1, \dots, v_n : \sigma_n \vdash \psi_t \\
		\mathcal{A}^{\mathcal{P}}\interpret{(v_1, \dots, v_n) : \sigma_1 \times \dots \times \sigma_n \vdash \psi_t} \subseteq \pi_2^{*} a(\mathbf{op})^{*} \{ \iota_1 \} \\
		v_1 : \sigma_1, \dots, v_n : \sigma_n \vdash \psi_f \\
		\mathcal{A}^{\mathcal{P}}\interpret{(v_1, \dots, v_n) : \sigma_1 \times \dots \times \sigma_n \vdash \psi_f} \subseteq \pi_2^{*} a(\mathbf{op})^{*} \{ \iota_2 \} \\
		\dot{\Gamma} \vdash M : \{ (v_1, \dots, v_n) : \sigma_1 \times \dots \times \sigma_n \mid \psi_t \land \phi_t[()/v] \lor \psi_f \land \phi_f[()/v] \}
	}{
		\dot{\Gamma} \vdash \mathbf{op}(M) : \{ v : 1 \mid \phi_t \} + \{ v : 1 \mid \phi_f \}
	}
\end{mathpar}

\section{Details of Proofs}\label{sec:detail-proof}

\subsection{Preparation}
\begin{definition}
	Let $\Gamma_1, \Gamma_2$ be contexts and $\sigma$ be a type.
	We define $\mathrm{proj}_{\Gamma_1; \sigma; \Gamma_2} : \mathcal{A}\interpret{\Gamma_1, x : \sigma, \Gamma_2} \to \mathcal{A}\interpret{\Gamma_1, \Gamma_2}$ as follows.
	\[ \mathrm{proj}_{\Gamma_1; \sigma; \cdot} \coloneqq \pi_1 \qquad \mathrm{proj}_{\Gamma_1; \sigma; \Gamma_2, y : \tau} \coloneqq \mathrm{proj}_{\Gamma_1; \sigma; \Gamma_2} \times \identity{} \]
\end{definition}

\begin{lemma}[weakening for terms]\label{lem:weakening-term}
	For any well-typed term $\Gamma_1, \Gamma_2 \vdash M : \tau$ and type $\sigma$,
	\[ \mathcal{A}\interpret{\Gamma_1, \Gamma_2 \vdash M : \tau} \comp \mathrm{proj}_{\Gamma_1; \sigma; \Gamma_2} = \mathcal{A}\interpret{\Gamma_1, x : \sigma, \Gamma_2 \vdash M : \tau} \]
\end{lemma}

\begin{lemma}[weakening for formulas]\label{lem:weakening-formula}
	For any well-formed $\Gamma_1, \Gamma_2 \vdash \phi$ and type $\sigma$,
	\[ \mathrm{proj}_{\Gamma_1; \sigma; \Gamma_2}^{*} \mathcal{A}^{\mathcal{P}}\interpret{\Gamma_1, \Gamma_2 \vdash \phi} = \mathcal{A}^{\mathcal{P}}\interpret{\Gamma_1, x : \sigma, \Gamma_2 \vdash \phi} \]
\end{lemma}
\begin{proof}
	By induction.
	\begin{itemize}
		\item For basic predicates including $\le_{\answertype}$ and $=_b$,
		\begin{align}
			\mathrm{proj}_{\Gamma_1; \sigma; \Gamma_2}^{*} \mathcal{A}^{\mathcal{P}}\interpret{\Gamma_1, \Gamma_2 \vdash a(M)} &= \mathrm{proj}_{\Gamma_1; \sigma; \Gamma_2}^{*} (\mathcal{A}\interpret{\Gamma_1, \Gamma_2 \vdash M : \tau})^{*} \mathcal{P}(a) \\
			&= (\mathcal{A}\interpret{\Gamma_1, x : \sigma, \Gamma_2 \vdash M : \tau})^{*} \mathcal{P}(a) \\
			&= \mathcal{A}^{\mathcal{P}}\interpret{\Gamma_1, x : \sigma, \Gamma_2 \vdash a(M)}
		\end{align}
		by Lemma~\ref{lem:weakening-term}.
		\item For $\bot, \top, \lor, \land, \implies$, we use the fact that fibred bi-cc structures are preserved by reindexing.
	\end{itemize}
\end{proof}

\begin{lemma}[weakening for types]\label{lem:weakening-type}
	Let $\dot{\Gamma}_1 \vdash \dot{\sigma}$ be a well-formed type.
	\begin{itemize}
		\item For any well-formed context $\vdash \dot{\Gamma}_1, \dot{\Gamma}_2$, we have
		\[ \mathrm{proj}_{\underlying{\dot{\Gamma}_1}; \underlying{\dot{\sigma}}; \underlying{\dot{\Gamma}_2}} : \mathcal{A}^{\mathcal{P}}\interpret{\dot{\Gamma}_1, x : \dot{\sigma}, \dot{\Gamma}_2} \dotTo \mathcal{A}^{\mathcal{P}}\interpret{\dot{\Gamma}_1, \dot{\Gamma}_2}. \]
		\item For any well-formed type $\dot{\Gamma}_1, \dot{\Gamma}_2 \vdash \dot{\tau}$, we have
		\[ \mathrm{proj}_{\underlying{\dot{\Gamma}_1}; \underlying{\dot{\sigma}}; \underlying{\dot{\Gamma}_2}}^{*} \mathcal{A}^{\mathcal{P}}\interpret{\dot{\Gamma}_1, \dot{\Gamma}_2 \vdash \dot{\tau}} = \mathcal{A}^{\mathcal{P}}\interpret{\dot{\Gamma}_1, x : \dot{\sigma}, \dot{\Gamma}_2 \vdash \dot{\tau}}. \]
	\end{itemize}
\end{lemma}
\begin{proof}
	By simultaneous induction on $\vdash \dot{\Gamma}_1, \dot{\Gamma}_2$ and $\dot{\Gamma}_1, \dot{\Gamma}_2 \vdash \dot{\tau}$.
	\begin{itemize}
		\item If $\dot{\Gamma}_2 = \cdot$, we have
		\[ \mathrm{proj}_{\underlying{\dot{\Gamma}_1}; \underlying{\dot{\sigma}}; \cdot} : \mathcal{A}^{\mathcal{P}}\interpret{\dot{\Gamma}_1, x : \dot{\sigma}} \dotTo \mathcal{A}^{\mathcal{P}}\interpret{\dot{\Gamma}_1} \]
		because $\mathrm{proj}_{\underlying{\dot{\Gamma}_1}; \underlying{\dot{\sigma}}; \cdot} = \pi_1 = \pi_{\mathcal{A}^{\mathcal{P}}\interpret{\dot{\Gamma}_1 \vdash \dot{\sigma}}}$.
		\item If $\dot{\Gamma}_2 = \dot{\Gamma}'_2, y : \dot{\sigma}'$, we have
		\begin{center}
			\begin{tikzcd}
				\mathcal{A}^{\mathcal{P}}\interpret{\dot{\Gamma}_1, x : \dot{\sigma}, \dot{\Gamma}'_2, y : \dot{\sigma}'} \ar[d, equal] \\
				\{ \mathrm{proj}_{\underlying{\dot{\Gamma}_1}; \underlying{\dot{\sigma}}; \underlying{\dot{\Gamma}_2}}^{*} \mathcal{A}^{\mathcal{P}}\interpret{\dot{\Gamma}_1, \dot{\Gamma}'_2 \vdash \dot{\sigma}'} \} \ar[d, "\{ \overline{\mathrm{proj}_{\underlying{\dot{\Gamma}_1}; \underlying{\dot{\sigma}}; \underlying{\dot{\Gamma}'_2}}}(\dots) \}"] \\
				\mathcal{A}^{\mathcal{P}}\interpret{\dot{\Gamma}_1, \dot{\Gamma}'_2, y : \dot{\sigma}'}
			\end{tikzcd}
		\end{center}
		\begin{align}
			\{ \overline{\mathrm{proj}_{\underlying{\dot{\Gamma}_1}; \underlying{\dot{\sigma}}; \underlying{\dot{\Gamma}_2}}}(\dots) \} &= \mathrm{proj}_{\underlying{\dot{\Gamma}_1}; \underlying{\dot{\sigma}}; \underlying{\dot{\Gamma}_2}} \times \identity{} \\
			&= \mathrm{proj}_{\underlying{\dot{\Gamma}_1}; \underlying{\dot{\sigma}}; \underlying{\dot{\Gamma}_2}, y : \underlying{\dot{\sigma}'}}
		\end{align}
		\item $\dot{\Gamma}_1, \dot{\Gamma}_2 \vdash \{ v : \tau \mid \phi \}$:
		By definition of the interpretation,
		\begin{align}
			&\mathrm{proj}_{\underlying{\dot{\Gamma}_1}; \underlying{\dot{\sigma}}; \underlying{\dot{\Gamma}_2}}^{*} \mathcal{A}^{\mathcal{P}}\interpret{\dot{\Gamma}_1, \dot{\Gamma}_2 \vdash \{ v : \tau \mid \phi \}} \\
			&= \mathrm{proj}_{\underlying{\dot{\Gamma}_1}; \underlying{\dot{\sigma}}; \underlying{\dot{\Gamma}_2}}^{*} ((\mathcal{A}\interpret{\underlying{\dot{\Gamma}_1, \dot{\Gamma}_2}}, \mathcal{A}\interpret{\tau}), \mathcal{A}^{\mathcal{P}}\interpret{\dot{\Gamma}_1, \dot{\Gamma}_2}, \pi^{*} \mathcal{A}^{\mathcal{P}}\interpret{\dot{\Gamma}_1, \dot{\Gamma}_2} \land \mathcal{A}^{\mathcal{P}}\interpret{\underlying{\dot{\Gamma}_1, \dot{\Gamma}_2}, v : \tau \vdash \phi}) \\
			&= ((\mathcal{A}\interpret{\underlying{\dot{\Gamma}_1, x : \dot{\sigma}, \dot{\Gamma}_2}}, \mathcal{A}\interpret{\tau}), \mathcal{A}^{\mathcal{P}}\interpret{\dot{\Gamma}_1, x : \dot{\sigma}, \dot{\Gamma}_2}, \\
			&\qquad \pi^{*} \mathcal{A}^{\mathcal{P}}\interpret{\dot{\Gamma}_1, x : \dot{\sigma}, \dot{\Gamma}_2} \land (\mathrm{proj}_{\underlying{\dot{\Gamma}_1}; \underlying{\dot{\sigma}}; \underlying{\dot{\Gamma}_2}} \times \identity{})^{*} \mathcal{A}^{\mathcal{P}}\interpret{\underlying{\dot{\Gamma}_1, \dot{\Gamma}_2}, v : \tau \vdash \phi})
		\end{align}
		and by Lemma~\ref{lem:weakening-formula},
		\begin{align}
			&(\mathrm{proj}_{\underlying{\dot{\Gamma}_1}; \underlying{\dot{\sigma}}; \underlying{\dot{\Gamma}_2}} \times \identity{})^{*} \mathcal{A}^{\mathcal{P}}\interpret{\underlying{\dot{\Gamma}_1, \dot{\Gamma}_2}, v : \tau \vdash \phi} \\
			&= \mathrm{proj}_{\underlying{\dot{\Gamma}_1}; \underlying{\dot{\sigma}}; \underlying{\dot{\Gamma}_2}, v : \tau}^{*} \mathcal{A}^{\mathcal{P}}\interpret{\underlying{\dot{\Gamma}_1, \dot{\Gamma}_2}, v : \tau \vdash \phi} \\
			&= \mathcal{A}^{\mathcal{P}}\interpret{\underlying{\dot{\Gamma}_1, x : \dot{\sigma}, \dot{\Gamma}_2}, v : \tau \vdash \phi}
		\end{align}
		\item $\dot{\Gamma}_1, \dot{\Gamma}_2 \vdash (x : \dot{\tau}_1) \to \{ v : \tau \mid \phi \}$ and $\dot{\Gamma}_1, \dot{\Gamma}_2 \vdash (x : \dot{\tau}_1) \times \dot{\tau}_2$: By the BC condition.
		\item $\dot{\Gamma}_1, \dot{\Gamma}_2 \vdash 0$ and $\dot{\Gamma}_1, \dot{\Gamma}_2 \vdash \dot{\tau}_1 + \dot{\tau}_2$: Because reindexing functors preserve fibred coproducts.
	\end{itemize}
\end{proof}

\begin{definition}
	Let $\Gamma_1, \Gamma_2$ be contexts and $\Gamma_1 \vdash M : \sigma$ be a well-typed term.
	We define $\mathrm{subst}_{M; \Gamma_2} : \mathcal{A}\interpret{\Gamma_1, \Gamma_2} \to \mathcal{A}\interpret{\Gamma_1, x : \sigma, \Gamma_2}$ as follows.
	\[ \mathrm{subst}_{M; \cdot} \coloneqq \tupling{\identity{}}{\mathcal{A}\interpret{M}} \qquad \mathrm{subst}_{M; \Gamma_2, y : \tau} \coloneqq \mathrm{subst}_{M; \Gamma_2} \times \identity{} \]
\end{definition}

\begin{lemma}
	For any well-typed terms $\Gamma_1, x : \sigma, \Gamma_2 \vdash N$ and $\Gamma_1 \vdash M : \sigma$, we have
	\[ \mathcal{A}\interpret{\Gamma_1, x : \sigma, \Gamma_2 \vdash N} \comp \mathrm{subst}_{M; \Gamma_2} = \mathcal{A}\interpret{\Gamma_1, \Gamma_2 \vdash N[M/x]} \]
\end{lemma}
\begin{proof}
	By induction on $N$.
\end{proof}

\begin{lemma}[substitution in formulas]\label{lem:subst-formula}
	For any well-formed $\Gamma_1, x : \sigma, \Gamma_2 \vdash \phi$ and $\Gamma_1 \vdash M : \sigma$,
	\[ \mathcal{A}^{\mathcal{P}}\interpret{\Gamma_1, \Gamma_2 \vdash \phi[M/x]} = \mathrm{subst}_{M; \Gamma_2}^{*} \mathcal{A}^{\mathcal{P}}\interpret{\Gamma_1, x : \sigma, \Gamma_2 \vdash \phi} \]
\end{lemma}
\begin{proof}
	Similarly to Lemma~\ref{lem:weakening-formula}
	\begin{itemize}
		\item For basic predicates including $\le_{\answertype}$ and $=_b$,
		\begin{align}
			&\mathrm{subst}_{M; \Gamma_2}^{*}  \mathcal{A}^{\mathcal{P}}\interpret{\Gamma_1, x : \sigma, \Gamma_2 \vdash a(N)} \\
			&= \mathrm{subst}_{M; \Gamma_2}^{*} (\mathcal{A}\interpret{\Gamma_1. x : \sigma, \Gamma_2 \vdash N : \tau})^{*} \mathcal{P}(a) \\
			&= (\mathcal{A}\interpret{\Gamma_1, \Gamma_2 \vdash N[M/x] : \tau})^{*} \mathcal{P}(a) \\
			&= \mathcal{A}^{\mathcal{P}}\interpret{\Gamma_1, \Gamma_2 \vdash a(N)[M/x]}
		\end{align}
		by Lemma~\ref{lem:weakening-term}.
		\item For $\bot, \top, \lor, \land, \implies$, we use the fact that fibred bi-cc structures are preserved by reindexing.
	\end{itemize}
\end{proof}

\begin{lemma}[substitution in types]\label{lem:subst-type}
	For any well-formed type $\dot{\Gamma}_1 \vdash \dot{\sigma}$ and well-typed term $\underlying{\dot{\Gamma}_1} \vdash M : \underlying{\dot{\sigma}}$, if $\tupling{\identity{}}{\mathcal{A}\interpret{M}} : \mathcal{A}^{\mathcal{P}}\interpret{\dot{\Gamma}_1} \to \mathcal{A}^{\mathcal{P}}\interpret{\dot{\Gamma}_1, x : \dot{\sigma}}$ in $\category{P}$, then
	\begin{itemize}
		\item For any well-formed context $\vdash \dot{\Gamma}_1, \dot{\Gamma}_2$,
		\[ \mathrm{subst}_{M; \underlying{\dot{\Gamma}_2}} : \mathcal{A}^{\mathcal{P}}\interpret{\dot{\Gamma}_1, \dot{\Gamma}_2[M/x]} \dotTo \mathcal{A}^{\mathcal{P}}\interpret{\dot{\Gamma}_1, x : \dot{\sigma}, \dot{\Gamma}_2} \qquad \text{in $\category{P}$} \]
		\item For any well-formed type $\dot{\Gamma}_1, x : \dot{\sigma}, \dot{\Gamma}_2 \vdash \dot{\tau}$,
		\[ \mathrm{subst}_{M; \underlying{\dot{\Gamma}_2}}^{*} \mathcal{A}^{\mathcal{P}}\interpret{\dot{\Gamma}_1, x : \dot{\sigma}, \dot{\Gamma}_2 \vdash \dot{\tau}} = \mathcal{A}^{\mathcal{P}}\interpret{\dot{\Gamma}_1, \dot{\Gamma}_2[M/x] \vdash \dot{\tau}[M/x]} \]
	\end{itemize}
\end{lemma}
\begin{proof}
	By simultaneous induction on $\vdash \dot{\Gamma}_1, \dot{\Gamma}_2$ and $\dot{\Gamma}_1, x : \dot{\sigma}, \dot{\Gamma}_2 \vdash \dot{\tau}$.
	\begin{itemize}
		\item If $\dot{\Gamma}_2 = \cdot$, then we have $\mathrm{subst}_{M; \cdot} : \mathcal{A}^{\mathcal{P}}\interpret{\dot{\Gamma}_1} \dotTo \mathcal{A}^{\mathcal{P}}\interpret{\dot{\Gamma}_1, x : \dot{\sigma}}$ by assumption.
		\item If $\dot{\Gamma}_2 = \dot{\Gamma}'_2, y : \dot{\tau}$, then
		\begin{center}
			\begin{tikzcd}
				\mathcal{A}^{\mathcal{P}}\interpret{\dot{\Gamma}_1, \dot{\Gamma}'_2[M/x], y : \dot{\tau}[M/x]} \ar[d, equal] \\
				\{ \mathrm{subst}_{M; \underlying{\dot{\Gamma}'_2}}^{*} \mathcal{A}^{\mathcal{P}}\interpret{\dot{\Gamma}_1, x : \dot{\sigma}, \dot{\Gamma}'_2 \vdash \dot{\tau}} \} \ar[d, "\{ \overline{\mathrm{subst}_{M; \underlying{\dot{\Gamma}'_2}}}(\dots) \}"] \\
				\mathcal{A}^{\mathcal{P}}\interpret{\dot{\Gamma}_1, x : \dot{\sigma}, \dot{\Gamma}'_2, y : \dot{\tau}}
			\end{tikzcd}
		\end{center}
		\[ \{ \overline{\mathrm{subst}_{M; \underlying{\dot{\Gamma}'_2}}}(\dots) \} = \mathrm{subst}_{M; \underlying{\dot{\Gamma}'_2}} \times \identity{} = \mathrm{subst}_{M; \underlying{\dot{\Gamma}'_2, y : \dot{\tau}}} \]
		\item $\{ v : \tau \mid \phi \}$: By the definition of reindexing in $\refinefib : \refinetotal \to \category{P}$.
		\begin{align}
			&\mathrm{subst}_{M; \underlying{\dot{\Gamma}_2}}^{*} \mathcal{A}^{\mathcal{P}}\interpret{\dot{\Gamma}_1, x : \dot{\sigma}, \dot{\Gamma}_2 \vdash \{ v : \tau \mid \phi \}} \\
			&= \mathrm{subst}_{M; \underlying{\dot{\Gamma}_2}}^{*} ((\mathcal{A}\interpret{\underlying{\dot{\Gamma}_1, x : \dot{\sigma}, \dot{\Gamma}_2}}, \mathcal{A}\interpret{\tau}), \mathcal{A}^{\mathcal{P}}\interpret{\dot{\Gamma}_1, x : \dot{\sigma}, \dot{\Gamma}_2}, \\
			&\qquad \pi_1^{*} \mathcal{A}^{\mathcal{P}}\interpret{\dot{\Gamma}_1, x : \dot{\sigma}, \dot{\Gamma}_2} \land \mathcal{A}^{\mathcal{P}}\interpret{\underlying{\dot{\Gamma}_1}, x : \underlying{\dot{\sigma}}, \underlying{\dot{\Gamma}_2}, v : \tau \vdash \phi}) \\
			&= ((\mathcal{A}\interpret{\underlying{\dot{\Gamma}_1, \dot{\Gamma}_2[M/x]}}, \mathcal{A}\interpret{\tau}), \mathcal{A}^{\mathcal{P}}\interpret{\dot{\Gamma}_1, \dot{\Gamma}_2[M/x]}, \\
			&\qquad \pi_1^{*} \mathcal{A}^{\mathcal{P}}\interpret{\dot{\Gamma}_1, \dot{\Gamma}_2[M/x]} \land (\mathrm{subst}_{M; \underlying{\dot{\Gamma}_2}} \times \identity{})^{*}(\pi_1^{*} \mathcal{A}^{\mathcal{P}}\interpret{\dot{\Gamma}_1, x : \dot{\sigma}, \dot{\Gamma}_2} \land \mathcal{A}^{\mathcal{P}}\interpret{\underlying{\dot{\Gamma}_1}, x : \underlying{\dot{\sigma}}, \underlying{\dot{\Gamma}_2}, v : \tau \vdash \phi})) \\
			&\pi_1^{*} \mathcal{A}^{\mathcal{P}}\interpret{\dot{\Gamma}_1, \dot{\Gamma}_2[M/x]} \land (\mathrm{subst}_{M; \underlying{\dot{\Gamma}_2}} \times \identity{})^{*}(\pi_1^{*} \mathcal{A}^{\mathcal{P}}\interpret{\dot{\Gamma}_1, x : \dot{\sigma}, \dot{\Gamma}_2} \land \mathcal{A}^{\mathcal{P}}\interpret{\underlying{\dot{\Gamma}_1}, x : \underlying{\dot{\sigma}}, \underlying{\dot{\Gamma}_2}, v : \tau \vdash \phi}) \\
			&= \pi_1^{*} \mathcal{A}^{\mathcal{P}}\interpret{\dot{\Gamma}_1, \dot{\Gamma}_2[M/x]} \land \pi_1^{*} \mathrm{subst}_{M; \underlying{\dot{\Gamma}_2}}^{*} \mathcal{A}^{\mathcal{P}}\interpret{\dot{\Gamma}_1, x : \dot{\sigma}, \dot{\Gamma}_2} \land \mathrm{subst}_{M; \underlying{\dot{\Gamma}_2}, v : \tau}^{*} \mathcal{A}^{\mathcal{P}}\interpret{\underlying{\dot{\Gamma}_1}, x : \underlying{\dot{\sigma}}, \underlying{\dot{\Gamma}_2}, v : \tau \vdash \phi} \\
			&= \pi_1^{*} \mathcal{A}^{\mathcal{P}}\interpret{\dot{\Gamma}_1, \dot{\Gamma}_2[M/x]} \land \mathcal{A}^{\mathcal{P}}\interpret{\underlying{\dot{\Gamma}_1}, \underlying{\dot{\Gamma}_2}, v : \tau \vdash \phi[M/x]}
		\end{align}
		Here, we used Lemma~\ref{lem:subst-formula} and $\underlying{\dot{\Gamma}_2[M/x]} = \underlying{\dot{\Gamma}_2}$.
		\item $(y : \dot{\tau}) \to \{ v : \answertype \mid \phi \}$: Apply the BC condition.
		\begin{center}
			\begin{tikzcd}[column sep=6em]
				\mathcal{A}^{\mathcal{P}}\interpret{\dot{\Gamma}_1, \dot{\Gamma}_2[M/x], y : \dot{\tau}[M/x]} \ar[d, "\pi"] \ar[r, "\mathrm{subst}_{M; \underlying{\dot{\Gamma}_2}, y : \underlying{\dot{\tau}}}"] \ar[rd, phantom, very near start, "\lrcorner"] & \mathcal{A}^{\mathcal{P}}\interpret{\dot{\Gamma}_1, x : \dot{\sigma}, \dot{\Gamma}_2, y : \dot{\tau}} \ar[d, "\pi"] \\
				\mathcal{A}^{\mathcal{P}}\interpret{\dot{\Gamma}_1, \dot{\Gamma}_2[M/x]} \ar[r, "\mathrm{subst}_{M; \underlying{\dot{\Gamma}_2}}"] & \mathcal{A}^{\mathcal{P}}\interpret{\dot{\Gamma}_1, x : \dot{\sigma}, \dot{\Gamma}_2}
			\end{tikzcd}
		\end{center}
		\begin{align}
			&\mathrm{subst}_{M; \underlying{\dot{\Gamma}_2}}^{*} \mathcal{A}^{\mathcal{P}}\interpret{\dot{\Gamma}_1, x : \dot{\sigma}, \dot{\Gamma}_2 \vdash (y : \dot{\tau}) \to \{ v : \answertype \mid \phi \}} \\
			&= \mathrm{subst}_{M; \underlying{\dot{\Gamma}_2}}^{*} \prod_{\mathcal{A}^{\mathcal{P}}\interpret{\dot{\Gamma}_1, x : \dot{\sigma}, \dot{\Gamma}_2 \vdash \dot{\tau}}} \mathcal{A}^{\mathcal{P}}\interpret{\dot{\Gamma}_1, x : \dot{\sigma}, \dot{\Gamma}_2, y : \dot{\tau} \vdash \{ v : \answertype \mid \phi \}} \\
			&= \prod_{\mathcal{A}^{\mathcal{P}}\interpret{\dot{\Gamma}_1, \dot{\Gamma}_2[M/x] \vdash \dot{\tau}[M/x]}} \mathrm{subst}_{M; \underlying{\dot{\Gamma}_2}, y : \underlying{\dot{\tau}}}^{*} \mathcal{A}^{\mathcal{P}}\interpret{\dot{\Gamma}_1, x : \dot{\sigma}, \dot{\Gamma}_2, y : \dot{\tau} \vdash \{ v : \answertype \mid \phi \}} \\
			&= \prod_{\mathcal{A}^{\mathcal{P}}\interpret{\dot{\Gamma}_1, \dot{\Gamma}_2[M/x] \vdash \dot{\tau}[M/x]}} \mathcal{A}^{\mathcal{P}}\interpret{\dot{\Gamma}_1, \dot{\Gamma}_2[M/x], y : \dot{\tau}[M/x] \vdash \{ v : \answertype \mid \phi[M/x] \}} \\
			&= \mathcal{A}^{\mathcal{P}}\interpret{\dot{\Gamma}_1, \dot{\Gamma}_2[M/x] \vdash (y : \dot{\tau}[M/x]) \to \{ v : \answertype \mid \phi[M/x] \}}
		\end{align}
		\item $(y : \dot{\tau}_1) \times \dot{\tau}_2$: Apply the BC condition (almost the same as above).
		\item $0$: Fibred initial objects are preserved by reindexing functors.
		\begin{align}
			\mathrm{subst}_{M; \underlying{\dot{\Gamma}_2}}^{*} \mathcal{A}^{\mathcal{P}}\interpret{\dot{\Gamma}_1, x : \dot{\sigma}, \dot{\Gamma}_2 \vdash 0} &= \mathrm{subst}_{M; \underlying{\dot{\Gamma}_2}}^{*} 0 \mathcal{A}^{\mathcal{P}}\interpret{\dot{\Gamma}_1, x : \dot{\sigma}, \dot{\Gamma}_2} \\
			&= 0 \mathcal{A}^{\mathcal{P}}\interpret{\dot{\Gamma}_1, \dot{\Gamma}_2[M/x]} \\
			&= \mathcal{A}^{\mathcal{P}}\interpret{\dot{\Gamma}_1, \dot{\Gamma}_2[M/x] \vdash 0}
		\end{align}
		\item $\dot{\tau}_1 + \dot{\tau}_2$: Fibred binary coproducts are preserved by reindexing functors.
		\begin{align}
			&\mathrm{subst}_{M; \underlying{\dot{\Gamma}_2}}^{*} \mathcal{A}^{\mathcal{P}}\interpret{\dot{\Gamma}_1, x : \dot{\sigma}, \dot{\Gamma}_2 \vdash \dot{\tau}_1 + \dot{\tau}_2} \\
			&= \mathrm{subst}_{M; \underlying{\dot{\Gamma}_2}}^{*} (\mathcal{A}^{\mathcal{P}}\interpret{\dot{\Gamma}_1, x : \dot{\sigma}, \dot{\Gamma}_2 \vdash \dot{\tau}_1} + \mathcal{A}^{\mathcal{P}}\interpret{\dot{\Gamma}_1, x : \dot{\sigma}, \dot{\Gamma}_2 \vdash \dot{\tau}_2}) \\
			&= \mathrm{subst}_{M; \underlying{\dot{\Gamma}_2}}^{*} \mathcal{A}^{\mathcal{P}}\interpret{\dot{\Gamma}_1, x : \dot{\sigma}, \dot{\Gamma}_2 \vdash \dot{\tau}_1} + \mathrm{subst}_{M; \underlying{\dot{\Gamma}_2}}^{*} \mathcal{A}^{\mathcal{P}}\interpret{\dot{\Gamma}_1, x : \dot{\sigma}, \dot{\Gamma}_2 \vdash \dot{\tau}_2} \\
			&= \mathcal{A}^{\mathcal{P}}\interpret{\dot{\Gamma}_1, \dot{\Gamma}_2[M/x] \vdash \dot{\tau}_1[M/x]} + \mathcal{A}^{\mathcal{P}}\interpret{\dot{\Gamma}_1, \dot{\Gamma}_2[M/x] \vdash \dot{\tau}_2[M/x]} \\
			&= \mathcal{A}^{\mathcal{P}}\interpret{\dot{\Gamma}_1, \dot{\Gamma}_2[M/x] \vdash (\dot{\tau}_1 + \dot{\tau}_2)[M/x]}
		\end{align}
	\end{itemize}
\end{proof}

\subsection{Proofs for Subtyping Rules}
\begin{proposition}[soundness of subtyping]\label{prop:subtype-sound}
	For any $\dot{\Gamma} \vdash \dot{\sigma} <: \dot{\tau}$, we have
	\[ \mathcal{A}^{\mathcal{P}}\interpret{\dot{\Gamma} \vdash \dot{\sigma}} \le \mathcal{A}^{\mathcal{P}}\interpret{\dot{\Gamma} \vdash \dot{\tau}} \]
	where we define $((I, X), P, Q) \le ((I, X), P, Q')$ if $Q \subseteq Q'$, or equivalently, $\identity{(I, X)} : ((I, X), P, Q) \dotTo ((I, X), P, Q')$.
\end{proposition}
\begin{proof}
	By induction on derivation of $\dot{\Gamma} \vdash \dot{\sigma} <: \dot{\tau}$.
	\begin{itemize}
		\item Sub-Refl, Sub-Trans: By reflexivity and transitivity of $\subseteq$.
		\item Sub-Refine: Obvious from the definition of the interpretation.
		\item Sub-Sum: Easy.
		\item Sub-DFun:
		By IH, we have
		\begin{gather}
			\mathcal{A}^{\mathcal{P}}\interpret{\dot{\Gamma} \vdash \dot{\sigma}_2} \le \mathcal{A}^{\mathcal{P}}\interpret{\dot{\Gamma} \vdash \dot{\sigma}_1} \\
			\mathcal{A}^{\mathcal{P}}\interpret{\dot{\Gamma}, x : \dot{\sigma}_2 \vdash \dot{\tau}_1} \le \mathcal{A}^{\mathcal{P}}\interpret{\dot{\Gamma}, x : \dot{\sigma}_2 \vdash \dot{\tau}_2}
		\end{gather}
		We apply the BC condition.
		\begin{center}
			\begin{tikzcd}
				\refinetotal_{\{ \mathcal{A}^{\mathcal{P}}\interpret{\dot{\Gamma} \vdash \dot{\sigma}_1} \}} \ar[r, "\prod"] \ar[d, "\{ \identity{} \}^{*}"] & \refinetotal_{\mathcal{A}^{\mathcal{P}}\interpret{\dot{\Gamma}}} \ar[d, equal] \\
				\refinetotal_{\{ \mathcal{A}^{\mathcal{P}}\interpret{\dot{\Gamma} \vdash \dot{\sigma}_2} \}} \ar[r, "\prod"] & \refinetotal_{\mathcal{A}^{\mathcal{P}}\interpret{\dot{\Gamma}}}
			\end{tikzcd}
		\end{center}
		\begin{align}
			\mathcal{A}^{\mathcal{P}}\interpret{\dot{\Gamma} \vdash (x : \dot{\sigma}_1) \to \dot{\tau}_1} &= \prod_{\mathcal{A}^{\mathcal{P}}\interpret{\dot{\Gamma} \vdash \dot{\sigma}_1}} \mathcal{A}^{\mathcal{P}}\interpret{\dot{\Gamma}, x : \dot{\sigma}_1 \vdash \dot{\tau}_1} \\
			&= \prod_{\mathcal{A}^{\mathcal{P}}\interpret{\dot{\Gamma} \vdash \dot{\sigma}_2}} \{ \identity{} \}^{*} \mathcal{A}^{\mathcal{P}}\interpret{\dot{\Gamma}, x : \dot{\sigma}_1 \vdash \dot{\tau}_1} \\
			&= \prod_{\mathcal{A}^{\mathcal{P}}\interpret{\dot{\Gamma} \vdash \dot{\sigma}_2}} \mathcal{A}^{\mathcal{P}}\interpret{\dot{\Gamma}, x : \dot{\sigma}_2 \vdash \dot{\tau}_1} \\
			&\le \prod_{\mathcal{A}^{\mathcal{P}}\interpret{\dot{\Gamma} \vdash \dot{\sigma}_2}} \mathcal{A}^{\mathcal{P}}\interpret{\dot{\Gamma}, x : \dot{\sigma}_2 \vdash \dot{\tau}_2} \\
			&= \mathcal{A}^{\mathcal{P}}\interpret{\dot{\Gamma} \vdash (x : \dot{\sigma}_2) \to \dot{\tau}_2}
		\end{align}
		\item Sub-DPair:
		By IH, we have
		\begin{gather}
			\mathcal{A}^{\mathcal{P}}\interpret{\dot{\Gamma} \vdash \dot{\sigma}_1} \le \mathcal{A}^{\mathcal{P}}\interpret{\dot{\Gamma} \vdash \dot{\sigma}_2} \\
			\mathcal{A}^{\mathcal{P}}\interpret{\dot{\Gamma}, x : \dot{\sigma}_1 \vdash \dot{\tau}_1} \le \mathcal{A}^{\mathcal{P}}\interpret{\dot{\Gamma}, x : \dot{\sigma}_1 \vdash \dot{\tau}_2}
		\end{gather}
		We apply the BC condition.
		\begin{align}
			\mathcal{A}^{\mathcal{P}}\interpret{\dot{\Gamma} \vdash (x : \dot{\sigma}_1) \times \dot{\tau}_1} &= \coprod_{\mathcal{A}^{\mathcal{P}}\interpret{\dot{\Gamma} \vdash \dot{\sigma}_1}} \mathcal{A}^{\mathcal{P}}\interpret{\dot{\Gamma}, x : \dot{\sigma}_1 \vdash \dot{\tau}_1} \\
			&\le \coprod_{\mathcal{A}^{\mathcal{P}}\interpret{\dot{\Gamma} \vdash \dot{\sigma}_1}} \mathcal{A}^{\mathcal{P}}\interpret{\dot{\Gamma}, x : \dot{\sigma}_1 \vdash \dot{\tau}_2} \\
			&= \coprod_{\mathcal{A}^{\mathcal{P}}\interpret{\dot{\Gamma} \vdash \dot{\sigma}_1}} \{ \identity{} \}^{*} \mathcal{A}^{\mathcal{P}}\interpret{\dot{\Gamma}, x : \dot{\sigma}_2 \vdash \dot{\tau}_2} \\
			&= \coprod_{\mathcal{A}^{\mathcal{P}}\interpret{\dot{\Gamma} \vdash \dot{\sigma}_2}} \mathcal{A}^{\mathcal{P}}\interpret{\dot{\Gamma}, x : \dot{\sigma}_2 \vdash \dot{\tau}_2} \\
			&= \mathcal{A}^{\mathcal{P}}\interpret{\dot{\Gamma} \vdash (x : \dot{\sigma}_2) \times \dot{\tau}_2}
		\end{align}
	\end{itemize}
\end{proof}

\subsection{Proofs for Typing Rules}
\begin{theorem}[soundness]\label{thm:soundness-detailed}
	For any well-typed term $\dot{\Gamma} \vdash M : \dot{\sigma}$,
	\[ (\identity{}, \mathcal{A}\interpret{M} \comp \pi_1) : 1 \mathcal{A}^{\mathcal{P}}\interpret{\dot{\Gamma}} \dotTo \mathcal{A}^{\mathcal{P}}\interpret{\dot{\Gamma} \vdash \dot{\sigma}} \]
	Note that this is equivalent to the following.
	\[ \tupling{\identity{}}{\mathcal{A}\interpret{M}} : \mathcal{A}^{\mathcal{P}}\interpret{\dot{\Gamma}} \dotTo \{ \mathcal{A}^{\mathcal{P}}\interpret{\dot{\Gamma} \vdash \dot{\sigma}} \} \]
\end{theorem}
\begin{proof}
	By induction on derivation of $\dot{\Gamma} \vdash M : \dot{\sigma}$.
	The basic idea of the proof is to interpret terms in $\refinefib : \refinetotal \to \category{P}$ as terms of a dependent type system.
	Proofs for individual cases are given below.
	\begin{itemize}
		\item \textsc{R-Sub}: Lemma~\ref{lem:R-Sub-sound}
		\item \textsc{R-App}: Lemma~\ref{lem:R-App-sound}
		\item \textsc{R-Abs}: Lemma~\ref{lem:R-Abs-sound}
		\item \textsc{R-Unit}: Lemma~\ref{lem:R-Unit-sound}
		\item \textsc{R-Pair}: Lemma~\ref{lem:R-Pair-sound}
		\item \textsc{R-Fst}: Lemma~\ref{lem:R-Fst-sound}
		\item \textsc{R-Snd}: Lemma~\ref{lem:R-Snd-sound}
		\item \textsc{R-Case0}: Lemma~\ref{lem:R-Case0-sound}
		\item \textsc{R-Inj}: Lemma~\ref{lem:R-Inj-sound}
		\item \textsc{R-Case2}: Lemma~\ref{lem:R-Case2-sound}
		\item \textsc{R-Var}: Lemma~\ref{lem:R-Var-sound}
		\item \textsc{R-VarRefine}: Lemma~\ref{lem:R-VarRefine-sound}
		\item \textsc{R-Fix}: Lemma~\ref{lem:R-Fix-sound}
		\item \textsc{R-BasicOp}: Lemma~\ref{lem:R-BasicOp-sound}
	\end{itemize}
\end{proof}
\begin{lemma}\label{lem:R-Sub-sound}
	\textsc{R-Sub} is sound.
\end{lemma}
\begin{proof}
	By IH and Proposition~\ref{prop:subtype-sound}, we have
	\begin{gather}
		(\identity{}, \mathcal{A}\interpret{M} \comp \pi_1) : 1 \mathcal{A}^{\mathcal{P}}\interpret{\dot{\Gamma}} \dotTo \mathcal{A}^{\mathcal{P}}\interpret{\dot{\Gamma} \vdash \dot{\sigma}} \\
		\mathcal{A}^{\mathcal{P}}\interpret{\dot{\Gamma} \vdash \dot{\sigma}} \le \mathcal{A}^{\mathcal{P}}\interpret{\dot{\Gamma} \vdash \dot{\tau}}.
	\end{gather}
	Therefore, we have
	\[ (\identity{}, \mathcal{A}\interpret{M} \comp \pi_1) : 1 \mathcal{A}^{\mathcal{P}}\interpret{\dot{\Gamma}} \dotTo \mathcal{A}^{\mathcal{P}}\interpret{\dot{\Gamma} \vdash \dot{\tau}}. \]
\end{proof}

\begin{lemma}\label{lem:R-App-sound}
	\textsc{R-App} is sound.
\end{lemma}
\begin{proof}
	By IH, we have
	\begin{gather}
		(\identity{}, \mathcal{A}\interpret{M} \comp \pi_1) : 1 \mathcal{A}^{\mathcal{P}}\interpret{\dot{\Gamma}} \dotTo \mathcal{A}^{\mathcal{P}}\interpret{\dot{\Gamma} \vdash (x : \dot{\sigma}) \to \{ v : \answertype \mid \phi \}} \\
		(\identity{}, \mathcal{A}\interpret{N} \comp \pi_1) : 1 \mathcal{A}^{\mathcal{P}}\interpret{\dot{\Gamma}} \dotTo \mathcal{A}^{\mathcal{P}}\interpret{\dot{\Gamma} \vdash \dot{\sigma}}
	\end{gather}
	We consider the following vertical morphism in $\refinetotal$.
	\begin{center}
		\begin{tikzcd}
			1 \mathcal{A}^{\mathcal{P}}\interpret{\dot{\Gamma}} \ar[d, "{(\identity{}, \mathcal{A}\interpret{M} \comp \pi_1)}"] \\
			\prod_{\mathcal{A}^{\mathcal{P}}\interpret{\dot{\Gamma} \vdash \dot{\sigma}}} \mathcal{A}^{\mathcal{P}}\interpret{\dot{\Gamma}, x : \dot{\sigma} \vdash \{ v : \answertype \mid \phi \}} \ar[d, equal] \\
			\tupling{\identity{}}{\mathcal{A}\interpret{N}}^{*} \pi_{\mathcal{A}^{\mathcal{P}}\interpret{\dot{\Gamma} \vdash \dot{\sigma}}}^{*} \prod_{\mathcal{A}^{\mathcal{P}}\interpret{\dot{\Gamma} \vdash \dot{\sigma}}} \mathcal{A}^{\mathcal{P}}\interpret{\dot{\Gamma}, x : \dot{\sigma} \vdash \{ v : \answertype \mid \phi \}} \ar[d, "{\tupling{\identity{}}{\mathcal{A}\interpret{N}}^{*} \epsilon^{\pi^{*} \dashv \prod}}"] \\
			\tupling{\identity{}}{\mathcal{A}\interpret{N}}^{*} \mathcal{A}^{\mathcal{P}}\interpret{\dot{\Gamma}, x : \dot{\sigma} \vdash \{ v : \answertype \mid \phi \}} \ar[d, equal, "\text{by Lemma~\ref{lem:subst-type}}"] \\
			\mathcal{A}^{\mathcal{P}}\interpret{\dot{\Gamma} \vdash \{ v : \answertype \mid \phi[N/x] \}}
		\end{tikzcd}
	\end{center}
	The composite is equal to $(\identity{}, \mathcal{A}\interpret{M\ N} \comp \pi_1) : 1 \mathcal{A}^{\mathcal{P}}\interpret{\dot{\Gamma}} \to \mathcal{A}^{\mathcal{P}}\interpret{\dot{\Gamma} \vdash \{ v : \answertype \mid \phi[N/x] \}}$.
	\begin{align}
		&\tupling{\identity{}}{\mathcal{A}\interpret{N}}^{*} (\identity{}, \eval \comp \braiding \comp (\pi_2 \times \identity{})) \comp (\identity{}, \mathcal{A}\interpret{M} \comp \pi_1) \\
		&= (\identity{}, \eval \comp \braiding \comp (\pi_2 \times \identity{}) \comp (\tupling{\identity{}}{\mathcal{A}\interpret{N}} \times \identity{}) \comp \tupling{\pi_1}{\mathcal{A}\interpret{M} \comp \pi_1}) \\
		&= (\identity{}, \eval \comp \tupling{\mathcal{A}\interpret{M}}{\mathcal{A}\interpret{N}} \comp \pi_1) \\
		&= (\identity{}, \mathcal{A}\interpret{M\ N} \comp \pi_1)
	\end{align}
\end{proof}

\begin{lemma}\label{lem:R-Abs-sound}
	\textsc{R-Abs} is sound.
\end{lemma}
\begin{proof}
	By IH, we have
	\begin{gather}
		(\identity{}, \mathcal{A}\interpret{M} \comp \pi_1) : 1 \mathcal{A}^{\mathcal{P}}\interpret{\dot{\Gamma}, x : \dot{\sigma}} \dotTo \mathcal{A}^{\mathcal{P}}\interpret{\dot{\Gamma}, x : \dot{\sigma} \vdash \{ v : \answertype \mid \phi \}}
	\end{gather}
	We consider the following vertical morphism in $\refinetotal$.
	\begin{center}
		\begin{tikzcd}
			1 \mathcal{A}^{\mathcal{P}}\interpret{\dot{\Gamma}} \ar[d, "\eta^{\pi^{*} \dashv \prod}"] \\
			\prod_{\mathcal{A}^{\mathcal{P}}\interpret{\dot{\Gamma} \vdash \dot{\sigma}}} \pi_{\mathcal{A}^{\mathcal{P}}\interpret{\dot{\Gamma} \vdash \dot{\sigma}}}^{*} 1 \mathcal{A}^{\mathcal{P}}\interpret{\dot{\Gamma}} \ar[d, equal] \\
			\prod_{\mathcal{A}^{\mathcal{P}}\interpret{\dot{\Gamma} \vdash \dot{\sigma}}} 1 \mathcal{A}^{\mathcal{P}}\interpret{\dot{\Gamma}, x : \dot{\sigma}} \ar[d, "{\prod (\identity{}, \mathcal{A}\interpret{M} \comp \pi_1)}"] \\
			\prod_{\mathcal{A}^{\mathcal{P}}\interpret{\dot{\Gamma} \vdash \dot{\sigma}}} \mathcal{A}^{\mathcal{P}}\interpret{\dot{\Gamma}, x : \dot{\sigma} \vdash \{ v : \answertype \mid \phi \}} \ar[d, equal] \\
			\mathcal{A}^{\mathcal{P}}\interpret{\dot{\Gamma} \vdash (x : \dot{\sigma}) \to \{ v : \answertype \mid \phi \}}
		\end{tikzcd}
	\end{center}
	The composite is equal to $(\identity{}, \mathcal{A}\interpret{\lambda x. M} \comp \pi_1)$.
	\begin{align}
		&\prod (\identity{}, \mathcal{A}\interpret{M} \comp \pi_1) \comp \eta^{\pi^{*} \dashv \prod} \\
		&= (\identity{}, \Lambda (\mathcal{A}\interpret{M} \comp \pi_1 \comp \tupling{\pi_1 \times \identity{}}{\eval \comp (\pi_2 \times \identity{})})) \comp \eta^{\pi^{*} \dashv \prod} \\
		&= (\identity{}, \Lambda (\mathcal{A}\interpret{M}) \comp \pi_1) \comp (\identity{}, \dots) \\
		&= (\identity{}, \Lambda (\mathcal{A}\interpret{M}) \comp \pi_1) \\
		&= (\identity{}, \mathcal{A}\interpret{\lambda x. M} \comp \pi_1)
	\end{align}
\end{proof}

\begin{lemma}\label{lem:R-Unit-sound}
	\textsc{R-Unit} is sound.
\end{lemma}
\begin{proof}
	Note that we have
	\[ \mathcal{A}^{\mathcal{P}}\interpret{\dot{\Gamma} \vdash \{ v : 1 \mid \top \}} = 1 \mathcal{A}^{\mathcal{P}}\interpret{\dot{\Gamma}} \]
	Thus,
	\[ (\identity{}, {!} \comp \pi_1) = {!} : 1 \mathcal{A}^{\mathcal{P}}\interpret{\dot{\Gamma}} \to 1 \mathcal{A}^{\mathcal{P}}\interpret{\dot{\Gamma}} \]
\end{proof}

\begin{lemma}\label{lem:R-Pair-sound}
	\textsc{R-Pair} is sound.
\end{lemma}
\begin{proof}
	By IH, we have
	\begin{gather}
		(\identity{}, \mathcal{A}\interpret{M} \comp \pi_1) : 1 \mathcal{A}^{\mathcal{P}}\interpret{\dot{\Gamma}} \dotTo \mathcal{A}^{\mathcal{P}}\interpret{\dot{\Gamma} \vdash \dot{\sigma}} \\
		(\identity{}, \mathcal{A}\interpret{N} \comp \pi_1) : 1 \mathcal{A}^{\mathcal{P}}\interpret{\dot{\Gamma}} \dotTo \mathcal{A}^{\mathcal{P}}\interpret{\dot{\Gamma} \vdash \dot{\tau}[M/x]}
	\end{gather}
	We consider the following vertical morphism in $\refinetotal$.
	\begin{center}
		\begin{tikzcd}
			1 \mathcal{A}^{\mathcal{P}}\interpret{\dot{\Gamma}} \ar[d, "{(\identity{}, \mathcal{A}\interpret{N} \comp \pi_1)}"] \\
			\mathcal{A}^{\mathcal{P}}\interpret{\dot{\Gamma} \vdash \dot{\tau}[M/x]} \ar[d, equal, "\text{by Lemma~\ref{lem:subst-type}}"] \\
			\tupling{\identity{}}{\mathcal{A}\interpret{M}}^{*} \mathcal{A}^{\mathcal{P}}\interpret{\dot{\Gamma}, x : \dot{\sigma} \vdash \dot{\tau}} \ar[d, "{\tupling{\identity{}}{\mathcal{A}\interpret{M}}^{*} \eta^{\coprod \dashv \pi^{*}}}"] \\
			\tupling{\identity{}}{\mathcal{A}\interpret{M}}^{*} \pi_{\mathcal{A}^{\mathcal{P}}\interpret{\dot{\Gamma} \vdash \dot{\sigma}}}^{*} \coprod_{\mathcal{A}^{\mathcal{P}}\interpret{\dot{\Gamma} \vdash \dot{\sigma}}} \mathcal{A}^{\mathcal{P}}\interpret{\dot{\Gamma}, x : \dot{\sigma} \vdash \dot{\tau}} \ar[d, equal] \\
			\coprod_{\mathcal{A}^{\mathcal{P}}\interpret{\dot{\Gamma} \vdash \dot{\sigma}}} \mathcal{A}^{\mathcal{P}}\interpret{\dot{\Gamma}, x : \dot{\sigma} \vdash \dot{\tau}} \ar[d, equal] \\
			\mathcal{A}^{\mathcal{P}}\interpret{\dot{\Gamma} \vdash (x : \dot{\sigma}) \times \dot{\tau}}
		\end{tikzcd}
	\end{center}
	The composite is equal to $(\identity{}, \mathcal{A}\interpret{(M, N)} \comp \pi_1)$.
	\begin{align}
		&\tupling{\identity{}}{\mathcal{A}\interpret{M}}^{*} \eta^{\coprod \dashv \pi^{*}} \comp (\identity{}, \mathcal{A}\interpret{N} \comp \pi_1) \\
		&= (\identity{}, (\pi_2 \times \identity{}) \comp (\tupling{\identity{}}{\mathcal{A}\interpret{M}} \times \identity{}) \comp \tupling{\pi_1}{\mathcal{A}\interpret{N} \comp \pi_1}) \\
		&= (\identity{}, \tupling{\mathcal{A}\interpret{M}}{\mathcal{A}\interpret{N}} \comp \pi_1) \\
		&= (\identity{}, \mathcal{A}\interpret{(M, N)} \comp \pi_1)
	\end{align}
\end{proof}

\begin{lemma}\label{lem:R-Fst-sound}
	\textsc{R-Fst} is sound.
\end{lemma}
\begin{proof}
	By IH, we have
	\begin{gather}
		(\identity{}, \mathcal{A}\interpret{M} \comp \pi_1) : 1 \mathcal{A}^{\mathcal{P}}\interpret{\dot{\Gamma}} \dotTo \mathcal{A}^{\mathcal{P}}\interpret{\dot{\Gamma} \vdash (x : \dot{\sigma}) \times \dot{\tau}}
	\end{gather}
	We consider the following composite.
	\begin{center}
		\begin{tikzcd}
			1 \mathcal{A}^{\mathcal{P}}\interpret{\dot{\Gamma}} \ar[d, "{(\identity{}, \mathcal{A}\interpret{M} \comp \pi_1)}"] \\
			\mathcal{A}^{\mathcal{P}}\interpret{\dot{\Gamma} \vdash (x : \dot{\sigma}) \times \dot{\tau}} \ar[d, equal] \\
			\coprod \mathcal{A}^{\mathcal{P}}\interpret{\dot{\Gamma}, x : \dot{\sigma} \vdash \dot{\tau}} \ar[d, "\mathbf{fst}"] \\
			\mathcal{A}^{\mathcal{P}}\interpret{\dot{\Gamma} \vdash \dot{\sigma}}
		\end{tikzcd}
	\end{center}
	This is equal to $(\identity{}, \mathcal{A}\interpret{\pi_1\ M} \comp \pi_1)$.
	\begin{align}
		\mathbf{fst} \comp (\identity{}, \mathcal{A}\interpret{M} \comp \pi_1) &= (\identity{}, \pi_1 \comp \pi_2 \comp \tupling{\pi_1}{\mathcal{A}\interpret{M} \comp \pi_1}) \\
		&= (\identity{}, \pi_1 \comp \mathcal{A}\interpret{M} \comp \pi_1) \\
		&= (\identity{}, \mathcal{A}\interpret{\pi_1\ M} \comp \pi_1)
	\end{align}
\end{proof}

\begin{lemma}\label{lem:R-Snd-sound}
	\textsc{R-Snd} is sound.
\end{lemma}
\begin{proof}
	By IH, we have
	\begin{gather}
		(\identity{}, \mathcal{A}\interpret{M} \comp \pi_1) : 1 \mathcal{A}^{\mathcal{P}}\interpret{\dot{\Gamma}} \dotTo \mathcal{A}^{\mathcal{P}}\interpret{\dot{\Gamma} \vdash (x : \dot{\sigma}) \times \dot{\tau}}
	\end{gather}
	We consider the following composite.
	\begin{center}
		\begin{tikzcd}
			1 \mathcal{A}^{\mathcal{P}}\interpret{\dot{\Gamma}} \ar[d, equal] \\
			\tupling{\identity{}}{\mathcal{A}\interpret{M}}^{*} (\kappa^{-1})^{*} 1 \mathcal{A}^{\mathcal{P}}\interpret{\dot{\Gamma}, x : \dot{\sigma}, y : \dot{\tau}} \ar[d, "{\tupling{\identity{}}{\mathcal{A}\interpret{M}}^{*} (\kappa^{-1})^{*} \eta^{\coprod \dashv \pi^{*}}}"] \\
			\tupling{\identity{}}{\mathcal{A}\interpret{M}}^{*} (\kappa^{-1})^{*} \pi_{\mathcal{A}^{\mathcal{P}}\interpret{\dot{\Gamma}, x : \dot{\sigma} \vdash \dot{\tau}}}^{*} \coprod_{\mathcal{A}^{\mathcal{P}}\interpret{\dot{\Gamma}, x : \dot{\sigma} \vdash \dot{\tau}}} 1 \mathcal{A}^{\mathcal{P}}\interpret{\dot{\Gamma}, x : \dot{\sigma}, y : \dot{\tau}} \ar[d, "{\tupling{\identity{}}{\mathcal{A}\interpret{M}}^{*} (\kappa^{-1})^{*} \pi^{*} \mathbf{fst}}"] \\
			\tupling{\identity{}}{\mathcal{A}\interpret{M}}^{*} (\kappa^{-1})^{*} \pi_{\mathcal{A}^{\mathcal{P}}\interpret{\dot{\Gamma}, x : \dot{\sigma} \vdash \dot{\tau}}}^{*} \mathcal{A}^{\mathcal{P}}\interpret{\dot{\Gamma}, x : \dot{\sigma} \vdash \dot{\tau}} \ar[d, equal] \\
			\tupling{\identity{}}{\mathcal{A}\interpret{\pi_1\ M}}^{*} \mathcal{A}^{\mathcal{P}}\interpret{\dot{\Gamma}, x : \dot{\sigma} \vdash \dot{\tau}} \ar[d, equal, "\text{by Lemma~\ref{lem:subst-type}}"] \\
			\mathcal{A}^{\mathcal{P}}\interpret{\dot{\Gamma} \vdash \dot{\tau}[\pi_1\ M/x]}
		\end{tikzcd}
	\end{center}
	Here, we used the following equation.
	\begin{align}
		\pi_{\mathcal{A}^{\mathcal{P}}\interpret{\dot{\Gamma}, x : \dot{\sigma} \vdash \dot{\tau}}} \comp \kappa^{-1} \comp \tupling{\identity{}}{\mathcal{A}\interpret{M}} &= \pi_1 \comp \associator^{-1} \comp \tupling{\identity{}}{\mathcal{A}\interpret{M}} \\
		&= \tupling{\identity{}}{\pi_1 \comp \mathcal{A}\interpret{M}} \\
		&= \tupling{\identity{}}{\mathcal{A}\interpret{\pi_1\ M}}
	\end{align}
	The composite above is equal to $(\identity{}, \mathcal{A}\interpret{\pi_2\ M} \comp \pi_1)$.
	\begin{align}
		&\tupling{\identity{}}{\mathcal{A}\interpret{M}}^{*} (\kappa^{-1})^{*} \pi^{*} \mathbf{fst} \comp \tupling{\identity{}}{\mathcal{A}\interpret{M}}^{*} (\kappa^{-1})^{*} \eta^{\coprod \dashv \pi^{*}} \\
		&= (\identity{}, \pi_1 \comp \pi_2 \comp (\tupling{\identity{}}{\pi_1 \comp \mathcal{A}\interpret{M}} \times \identity{})) \comp (\identity{}, (\pi_2 \times \identity{}) \comp ((\associator^{-1} \comp \tupling{\identity{}}{\mathcal{A}\interpret{M}}) \times \identity{})) \\
		&= (\identity{}, \pi_1 \comp \pi_2) \comp (\identity{}, (\pi_2 \comp \mathcal{A}\interpret{M}) \times \identity{}) \\
		&= (\identity{}, \pi_1 \comp \pi_2 \comp \tupling{\pi_1}{(\pi_2 \comp \mathcal{A}\interpret{M}) \times \identity{}}) \\
		&= (\identity{}, \pi_2 \comp \mathcal{A}\interpret{M} \comp \pi_1) \\
		&= (\identity{}, \mathcal{A}\interpret{\pi_2\ M} \comp \pi_1)
	\end{align}
\end{proof}

\begin{lemma}\label{lem:R-Case0-sound}
	\textsc{R-Case0} is sound.
\end{lemma}
\begin{proof}
	By IH, we have
	\begin{gather}
		(\identity{}, \mathcal{A}\interpret{M} \comp \pi_1) : 1 \mathcal{A}^{\mathcal{P}}\interpret{\dot{\Gamma}} \dotTo 0 \mathcal{A}^{\mathcal{P}}\interpret{\dot{\Gamma}}
	\end{gather}
	Consider the following.
	\begin{center}
		\begin{tikzcd}
			1 \mathcal{A}^{\mathcal{P}}\interpret{\dot{\Gamma}} \ar[d, "{(\identity{}, \mathcal{A}\interpret{M} \comp \pi_1)}"] \\
			0 \mathcal{A}^{\mathcal{P}}\interpret{\dot{\Gamma}} \ar[d, "?"] \\
			\mathcal{A}^{\mathcal{P}}\interpret{\dot{\Gamma} \vdash \dot{\tau}}
		\end{tikzcd}
	\end{center}
	This is equal to $(\identity{}, \mathcal{A}\interpret{\delta(M)} \comp \pi_1)$.
	\begin{align}
		{?} \comp (\identity{}, \mathcal{A}\interpret{M} \comp \pi_1) &= (\identity{}, {?} \comp \pi_2 \comp \tupling{\pi_1}{\mathcal{A}\interpret{M} \comp \pi_1}) \\
		&= (\identity{}, {?} \comp \mathcal{A}\interpret{M} \comp \pi_1) \\
		&= (\identity{}, \mathcal{A}\interpret{\delta(M)} \comp \pi_1)
	\end{align}
\end{proof}

\begin{lemma}\label{lem:R-Inj-sound}
	\textsc{R-Inj} is sound.
\end{lemma}
\begin{proof}
	By IH, we have
	\begin{gather}
		(\identity{}, \mathcal{A}\interpret{M} \comp \pi_1) : 1 \mathcal{A}^{\mathcal{P}}\interpret{\dot{\Gamma}} \dotTo \mathcal{A}^{\mathcal{P}}\interpret{\dot{\Gamma} \vdash \dot{\sigma}_i}
	\end{gather}
	Consider the following.
	\begin{center}
		\begin{tikzcd}
			1 \mathcal{A}^{\mathcal{P}}\interpret{\dot{\Gamma}} \ar[d, "{(\identity{}, \mathcal{A}\interpret{M} \comp \pi_1)}"] \\
			\mathcal{A}^{\mathcal{P}}\interpret{\dot{\Gamma} \vdash \dot{\sigma}_i} \ar[d, "\iota_i"] \\
			\mathcal{A}^{\mathcal{P}}\interpret{\dot{\Gamma} \vdash \dot{\sigma}_1 + \dot{\sigma}_2}
		\end{tikzcd}
	\end{center}
	This is equal to $(\identity{}, \mathcal{A}\interpret{\iota_i\ M} \comp \pi_1)$.
	\begin{align}
		\iota_i \comp (\identity{}, \mathcal{A}\interpret{M} \comp \pi_1) &= (\identity{}, \iota_i \comp \pi_2 \comp \tupling{\pi_1}{\mathcal{A}\interpret{M} \comp \pi_1}) \\
		&= (\identity{}, \iota_i \comp \mathcal{A}\interpret{M} \comp \pi_1) \\
		&= (\identity{}, \mathcal{A}\interpret{\iota_i\ M} \comp \pi_1)
	\end{align}
\end{proof}

\begin{lemma}\label{lem:R-Case2-sound}
	\textsc{R-Case2} is sound.
\end{lemma}
\begin{proof}
	By IH, we have
	\begin{gather}
		(\identity{}, \mathcal{A}\interpret{M} \comp \pi_1) : 1 \mathcal{A}^{\mathcal{P}}\interpret{\dot{\Gamma}} \dotTo \mathcal{A}^{\mathcal{P}}\interpret{\dot{\Gamma} \vdash \dot{\sigma}_1 + \dot{\sigma}_2} \\
		(\identity{}, \mathcal{A}\interpret{N_1} \comp \pi_1) : 1 \mathcal{A}^{\mathcal{P}}\interpret{\dot{\Gamma}, x_1 : \dot{\sigma}_1} \dotTo \mathcal{A}^{\mathcal{P}}\interpret{\dot{\Gamma}, x_1 : \dot{\sigma}_1 \vdash \dot{\tau}[\iota_1\ x_1/z]} \\
		(\identity{}, \mathcal{A}\interpret{N_2} \comp \pi_1) : 1 \mathcal{A}^{\mathcal{P}}\interpret{\dot{\Gamma}, x_2 : \dot{\sigma}_2} \dotTo \mathcal{A}^{\mathcal{P}}\interpret{\dot{\Gamma}, x_2 : \dot{\sigma}_2 \vdash \dot{\tau}[\iota_2\ x_2/z]}
	\end{gather}
	By Lemma~\ref{lem:weakening-type},\ref{lem:subst-type}, we have the following for $i = 1, 2$.
	\begin{align}
		&\mathcal{A}^{\mathcal{P}}\interpret{\dot{\Gamma}, x_i : \dot{\sigma}_i \vdash \dot{\tau}[\iota_i\ x_i/z]} \\
		&= \tupling{\identity{}}{\mathcal{A}\interpret{\underlying{\dot{\Gamma}}, x_i : \underlying{\dot{\sigma}_i} \vdash \iota_i\ x_i : \underlying{\dot{\sigma}_1} + \underlying{\dot{\sigma}_2}}}^{*} \mathcal{A}^{\mathcal{P}}\interpret{\dot{\Gamma}, x_i : \dot{\sigma}_i, z : \dot{\sigma}_1 + \dot{\sigma}_2 \vdash \dot{\tau}} \\
		&= \tupling{\identity{}}{\iota_i \comp \pi_2}^{*} \mathrm{proj}_{\underlying{\dot{\Gamma}}; \underlying{\dot{\sigma}_i}; z : \underlying{\dot{\sigma}_1} + \underlying{\dot{\sigma}_2}} \mathcal{A}^{\mathcal{P}}\interpret{\dot{\Gamma}, z : \dot{\sigma}_1 + \dot{\sigma}_2 \vdash \dot{\tau}} \\
		&= (\identity{} \times \iota_i)^{*} \mathcal{A}^{\mathcal{P}}\interpret{\dot{\Gamma}, z : \dot{\sigma}_1 + \dot{\sigma}_2 \vdash \dot{\tau}}
	\end{align}
	\[ 1 \mathcal{A}^{\mathcal{P}}\interpret{\dot{\Gamma}, x_1 : \dot{\sigma}_1} = \{ \iota_1 \}^{*} 1 \mathcal{A}^{\mathcal{P}}\interpret{\dot{\Gamma}, z : \dot{\sigma}_1 + \dot{\sigma}_2} = (\identity{} \times \iota_1)^{*} 1 \mathcal{A}^{\mathcal{P}}\interpret{\dot{\Gamma}, z : \dot{\sigma}_1 + \dot{\sigma}_2} \]
	Since we have strong fibred coproducts,
	\begin{align}
		&(\identity{}, [\mathcal{A}\interpret{N_1} \comp \pi_1, \mathcal{A}\interpret{N_2} \comp \pi_1] \comp [(\identity{} \times \iota_1) \times \identity{}, (\identity{} \times \iota_2) \times \identity{}]^{-1}) \\
		& : 1 \mathcal{A}^{\mathcal{P}}\interpret{\dot{\Gamma}, z : \dot{\sigma}_1 + \dot{\sigma}_2} \dotTo \mathcal{A}^{\mathcal{P}}\interpret{\dot{\Gamma}, z : \dot{\sigma}_1 + \dot{\sigma}_2 \vdash \dot{\tau}}
	\end{align}
	Now consider the following.
	\begin{center}
		\begin{tikzcd}
			1 \mathcal{A}^{\mathcal{P}}\interpret{\dot{\Gamma}} \ar[d, equal] \\
			\tupling{\identity{}}{\mathcal{A}\interpret{M}}^{*} \pi_{\mathcal{A}^{\mathcal{P}}\interpret{\dot{\Gamma} \vdash \dot{\sigma}_1 + \dot{\sigma}_2}}^{*} 1 \mathcal{A}^{\mathcal{P}}\interpret{\dot{\Gamma}} \ar[d, equal] \\
			\tupling{\identity{}}{\mathcal{A}\interpret{M}}^{*} 1 \mathcal{A}^{\mathcal{P}}\interpret{\dot{\Gamma}, z : \dot{\sigma}_1 + \dot{\sigma}_2} \ar[d, "{\tupling{\identity{}}{\mathcal{A}\interpret{M}}^{*} (\identity{}, \dots)}"] \\
			\tupling{\identity{}}{\mathcal{A}\interpret{M}}^{*} \mathcal{A}^{\mathcal{P}}\interpret{\dot{\Gamma}, z : \dot{\sigma}_1 + \dot{\sigma}_2 \vdash \dot{\tau}} \ar[d, equal, "\text{by Lemma~\ref{lem:subst-type}}"] \\
			\mathcal{A}^{\mathcal{P}}\interpret{\dot{\Gamma} \vdash \dot{\tau}[M/z]}
		\end{tikzcd}
	\end{center}
	This is equal to $(\identity{}, \mathcal{A}\interpret{\delta(M, x_1. N_1, x_2. N_2)} \comp \pi_1)$.
	\begin{align}
		&\tupling{\identity{}}{\mathcal{A}\interpret{M}}^{*} (\identity{}, [\mathcal{A}\interpret{N_1} \comp \pi_1, \mathcal{A}\interpret{N_2} \comp \pi_1] \comp [(\identity{} \times \iota_1) \times \identity{}, (\identity{} \times \iota_2) \times \identity{}]^{-1}) \\
		&= (\identity{}, [\mathcal{A}\interpret{N_1} \comp \pi_1, \mathcal{A}\interpret{N_2} \comp \pi_1] \comp [(\identity{} \times \iota_1) \times \identity{}, (\identity{} \times \iota_2) \times \identity{}]^{-1} \comp (\tupling{\identity{}}{\mathcal{A}\interpret{M}} \times \identity{})) \\
		&= (\identity{}, [\mathcal{A}\interpret{N_1}, \mathcal{A}\interpret{N_2}] \comp [\identity{} \times \iota_1, \identity{} \times \iota_2]^{-1} \comp \pi_1 \comp (\tupling{\identity{}}{\mathcal{A}\interpret{M}} \times \identity{})) \\
		&= (\identity{}, [\mathcal{A}\interpret{N_1}, \mathcal{A}\interpret{N_2}] \comp [\identity{} \times \iota_1, \identity{} \times \iota_2]^{-1} \comp \tupling{\identity{}}{\mathcal{A}\interpret{M}} \comp \pi_1) \\
		&= (\identity{}, \mathcal{A}\interpret{\delta(M, x_1. N_1, x_2. N_2)} \comp \pi_1)
	\end{align}
\end{proof}

\begin{lemma}\label{lem:R-Var-sound}
	\textsc{R-Var} is sound.
\end{lemma}
\begin{proof}
	By induction on the context.
	\begin{itemize}
		\item If $\dot{\Gamma}, x : \dot{\sigma} \vdash x : \dot{\sigma}$, then we prove
		\[ (\identity{}, \pi_2 \comp \pi_1) : 1 \mathcal{A}^{\mathcal{P}}\interpret{\dot{\Gamma}, x : \dot{\sigma}} \dotTo \mathcal{A}^{\mathcal{P}}\interpret{\dot{\Gamma}, x : \dot{\sigma} \vdash \dot{\sigma}}. \]
		Consider the following.
		\begin{center}
			\begin{tikzcd}
				1 \mathcal{A}^{\mathcal{P}}\interpret{\dot{\Gamma}, x : \dot{\sigma}} \ar[d, "\eta^{\coprod \dashv \pi^{*}}"] \\
				\pi^{*}_{\mathcal{A}^{\mathcal{P}}\interpret{\dot{\Gamma} \vdash \dot{\sigma}}} \coprod_{\mathcal{A}^{\mathcal{P}}\interpret{\dot{\Gamma} \vdash \dot{\sigma}}} 1 \mathcal{A}^{\mathcal{P}}\interpret{\dot{\Gamma}, x : \dot{\sigma}} \ar[d, "\pi^{*} \mathbf{fst}"] \\
				\pi_{\mathcal{A}^{\mathcal{P}}\interpret{\dot{\Gamma} \vdash \dot{\sigma}}}^{*} \mathcal{A}^{\mathcal{P}}\interpret{\dot{\Gamma} \vdash \dot{\sigma}} \ar[d, equal, "\text{by Lemma~\ref{lem:weakening-type}}"] \\
				\mathcal{A}^{\mathcal{P}}\interpret{\dot{\Gamma}, x : \dot{\sigma} \vdash \dot{\sigma}}
			\end{tikzcd}
		\end{center}
		This is equal to $(\identity{}, \pi_2 \comp \pi_1)$.
		\begin{align}
			\pi^{*} \mathbf{fst} \comp \eta^{\coprod \dashv \pi^{*}} &= (\identity{}, \pi_1 \comp \pi_2 \comp (\pi_1 \times \identity{})) \comp (\identity{}, \pi_2 \times \identity{}) \\
			&= (\identity{}, \pi_1 \comp \pi_2 \comp (\pi_1 \times \identity{}) \comp \tupling{\pi_1}{\pi_2 \times \identity{}}) \\
			&= (\identity{}, \pi_2 \comp \pi_1)
		\end{align}
		\item If $\dot{\Gamma}, y : \dot{\tau} \vdash x : \dot{\sigma}$ and $(x : \dot{\sigma}) \in \dot{\Gamma}$, then we prove
		\[ (\identity{}, \mathcal{A}^{\mathcal{P}}\interpret{\underlying{\dot{\Gamma}}, y : \underlying{\dot{\tau}} \vdash x : \underlying{\dot{\sigma}}} \comp \pi_1) : 1 \mathcal{A}^{\mathcal{P}}\interpret{\dot{\Gamma}, y : \dot{\tau}} \dotTo \mathcal{A}^{\mathcal{P}}\interpret{\dot{\Gamma}, y : \dot{\tau} \vdash \dot{\sigma}} \]
		where IH is given as follows.
		\[ (\identity{}, \interpret{\underlying{\dot{\Gamma}} \vdash x : \underlying{\dot{\sigma}}} \comp \pi_1) : 1 \mathcal{A}^{\mathcal{P}}\interpret{\dot{\Gamma}} \dotTo \mathcal{A}^{\mathcal{P}}\interpret{\dot{\Gamma} \vdash \dot{\sigma}} \]
		Consider the following.
		\begin{center}
			\begin{tikzcd}
				1 \mathcal{A}^{\mathcal{P}}\interpret{\dot{\Gamma}, y : \dot{\tau}} \ar[d, equal] \\
				\pi_{\mathcal{A}^{\mathcal{P}}\interpret{\dot{\Gamma} \vdash \dot{\tau}}}^{*} 1 \mathcal{A}^{\mathcal{P}}\interpret{\dot{\Gamma}} \ar[d, "{\pi^{*} (\identity{}, \interpret{\underlying{\dot{\Gamma}} \vdash x : \underlying{\dot{\sigma}}} \comp \pi_1)}"] \\
				\pi_{\mathcal{A}^{\mathcal{P}}\interpret{\dot{\Gamma} \vdash \dot{\tau}}}^{*} \mathcal{A}^{\mathcal{P}}\interpret{\dot{\Gamma} \vdash \dot{\sigma}} \ar[d, equal, "\text{by Lemma~\ref{lem:weakening-type}}"] \\
				\mathcal{A}^{\mathcal{P}}\interpret{\dot{\Gamma}, y : \dot{\tau} \vdash \dot{\sigma}}
			\end{tikzcd}
		\end{center}
		This is equal to $(\identity{}, \mathcal{A}^{\mathcal{P}}\interpret{\underlying{\dot{\Gamma}}, y : \underlying{\dot{\tau}} \vdash x : \underlying{\dot{\sigma}}} \comp \pi_1)$.
		\begin{align}
			&\pi^{*} (\identity{}, \interpret{\underlying{\dot{\Gamma}} \vdash x : \underlying{\dot{\sigma}}} \comp \pi_1) \\
			&= (\identity{}, \interpret{\underlying{\dot{\Gamma}} \vdash x : \underlying{\dot{\sigma}}} \comp \pi_1 \comp (\pi_1 \times \identity{})) \\
			&= (\identity{}, \interpret{\underlying{\dot{\Gamma}} \vdash x : \underlying{\dot{\sigma}}} \comp \pi_1 \comp \pi_1) \\
			&= (\identity{}, \interpret{\underlying{\dot{\Gamma}, y : \underlying{\dot{\tau}}} \vdash x : \underlying{\dot{\sigma}}} \comp \pi_1)
		\end{align}
	\end{itemize}
\end{proof}

\begin{lemma}\label{lem:R-VarRefine-sound}
	\textsc{R-VarRefine} is sound.
\end{lemma}
\begin{proof}
	\begin{itemize}
		\item If $\dot{\Gamma}, x : \{ v : b \mid \phi \} \vdash x : \{ v : b \mid v = x \}$, then we prove
		\[ (\identity{}, \pi_2 \comp \pi_1) : 1 \mathcal{A}^{\mathcal{P}}\interpret{\dot{\Gamma}, x : \{ v : b \mid \phi \}} \dotTo \mathcal{A}^{\mathcal{P}}\interpret{\dot{\Gamma}, x : \{ v : b \mid \phi \} \vdash \{ v : b \mid v = x \}}. \]
		Since
		\begin{align}
			&\mathcal{A}^{\mathcal{P}}\interpret{\dot{\Gamma}, x : \{ v : b \mid \phi \} \vdash \{ v : b \mid v = x \}} \\
			&= ((\dots, \dots), \mathcal{A}^{\mathcal{P}}\interpret{\dot{\Gamma}, x : \{ v : b \mid \phi \}}, \\
			&\qquad \pi^{*} \mathcal{A}^{\mathcal{P}}\interpret{\dot{\Gamma}, x : \{ v : b \mid \phi \}} \land \mathcal{A}^{\mathcal{P}}\interpret{\underlying{\dot{\Gamma}}, x : b, v : b \vdash v = x }) \\
			&= ((\dots, \dots), \mathcal{A}^{\mathcal{P}}\interpret{\dot{\Gamma}, x : \{ v : b \mid \phi \}}, \\
			&\qquad \pi^{*} \mathcal{A}^{\mathcal{P}}\interpret{\dot{\Gamma}, x : \{ v : b \mid \phi \}} \land \tupling{\pi_2}{\pi_2 \comp \pi_1}^{*} \mathcal{P}({=}))
		\end{align}
		it suffices to prove
		\[ \pi^{*} \mathcal{A}^{\mathcal{P}}\interpret{\dot{\Gamma}, x : \{ v : b \mid \phi \}} \le \tupling{\pi_1}{\pi_2 \comp \pi_1}^{*} (\pi^{*} \mathcal{A}^{\mathcal{P}}\interpret{\dot{\Gamma}, x : \{ v : b \mid \phi \}} \land \tupling{\pi_2}{\pi_2 \comp \pi_1}^{*} \mathcal{P}({=})) \]
		This follows from the following equations.
		\[ \tupling{\pi_1}{\pi_2 \comp \pi_1}^{*} \pi^{*} \mathcal{A}^{\mathcal{P}}\interpret{\dot{\Gamma}, x : \{ v : b \mid \phi \}} = \pi^{*} \mathcal{A}^{\mathcal{P}}\interpret{\dot{\Gamma}, x : \{ v : b \mid \phi \}} \]
		\begin{align}
			\tupling{\pi_1}{\pi_2 \comp \pi_1}^{*} \tupling{\pi_2}{\pi_2 \comp \pi_1}^{*} \mathcal{P}({=}) &= \pi_1^{*} \pi_2^{*} \tupling{\identity{}}{\identity{}}^{*} \mathcal{P}({=}) \\
			&= \pi_1^{*} \pi_2^{*} \top (A b) \\
			&= \top \mathcal{A}\interpret{\underlying{\dot{\Gamma}}, x : b, v : b}
		\end{align}
		\item If $\dot{\Gamma}, y : \dot{\tau} \vdash x : \{ v : b \mid v = x \}$, then the proof is the same as R-Var.
	\end{itemize}
\end{proof}

\begin{definition}
	Let $\Gamma = x_1 : \sigma_1, \dots, x_n : \sigma_n$ be a context and $\gamma \in 1 \to \mathcal{A}\interpret{\Gamma}$.
	For each $f : 1 \to \mathcal{A}\interpret{\sigma_i}$, we define $\gamma[x_i \mapsto f] : 1 \to \mathcal{A}\interpret{\Gamma}$ by
	\[ \gamma[x_i \mapsto f] \coloneqq \langle \dots \tupling{\identity{}}{\gamma_1}, \dots, f, \dots, \gamma_n \rangle \]
	where $\gamma = \langle \dots \tupling{\identity{}}{\gamma_1}, \dots, \gamma_i, \dots, \gamma_n \rangle$.
\end{definition}

\begin{definition}
	Let $\Gamma = x_1 : \sigma_1, \dots, x_n : \sigma_n$ be a context.
	We say $Q \in \category{P}_{\mathcal{A}\interpret{\Gamma}}$ is \emph{admissible at variable $x_i$} if
	\begin{itemize}
		\item $\category{C}(1, \mathcal{A}\interpret{\sigma_i})$ has a bottom element $\bot_{\mathcal{A}\interpret{\sigma_i}}$
		\item for any $\gamma : 1 \to \mathcal{A}\interpret{\Gamma}$, we have $\gamma[x_i \mapsto \bot_{\mathcal{A}\interpret{\sigma_i}}] \in Q$
		\item for any $\gamma : 1 \to \mathcal{A}\interpret{\Gamma}$ and any $\omega$-chain $f_{0} \le f_1 \le \dots$ in $\category{C}(1, \mathcal{A}\interpret{\sigma_i})$, if $\gamma[x_i \mapsto f_k] \in Q$ holds for each $k$, then $\gamma[x_i \mapsto \sup_k f_k] \in Q$ (i.e.\ Q is chain-closed at $x_i$)
	\end{itemize}
\end{definition}

\begin{definition}\label{def:admissible-refinement}
	$((I, X), P, Q) \in \{ s(\category{C}) \mid \category{P} \}$ is \emph{admissible} if
	\begin{itemize}
		\item $X$ is an EM $T$-algebra (which implies that $\category{C}(1, X)$ has a bottom element $\bot_X : 1 \to X$)
		\item for each $i \in P$, $\tupling{i}{\bot_X} \in Q$
		\item for any $i \in P$ and $\omega$-chain $f_0 \le f_1 \le \dots$ such that $\tupling{i}{f_n} \in Q$ for any $n$, we have $\tupling{i}{\sup_n f_n} \in Q$
	\end{itemize}
\end{definition}

\begin{lemma}\label{lem:refinement-admissible}
	If $\vdash \dot{\Gamma}$ and $\mathcal{A}^{\mathcal{P}}\interpret{\underlying{\dot{\Gamma}}, v : \answertype \vdash \phi}$ is admissible at $v$, then $\mathcal{A}^{\mathcal{P}}\interpret{\dot{\Gamma} \vdash \{ v : \answertype \mid \phi \}}$ is admissible.
\end{lemma}
\begin{proof}
	By definition,
	\[ \mathcal{A}^{\mathcal{P}}\interpret{ \dot{\Gamma} \vdash \{ v : \answertype \mid \phi \}} = ((\mathcal{A}\interpret{\underlying{\dot{\Gamma}}}, \answerobj), \mathcal{A}^{\mathcal{P}}\interpret{\dot{\Gamma}}, \mathcal{A}^{\mathcal{P}}\interpret{\underlying{\dot{\Gamma}}, v : \answertype \vdash \phi}) \]
	is admissible.
\end{proof}

\begin{lemma}\label{lem:product-admissible}
	Let $((I, X), P, Q) \in \{ s(\category{C}) \mid \category{P} \}$.
	If $((I \times X, Y), Q, R) \in \{ s(\category{C}) \mid \category{P} \}$ is admissible, then $\dot{\prod}_{((I, X), P, Q)} ((I \times X, Y), Q, R)$ is admissible.
\end{lemma}
\begin{proof}
	Note that we have $\bot_{\exponential{X}{Y}} = \Lambda(\bot_{Y} \comp \pi_1)$.
	\begin{align}
		\bot_{\exponential{X}{Y}} &= \Lambda(\nu_Y \comp T \eval \comp \strength'^T) \comp T {?}_{\exponential{X}{Y}} \comp \bot_{T 0} \\
		&= \Lambda(\nu_Y \comp T \eval \comp T ({?}_{\exponential{X}{Y}} \times \identity{}) \comp \strength'^T \comp (\bot_{T 0} \times \identity{})) \\
		&= \Lambda(\nu_Y \comp T \eval \comp T (\Lambda({?}_{Y} \comp \pi_1) \times \identity{}) \comp \strength'^T \comp (\bot_{T 0} \times \identity{})) \\
		&= \Lambda(\nu_Y \comp T {?}_{Y} \comp T \pi_1 \comp \strength'^T \comp (\bot_{T 0} \times \identity{})) \\
		&= \Lambda(\nu_Y \comp T {?}_{Y} \comp \pi_1 \comp (\bot_{T 0} \times \identity{})) \\
		&= \Lambda(\nu_Y \comp T {?}_{Y} \comp \bot_{T 0} \comp \pi_1) \\
		&= \Lambda(\bot_{Y} \comp \pi_1)
	\end{align}
	So, \eqref{eq:prod_pred} satisfies the second condition of Def.~\ref{def:admissible-refinement} because for each $x : 1 \to X$, if $\tupling{i}{x} \in Q$, then
	\[ \tupling{\tupling{i}{x}}{\eval \comp \tupling{\bot_{\exponential{X}{Y}}}{x}} = \tupling{\tupling{i}{x}}{\bot_{Y} \comp \pi_1 \comp \tupling{\identity{}}{x}} = \tupling{\tupling{i}{x}}{\bot_{Y}} \in R. \]
	It is routine to check that \eqref{eq:prod_pred} satisfies the third condition of Def.~\ref{def:admissible-refinement}.
\end{proof}

\begin{lemma}\label{lem:admissible-fix-point}
	Let $A$ be an EM $T$-algebra for a pseudo-lifting strong monad $T$ with structure map $\nu : T A \to A$.
	If $((I, A), P, Q) \in \{ s(\category{C}) \mid \category{P} \}$ is admissible and $(\identity{}, f) : ((I, A), P, Q) \dotTo ((I, A), P, Q)$, then $(\identity{}, f^{\dagger} \comp \pi_1) : 1 (I, P) \dotTo ((I, A), P, Q)$
\end{lemma}
\begin{proof}
	It suffices to show $\tupling{\identity{}}{f^{\dagger}} : P \dotTo Q$ (that is, $\tupling{i}{f^{\dagger} \comp i} \in Q$ for any $i \in P$) under the assumption that we have $\tupling{\pi_1}{f} : Q \dotTo Q$.

	Recall that $f^{\dagger}$ is defined by
	\[ f^{\dagger} \coloneqq \Lambda^{-1} (\nu_X) \comp \langle \sup_n \lceil f \rceil^{n} \comp \bot_{T (\exponential{I}{A})} \comp {!}, \identity{} \rangle \]
	where
	\begin{align}
		\lceil f \rceil &\coloneqq \eta^T \comp \Lambda (f \comp \langle \pi_2, \Lambda^{-1} (\nu_X) \rangle) : T (\exponential{I}{A}) \to T (\exponential{I}{A}) \\
		\nu_X &\coloneqq \Lambda (\nu \comp T \mathbf{ev} \comp \strength'^T) : T (\exponential{I}{A}) \to \exponential{I}{A}
	\end{align}

	We show
	\[ \Lambda^{-1} (\nu_X) \comp \langle \lceil f \rceil^{n} \comp \bot_{T (\exponential{I}{A})} \comp {!}, \identity{} \rangle \comp i = (f \comp \tupling{i \comp {!}}{\identity{}})^n \comp \bot_A \]
	by induction on $n$.

	\begin{itemize}
		\item Base case:
			Note that we have $\bot_{T (\exponential{I}{A})} = T {?}_{\exponential{I}{A}} \comp \bot_0$ and ${?}_{\exponential{I}{A}} = \Lambda({?}_{A} \comp \pi_1)$.
			\begin{align}
				\Lambda^{-1} (\nu_X) \comp \langle \bot_{T (\exponential{I}{A})} \comp {!}, \identity{} \rangle \comp i &= \nu \comp T \mathbf{ev} \comp \strength'^T \comp \langle T {?}_{\exponential{I}{A}} \comp \bot_0 \comp {!}, \identity{} \rangle \comp i \\
				&= \nu \comp T \mathbf{ev} \comp \strength'^T \comp \langle T \Lambda({?}_{A} \comp \pi_1) \comp \bot_0, i \rangle \\
				&= \nu \comp T ({?}_{A} \comp \pi_1) \comp \strength'^T \comp \langle \bot_0, i \rangle \\
				&= \nu \comp T {?}_{A} \comp T \pi_2 \comp \strength^T \comp \langle i, \bot_0 \rangle \\
				&= \nu \comp T {?}_{A} \comp \pi_2 \comp \langle i, \bot_0 \rangle \\
				&= \nu \comp T {?}_{A} \comp \bot_0 \\
				&= \bot_A
			\end{align}
		\item Step case:
			\begin{align}
				&\Lambda^{-1} (\nu_X) \comp \langle \lceil f \rceil^{n+1} \comp \bot_{T (\exponential{I}{A})} \comp {!}, \identity{} \rangle \comp i \\
				&= \nu \comp T \mathbf{ev} \comp \strength'^T \comp (\lceil f \rceil \times \identity{}) \comp \langle \lceil f \rceil^{n} \comp \bot_{T (\exponential{I}{A})}, i \rangle \\
				&= \nu \comp T \mathbf{ev} \comp \eta^T \comp (\Lambda (f \comp \langle \pi_2, \Lambda^{-1} (\nu_X) \rangle) \times \identity{}) \comp \langle \lceil f \rceil^{n} \comp \bot_{T (\exponential{I}{A})}, i \rangle \\
				&= f \comp \langle \pi_2, \Lambda^{-1} (\nu_X) \rangle \comp \langle \lceil f \rceil^{n} \comp \bot_{T (\exponential{I}{A})}, i \rangle \\
				&= f \comp \langle i, \Lambda^{-1} (\nu_X) \comp \langle \lceil f \rceil^{n} \comp \bot_{T (\exponential{I}{A})}, i \rangle \rangle \\
				&= f \comp \langle i \comp {!}, \identity{} \rangle \comp \Lambda^{-1} (\nu_X) \comp \langle \lceil f \rceil^{n} \comp \bot_{T (\exponential{I}{A})}, i \rangle
			\end{align}
	\end{itemize}
	
	Now we prove $\tupling{i}{f^{\dagger} \comp i} = \tupling{i}{\sup_n (f \comp \tupling{i \comp {!}}{\identity{}})^n \comp \bot_A} \in Q$ using the admissibility of $((I, A), P, Q)$.
	It suffices to prove $\tupling{i}{(f \comp \tupling{i \comp {!}}{\identity{}})^n \comp \bot_A} \in Q$ for each $n$.
	\begin{itemize}
		\item We have $\tupling{i}{\bot_A} \in Q$ by the definition of the admissibility of $((I, A), P, Q)$.
		\item Let $g = f \comp \tupling{i \comp {!}}{\identity{}}$.
			If $\tupling{i}{g^n \comp \bot_A} \in Q$, then
		\begin{align}
			\tupling{i}{g^{n + 1} \comp \bot_A} &= \tupling{i}{(f \comp \tupling{i \comp {!}}{\identity{}}) \comp g^{n} \comp \bot_A} \\
			&= \tupling{i}{f \comp \tupling{i}{g^{n} \comp \bot_A}} \\
			&= \tupling{\pi_1}{f} \comp \tupling{i}{g^{n} \comp \bot_A}
		\end{align}
		and since we have $\tupling{\pi_1}{f} : Q \dotTo Q$, it follows that $\tupling{i}{g^{n+1} \comp \bot_A} \in Q$ holds.
	\end{itemize}
\end{proof}

\begin{lemma}\label{lem:R-Fix-sound}
	\textsc{R-Fix} is sound.
\end{lemma}
\begin{proof}
	By Lemma~\ref{lem:refinement-admissible}, $\mathcal{A}^{\mathcal{P}}\interpret{\dot{\Gamma}, x : \dot{\sigma} \vdash \{ v : \answertype \mid \phi \}}$ is admissible.
	By Lemma~\ref{lem:product-admissible}, $\mathcal{A}^{\mathcal{P}}\interpret{\dot{\Gamma} \vdash (x : \dot{\sigma}) \to \{ v : \answertype \mid \phi \}}$ is admissible.
	Let $X \coloneqq \mathcal{A}^{\mathcal{P}}\interpret{\dot{\Gamma} \vdash (x : \dot{\sigma}) \to \{v : \answertype \mid \phi \}}$.
	By IH, we have 
	\begin{gather}
		(\identity{}, \mathcal{A}\interpret{M} \comp \pi_1) : 1 \{ X \} \dotTo \pi_X^{*} X
	\end{gather}
	Since there is an isomorphism $\category{E}_{\{ X \}}(1 \{ X \}, \pi^{*} Y) \cong \category{E}_{I}(X, Y)$ for any SCCompC $p : \category{E} \to \category{B}$ and $X, Y \in \category{E}_I$, we have
	\begin{gather}
		(\identity{}, \mathcal{A}\interpret{M}) : X \dotTo X
	\end{gather}
	By Lemma~\ref{lem:admissible-fix-point}, we have
	\[ (\identity{}, \mathcal{A}\interpret{\fixpoint{f}{M}} \comp \pi_1) = (\identity{}, (\mathcal{A}\interpret{M})^{\dagger} \comp \pi_1) : 1 \interpret{\dot{\Gamma}} \dotTo \interpret{\dot{\Gamma} \vdash (x : \dot{\sigma}) \to \{v : \answertype \mid \phi \}} \]
\end{proof}

\begin{lemma}\label{lem:R-BasicOp-sound}
	\textsc{R-BasicOp} is sound.
\end{lemma}
\begin{proof}
	We have
	\[ (\identity{}, \mathcal{A}\interpret{\mathbf{op}(M)} \comp \pi_1) : 1 \mathcal{A}^{\mathcal{P}}\interpret{\dot{\Gamma}} \dotTo \mathcal{A}^{\mathcal{P}}\interpret{\dot{\Gamma} \vdash \dot{\tau}} \]
	as the following composite.
	\begin{center}
		\begin{tikzcd}[column sep=huge]
			1 \mathcal{A}^{\mathcal{P}}\interpret{\dot{\Gamma}} \ar[r, "{(\identity{}, \mathcal{A}\interpret{M} \comp \pi_1)}"] &
			\mathcal{A}^{\mathcal{P}}\interpret{\dot{\Gamma} 
			\vdash \dot{\sigma}} \ar[r, "{(\identity{}, a(\mathbf{op}) \comp \pi_2)}"] &
			\mathcal{A}^{\mathcal{P}}\interpret{\dot{\Gamma} \vdash \dot{\tau}}
		\end{tikzcd}
	\end{center}
	Here, the second morphism exists because we have a fully faithful functor $\refinetotal \to \category{P}^{\to}$.
	\begin{equation}
		\mathcal{A}^{\mathcal{P}}\interpret{\dot{\Gamma} \vdash \dot{\sigma}} \xrightarrow{(\identity{}, a(\mathbf{op}) \comp \pi_2)} \mathcal{A}^{\mathcal{P}}\interpret{\dot{\Gamma} \vdash \dot{\tau}}
		\qquad\mapsto\qquad
		\begin{tikzcd}
			\mathcal{A}^{\mathcal{P}}\interpret{\dot{\Gamma}, v : \dot{\sigma}} \ar[d, "\pi_{\mathcal{A}^{\mathcal{P}}\interpret{\dot{\Gamma} \vdash \dot{\sigma}}}"] \ar[r, "\identity{} \times a(\mathbf{op})"] & \mathcal{A}^{\mathcal{P}}\interpret{\dot{\Gamma}, v : \dot{\tau}} \ar[d, "\pi_{\mathcal{A}^{\mathcal{P}}\interpret{\dot{\Gamma} \vdash \dot{\tau}}}"] \\
			\mathcal{A}^{\mathcal{P}}\interpret{\dot{\Gamma}} \ar[r, "\identity{}"] & \mathcal{A}^{\mathcal{P}}\interpret{\dot{\Gamma}}
		\end{tikzcd}
	\end{equation}
\end{proof}

\subsection{Proofs for Specialised Rules for Basic Operators}

\begin{corollary}
	\textsc{R-BasicConst} is sound.
\end{corollary}
\begin{proof}
	By Lemma~\ref{lem:R-BasicOp-sound}, it suffices to prove
	\[ (\identity{} \times a(\mathbf{op})) : \mathcal{A}^{\mathcal{P}}\interpret{\dot{\Gamma}, v : 1} \dotTo \mathcal{A}^{\mathcal{P}}\interpret{\dot{\Gamma}, v : \{ v : \tau \mid v =_{\tau} \mathbf{op}() \}} \]
	This follows from the equation below.
	\begin{align}
		&(\identity{} \times a(\mathbf{op}))^{*} \mathcal{A}^{\mathcal{P}}\interpret{\underlying{\dot{\Gamma}}, v : \tau \vdash v =_{\tau} \mathbf{op}() } \\
		&= (\identity{} \times a(\mathbf{op}))^{*} \tupling{\pi_2}{a(\mathbf{op}) \comp {!}}^{*} \mathcal{P}({=}) \\
		&= \tupling{a(\mathbf{op}) \comp {!}}{a(\mathbf{op}) \comp {!}}^{*} \mathcal{P}({=}) \\
		&= \top \mathcal{A}\interpret{\underlying{\dot{\Gamma}}, v : 1}
	\end{align}
\end{proof}

\begin{lemma}\label{lem:tuple-context-formula}
	Let $\Gamma, x_1 : \sigma_1, x_2 : \sigma_2 \vdash \phi$ be a formula.
	\[ \associator^{*} \mathcal{A}^{\mathcal{P}}\interpret{\Gamma, (x_1, x_2) : \sigma_1 \times \sigma_2 \vdash \phi} = \mathcal{A}\interpret{\Gamma, x_1 : \sigma_1, x_2 : \sigma_2 \vdash \phi} \]
\end{lemma}
\begin{proof}
	By substitution and weakening.
	\begin{align}
		&\associator^{*} \mathcal{A}^{\mathcal{P}}\interpret{\Gamma, (x_1, x_2) : \sigma_1 \times \sigma_2 \vdash \phi} \\
		&= \associator^{*} \mathcal{A}^{\mathcal{P}}\interpret{\Gamma, x : \sigma_1 \times \sigma_2 \vdash \phi[\pi_1\ x/x_1, \pi_2\ x/x_2]} \\
		&= \associator^{*} \tupling{\identity{}}{\pi_1 \comp \pi_2}^{*} \tupling{\identity{}}{\pi_2 \comp \pi_2 \comp \pi_1}^{*} \mathcal{A}^{\mathcal{P}}\interpret{\Gamma, (x_1, x_2) : \sigma_1 \times \sigma_2, x_1 : \sigma_1, x_2 : \sigma_2 \vdash \phi} \\
		&= \associator^{*} \tupling{\identity{}}{\pi_1 \comp \pi_2}^{*} \tupling{\identity{}}{\pi_2 \comp \pi_2 \comp \pi_1}^{*} ((\pi_1 \times \identity{}) \times \identity{})^{*} \mathcal{A}^{\mathcal{P}}\interpret{\Gamma, x_1 : \sigma_1, x_2 : \sigma_2 \vdash \phi} \\
		&= \mathcal{A}^{\mathcal{P}}\interpret{\Gamma, x_1 : \sigma_1, x_2 : \sigma_2 \vdash \phi}
	\end{align}
\end{proof}

\begin{lemma}\label{lem:refine-tuple-interpret}
	Let $\sigma_1, \dots, \sigma_n \in B \cup \{ 1, \answertype \}$.
	\begin{align}
		&\mathcal{A}^{\mathcal{P}}\interpret{\dot{\Gamma} \vdash \{ (v_1, \dots, v_n) : \sigma_1 \times \dots \times \sigma_n \mid \phi \}} \\
		&= ((\mathcal{A}\interpret{\underlying{\dot{\Gamma}}},\quad \mathcal{A}\interpret{\sigma_1} \times ( \dots \times \mathcal{A}\interpret{\sigma_n})),\quad \mathcal{A}^{\mathcal{P}}\interpret{\dot{\Gamma}}, \\
		&\qquad \pi^{*} \mathcal{A}^{\mathcal{P}}\interpret{\dot{\Gamma}} \land \mathcal{A}^{\mathcal{P}}\interpret{\underlying{\dot{\Gamma}}, (v_1, (\dots, v_n)) : \sigma_1 \times (\dots \times \sigma_n) \vdash \phi})
	\end{align}
\end{lemma}
\begin{proof}
	By induction on $n$.
	The base case is trivial.
	If $n > 1$,
	\begin{align}
		&\mathcal{A}^{\mathcal{P}}\interpret{\dot{\Gamma} \vdash \{ (v_1, \dots, v_n) : \sigma_1 \times \dots \times \sigma_n \mid \phi \}} \\
		&= \coprod_{\mathcal{A}^{\mathcal{P}}\interpret{\dot{\Gamma} \vdash \{ v_1 : \sigma_1 \mid \top \}}} \mathcal{A}^{\mathcal{P}}\interpret{\dot{\Gamma}, v_1 : \sigma_1 \vdash \{ (v_2, \dots, v_n) : \sigma_2 \times \dots \times \sigma_n \mid \phi \}} \\
		&= \coprod_{\mathcal{A}^{\mathcal{P}}\interpret{\dot{\Gamma} \vdash \{ v_1 : \sigma_1 \mid \top \}}} ((\mathcal{A}\interpret{\underlying{\dot{\Gamma}}, v_1 : \sigma_1}, \mathcal{A}\interpret{\sigma_2} \times ( \dots \times \mathcal{A}\interpret{\sigma_n})), \mathcal{A}^{\mathcal{P}}\interpret{\dot{\Gamma}, v_1 : \sigma_1}, \\
		&\qquad\qquad \pi^{*} \mathcal{A}^{\mathcal{P}}\interpret{\dot{\Gamma}, v_1 : \sigma_1} \land \mathcal{A}^{\mathcal{P}}\interpret{\underlying{\dot{\Gamma}}, v_1 : \sigma_1, (v_2, (\dots, v_n)) : \sigma_2 \times (\dots \times \sigma_n) \vdash \phi}) \\
		&= ((\mathcal{A}\interpret{\underlying{\dot{\Gamma}}}, \mathcal{A}\interpret{b_1} \times ( \dots \times \mathcal{A}\interpret{b_n})), \mathcal{A}^{\mathcal{P}}\interpret{\dot{\Gamma}}, \\
		&\qquad (\associator^{-1})^{*}(\pi^{*} \mathcal{A}^{\mathcal{P}}\interpret{\dot{\Gamma}, v_1 : b_1} \land \mathcal{A}^{\mathcal{P}}\interpret{\underlying{\dot{\Gamma}}, v_1 : b_1, (v_2, (\dots, v_n)) : b_2 \times (\dots \times b_n) \vdash \phi}))
	\end{align}
	We have
	\[ (\associator^{-1})^{*}\pi^{*} \mathcal{A}^{\mathcal{P}}\interpret{\dot{\Gamma}, v_1 : b_1} = (\associator^{-1})^{*}\pi^{*} \pi^{*} \mathcal{A}^{\mathcal{P}}\interpret{\dot{\Gamma}} = \pi^{*} \mathcal{A}^{\mathcal{P}}\interpret{\dot{\Gamma}} \]
	We also have
	\begin{align}
		&(\associator^{-1})^{*} \mathcal{A}^{\mathcal{P}}\interpret{\underlying{\dot{\Gamma}}, v_1 : b_1, (v_2, (\dots, v_n)) : b_2 \times (\dots \times b_n) \vdash \phi} \\
		&= (\associator^{-1})^{*} \mathcal{A}^{\mathcal{P}}\interpret{\underlying{\dot{\Gamma}}, v_1 : b_1, v' : b_2 \times (\dots \times b_n) \vdash \phi[\dots]} \\
		&= (\associator^{-1})^{*} \associator^{*} \mathcal{A}^{\mathcal{P}}\interpret{\underlying{\dot{\Gamma}}, v : b_1 \times (\dots \times b_n) \vdash \phi[\dots][\pi_1\ v/v_1, \pi_2\ v/v']} \\
		&= \mathcal{A}^{\mathcal{P}}\interpret{\underlying{\dot{\Gamma}}, (v_1, (\dots, v_n)) : b_1 \times (\dots \times b_n) \vdash \phi}
	\end{align}
	by Lemma~\ref{lem:tuple-context-formula}.
\end{proof}

\begin{corollary}
	\textsc{R-BasicSimp} is sound.
\end{corollary}
\begin{proof}
	By Lemma~\ref{lem:R-BasicOp-sound}, it suffices to prove
	\[ (\identity{} \times a(\mathbf{op})) : \mathcal{A}^{\mathcal{P}}\interpret{\dot{\Gamma}, v : \{ (v_1, \dots, v_n) : \sigma_1 \times (\dots \times \sigma_n) \mid \phi[\mathbf{op}(v_1, \dots, v_n)/v] \}} \dotTo \mathcal{A}^{\mathcal{P}}\interpret{\dot{\Gamma}, v : \{ v : \tau \mid \phi \}} \]
	By Lemma~\ref{lem:refine-tuple-interpret},
	\begin{align}
		&\mathcal{A}^{\mathcal{P}}\interpret{\dot{\Gamma}, v : \{ (v_1, \dots, v_n) : \answertype^n \mid \phi[\mathbf{op}(v_1, \dots, v_n)/v] \}} \\
		&= \pi^{*} \mathcal{A}^{\mathcal{P}}\interpret{\dot{\Gamma}} \land \mathcal{A}^{\mathcal{P}}\interpret{\underlying{\dot{\Gamma}}, (v_1, (\dots, v_n)) : \sigma_1 \times (\dots \times \sigma_n) \vdash \phi[\mathbf{op}(v_1, \dots, v_n)/v]} \\
		&= \pi^{*} \mathcal{A}^{\mathcal{P}}\interpret{\dot{\Gamma}} \land \mathcal{A}^{\mathcal{P}}\interpret{\underlying{\dot{\Gamma}}, u : \sigma_1 \times (\dots \times \sigma_n) \vdash \phi[\mathbf{op}(u)/v]} \\
		&\mathcal{A}^{\mathcal{P}}\interpret{\dot{\Gamma}, v : \{ v : \answertype \mid \phi \}} \\
		&= \pi^{*} \mathcal{A}^{\mathcal{P}}\interpret{\dot{\Gamma}} \land \mathcal{A}^{\mathcal{P}}\interpret{\underlying{\dot{\Gamma}}, v : \tau \vdash \phi}
	\end{align}
	By Lemma~\ref{lem:weakening-formula},\ref{lem:subst-formula},
	\begin{align}
		&\mathcal{A}^{\mathcal{P}}\interpret{\underlying{\dot{\Gamma}}, u : \sigma_1 \times (\dots \times \sigma_n) \vdash \phi[\mathbf{op}(u)/v]} \\
		&= \tupling{\identity{}}{a(\mathbf{op}) \comp \pi_2}^{*} (\pi_1 \times \identity{})^{*} \mathcal{A}^{\mathcal{P}}\interpret{\underlying{\dot{\Gamma}}, v : \tau \vdash \phi} \\
		&= (\identity{} \times a(\mathbf{op}))^{*} \mathcal{A}^{\mathcal{P}}\interpret{\underlying{\dot{\Gamma}}, v : \tau \vdash \phi}
	\end{align}
\end{proof}

\begin{corollary}\label{cor:predicate-sound}
	\text{R-BasicBool} is sound
\end{corollary}
\begin{proof}
	By Lemma~\ref{lem:R-BasicOp-sound}, it suffices to prove
	\begin{align}
		(\identity{} \times a(\mathbf{op})) &: \mathcal{A}^{\mathcal{P}}\interpret{\dot{\Gamma}, v : \{ (v_1, \dots, v_n) : \sigma_1 \times \dots \times \sigma_n \mid \psi_t \land \phi_t[()/v] \lor \psi_f \land \phi_f[()/v] \}} \\
		&\qquad\dotTo \mathcal{A}^{\mathcal{P}}\interpret{\dot{\Gamma}, v : \{ v : 1 \mid \phi_t \} + \{ v : 1 \mid \phi_f \}}
	\end{align}
	By definition,
	\begin{align}
		&\mathcal{A}^{\mathcal{P}}\interpret{\dot{\Gamma}, v : \{ v : 1 \mid \phi_t \}} \\
		&= \pi^{*} \mathcal{A}^{\mathcal{P}}\interpret{\dot{\Gamma}} \land \mathcal{A}^{\mathcal{P}}\interpret{\underlying{\dot{\Gamma}},  v : 1 \vdash \phi_t } \\
		&= \pi^{*} \mathcal{A}^{\mathcal{P}}\interpret{\dot{\Gamma}} \land \pi^{*} \tupling{\identity{}}{{!}}^{*} \mathcal{A}^{\mathcal{P}}\interpret{\underlying{\dot{\Gamma}},  v : 1 \vdash \phi_t } \\
		&= \pi^{*} (\mathcal{A}^{\mathcal{P}}\interpret{\dot{\Gamma}} \land \mathcal{A}^{\mathcal{P}}\interpret{\underlying{\dot{\Gamma}} \vdash \phi_t[()/v] })
	\end{align}
	\begin{align}
		&\mathcal{A}^{\mathcal{P}}\interpret{\dot{\Gamma}, v : \{ v : 1 \mid \phi_t \} + \{ v : 1 \mid \phi_f \}} \\
		&= \{ \tupling{\gamma}{\iota_1} \mid \gamma \in \mathcal{A}^{\mathcal{P}}\interpret{\dot{\Gamma}} \cap \mathcal{A}^{\mathcal{P}}\interpret{\underlying{\dot{\Gamma}} \vdash \phi_t[()/v] } \}  \cup \{ \tupling{\gamma}{\iota_2} \mid \gamma \in \mathcal{A}^{\mathcal{P}}\interpret{\dot{\Gamma}} \cap \mathcal{A}^{\mathcal{P}}\interpret{\underlying{\dot{\Gamma}} \vdash \phi_f[()/v] } \} \\
		&= (\mathcal{A}^{\mathcal{P}}\interpret{\dot{\Gamma}} \land \mathcal{A}^{\mathcal{P}}\interpret{\underlying{\dot{\Gamma}} \vdash \phi_t[()/v] }) \dotTimes \{ \iota_1 \} \lor (\mathcal{A}^{\mathcal{P}}\interpret{\dot{\Gamma}} \land \mathcal{A}^{\mathcal{P}}\interpret{\underlying{\dot{\Gamma}} \vdash \phi_f[()/v] }) \dotTimes \{ \iota_2 \}
	\end{align}
	\begin{align}
		&(\identity{} \times a(\mathbf{op}))^{*} \mathcal{A}^{\mathcal{P}}\interpret{\dot{\Gamma}, v : \{ v : 1 \mid \phi_t \} + \{ v : 1 \mid \phi_f \}} \\
		&= (\mathcal{A}^{\mathcal{P}}\interpret{\dot{\Gamma}} \land \mathcal{A}^{\mathcal{P}}\interpret{\underlying{\dot{\Gamma}} \vdash \phi_t[()/v] }) \dotTimes (a(\mathbf{op}))^{*} \{ \iota_1 \} \\
		&\qquad \lor (\mathcal{A}^{\mathcal{P}}\interpret{\dot{\Gamma}} \land \mathcal{A}^{\mathcal{P}}\interpret{\underlying{\dot{\Gamma}} \vdash \phi_f[()/v] }) \dotTimes (a(\mathbf{op}))^{*}\{ \iota_2 \} \\
		&= \pi_1^{*} (\mathcal{A}^{\mathcal{P}}\interpret{\dot{\Gamma}} \land \mathcal{A}^{\mathcal{P}}\interpret{\underlying{\dot{\Gamma}} \vdash \phi_t[()/v] }) \land ({!} \times \identity{})^{*} \pi_2^{*} (a(\mathbf{op}))^{*} \{ \iota_1 \} \\
		&\qquad \lor \pi_1^{*} (\mathcal{A}^{\mathcal{P}}\interpret{\dot{\Gamma}} \land \mathcal{A}^{\mathcal{P}}\interpret{\underlying{\dot{\Gamma}} \vdash \phi_f[()/v] }) \land ({!} \times \identity{})^{*} \pi_2^{*} (a(\mathbf{op}))^{*}\{ \iota_2 \} \\
		&\ge \pi_1^{*} \mathcal{A}^{\mathcal{P}}\interpret{\dot{\Gamma}} \land (\pi_1^{*} \mathcal{A}^{\mathcal{P}}\interpret{\underlying{\dot{\Gamma}} \vdash \phi_t[()/v]} \land ({!} \times \identity{})^{*} \mathcal{A}^{\mathcal{P}}\interpret{(v_1, \dots, v_n) : \sigma_1 \times \dots \times \sigma_n \vdash \psi_t} \\
		&\qquad \lor \pi_1^{*} \mathcal{A}^{\mathcal{P}}\interpret{\underlying{\dot{\Gamma}} \vdash \phi_f[()/v]} \land ({!} \times \identity{})^{*} \mathcal{A}^{\mathcal{P}}\interpret{(v_1, \dots, v_n) : \sigma_1 \times \dots \times \sigma_n \vdash \psi_f}) \\
		&= \pi_1^{*} \mathcal{A}^{\mathcal{P}}\interpret{\dot{\Gamma}} \land \mathcal{A}^{\mathcal{P}}\interpret{\underlying{\dot{\Gamma}}, (v_1, \dots, v_n) : \sigma_1 \times \dots \times \sigma_n \vdash \phi_t[()/v] \land \psi_t \lor \phi_f[()/v] \land \psi_f} \\
		&= \mathcal{A}^{\mathcal{P}}\interpret{\dot{\Gamma}, v : \{ (v_1, \dots, v_n) : \sigma_1 \times \dots \times \sigma_n \mid \psi_t \land \phi_t[()/v] \lor \psi_f \land \phi_f[()/v] \}}
	\end{align}
\end{proof}

\begin{lemma}
	\textsc{R-Unif} is sound in the $\lambdaHFL$-model for expected costs and weakest pre-expectations.
\end{lemma}
\begin{proof}
	By Lemma~\ref{lem:R-BasicOp-sound}, it suffices to prove
	\begin{align}
		&(\identity{} \times a(\mathbf{unif})) \\
		&: \mathcal{A}^{\xi, \mathcal{P}}\interpret{\dot{\Gamma}, v : (x : \{ x : \mathbf{real} \mid 0 \le x \land x \le 1 \}) \to \{ v : \answertype \mid v \le N\ x \}} \\
		&\qquad\dotTo \mathcal{A}^{\xi, \mathcal{P}}\interpret{\dot{\Gamma}, v : \{ v : \answertype \mid v \le \mathbf{unif}(N) \}}
	\end{align}
	By definition,
	\begin{align}
		&\mathcal{A}^{\xi, \mathcal{P}}\interpret{\dot{\Gamma}, v : (x : \{ x : \mathbf{real} \mid 0 \le x \land x \le 1 \}) \to \{ v : \answertype \mid v \le N\ x \}} \\
		&= \{ \prod_{\mathcal{A}^{\xi, \mathcal{P}}\interpret{\dot{\Gamma} \vdash \{ x : \mathbf{real} \mid 0 \le x \land x \le 1 \}}} \mathcal{A}^{\xi, \mathcal{P}}\interpret{\dot{\Gamma}, x : \{ x : \mathbf{real} \mid 0 \le x \land x \le 1 \} \vdash \{ v : \answertype \mid v \le N\ x \}} \} \\
		&= \{ \tupling{\gamma}{f} : 1 \to \mathcal{A}^{\xi}\interpret{\underlying{\dot{\Gamma}}} \times (\exponential{\mathbb{R}}{\Omega}) \mid \gamma \in \mathcal{A}^{\xi, \mathcal{P}}\interpret{\dot{\Gamma}} \land \\
		&\qquad\forall x : 1 \to \mathbb{R}, 0 \le x \le 1 \implies \eval \comp \tupling{f}{x} \le \eval \comp \tupling{\mathcal{A}^{\xi}\interpret{N} \comp \gamma}{x} \}
	\end{align}
	Recall that $a(\mathbf{unif}) : \exponential{\mathbb{R}}{\answerobj} \to \answerobj$ is defined by $a(\mathbf{unif})(f) = \int_0^1 f(x)\, \mathrm{d} x$.
	By monotonicity of $a(\mathbf{unif})$, if
	\[ \tupling{\gamma}{f} \in \mathcal{A}^{\xi, \mathcal{P}}\interpret{\dot{\Gamma}, v : (x : \{ x : \mathbf{real} \mid 0 \le x \land x \le 1 \}) \to \{ v : \answertype \mid v \le N\ x \}} \]
	then
	\[ (\identity{} \times a(\mathbf{unif})) \comp \tupling{\gamma}{f} \in \{ \tupling{\gamma}{y} \mid \gamma \in \mathcal{A}^{\xi, \mathcal{P}}\interpret{\dot{\Gamma}} \land y \le a(\mathbf{unif}) \comp \mathcal{A}^{\xi}\interpret{N} \comp \gamma \} \]
\end{proof}

\subsection{Proofs for Rules for Admissibility}

\begin{lemma}\label{lem:atomic-chain-closed}
	If $\interpret{a}$ is chain-closed, then $\interpret{a(M)}$ is chain-closed at any variable.
\end{lemma}
\begin{proof}
	Obvious.
\end{proof}

\begin{lemma}\label{lem:le-chain-closed}
	$\interpret{\le_{\answertype}}$ is chain-closed.
\end{lemma}
\begin{proof}
	Easy.
\end{proof}

\begin{lemma}\label{lem:and-or-admissible}
	If $Q_1, Q_2 \in \category{P}_{\interpret{\Gamma}}$ are admissible at $x_i$, then so are $Q_1 \cap Q_2$ and $Q_1 \cup Q_2$.
\end{lemma}
\begin{proof}
	The only non-trivial part of the proof is the chain-closed-ness of $Q_1 \cup Q_2$.
	Assume $\gamma[x_i \mapsto f_k] \in Q_1 \cup Q_2$ holds for each $k$ where $f_{0} \le f_1 \le \dots$ is an $\omega$-chain in $\category{C}(1, \interpret{\sigma_i})$ and $\gamma : 1 \to \interpret{\Gamma}$.
	Then either $\{ k \mid \gamma[x_i \mapsto f_k] \in Q_1 \}$ or $\{ k \mid \gamma[x_i \mapsto f_k] \in Q_2 \}$ must be an infinite set.
	That is, we have an infinite sequence of indices $k_0 \le k_1 \le \dots$ such that either $\gamma[x_i \mapsto \sup_l f_{k_l}] \in Q_1$ or $\gamma[x_i \mapsto \sup_l f_{k_l}] \in Q_2$ holds.
	Since we have $\gamma[x_i \mapsto \sup_l f_{k_l}] = \gamma[x_i \mapsto \sup_k f_{k}]$, we have $\gamma[x_i \mapsto \sup_k f_{k}] \in Q_1 \cup Q_2$.
\end{proof}

\begin{lemma}\label{lem:implies-admissible}
	If $Q \in \category{P}_{\interpret{\Gamma}}$ is admissible at $x_i$ and $R \in \category{P}_{\interpret{\Gamma \setminus x_i}}$, then
	\begin{equation}
		\{ \gamma : 1 \to \interpret{\Gamma} \mid \gamma \in \mathrm{proj}_{\Gamma; x_i}^{*} R \implies \gamma \in Q \}
		\label{eq:implies}
	\end{equation}
	is admissible at $x_i$.
	Here $\mathrm{proj}_{\Gamma; v} : \interpret{\Gamma} \to \interpret{\Gamma \setminus v}$ is the canonical projection.
	\[ \mathrm{proj}_{\Gamma, v : \sigma; v} = \pi_1 \qquad \mathrm{proj}_{\Gamma, x : \sigma; v} = \mathrm{proj}_{\Gamma; v} \times \identity{} \]
\end{lemma}
\begin{proof}
	Note that for any $\gamma : 1 \to \interpret{\Gamma}$ and $f : 1 \to \interpret{\sigma_i}$, if $\gamma \in \mathrm{proj}^{*} R$, then $\gamma[x_i \mapsto f] \in \mathrm{proj}^{*} R$.
	For and $\gamma : 1 \to \interpret{\Gamma}$ and any $\omega$-chain $f_0 \le f_1 \le \dots$ in $\category{C}(1, \interpret{\sigma_i})$, if $\gamma[x_i \mapsto f_j]$ is in \eqref{eq:implies} for any $j$, then there are two cases.
	\begin{itemize}
		\item If $\gamma[x_i \mapsto f_j] \notin \mathrm{proj}^{*} R$ for some $j$, then $\gamma[x_i \mapsto \sup_j f_j] \notin \mathrm{proj}^{*} R$.
		\item If $\gamma[x_i \mapsto f_j] \in \mathrm{proj}^{*} R$ for any $j$, then $\gamma[x_i \mapsto f_j] \in Q$ for any $j$.
		Therefore, $\gamma[x_i \mapsto \sup_j f_j] \in Q$.
	\end{itemize}
	In any case, $\gamma[x_i \mapsto \sup_j f_j]$ is in \eqref{eq:implies}.
\end{proof}

\begin{lemma}\label{lem:adm-sound}
	If $\Gamma \vdash \mathrm{adm}(v, \phi)$, then $\mathcal{A}^{\mathcal{P}}\interpret{\Gamma \vdash \phi}$ is admissible at $v$.
\end{lemma}
\begin{proof}
	By induction on derivation of $\Gamma \vdash \mathrm{adm}(v, \phi)$.
	\begin{itemize}
		\item \textsc{Adm-Leq}: $\interpret{\answertype}$ has a bottom element because it is an EM $T$-algebra.
			For each $\gamma : 1 \to \interpret{\Gamma}$, we have $\gamma[v \mapsto \bot_{\interpret{\answertype}}] \in \interpret{v \le_{\answertype} M}$ because $\bot_{\interpret{\answertype}} \le \interpret{M} \comp \gamma[v \mapsto \bot_{\interpret{\answertype}}]$.
			$\interpret{v \le_{\answertype} M}$ is chain-closed at $v$ by Lemma~\ref{lem:atomic-chain-closed},\ref{lem:le-chain-closed}.
		\item \textsc{Adm-And}, \textsc{Adm-Or}: by Lemma~\ref{lem:and-or-admissible}.
		\item \textsc{Adm-Imp}: by Lemma~\ref{lem:implies-admissible}
	\end{itemize}
\end{proof}

\section{Benchmark Programs}\label{sec:benchmark}
\subsection{Weakest Pre-Expectation}

\begin{lstlisting}[title=lics16\_rec3]
let[@adm] rec rec3 k =
  let p = 0.5 in
  p *. k ()
    +. (1.0 -. p) *. rec3 (fun () -> rec3 (fun () -> rec3 (fun () -> k ())))

[@@@assert "typeof(rec3) <: (unit -> { r : prop | r = 1.0 })
  -> { ret : prop | 0.0 <= ret && ret <= (2.236068 - 1.0) / 2.0 }"]
\end{lstlisting}
\begin{lstlisting}[title=lics16\_rec3\_ghost]
let[@adm] rec rec3 (*ghost*)x k =
  let p = 0.5 in
  p *. k () +. (1.0 -. p) *. rec3 ((2.0 /. 3.0) *. (2.0 /. 3.0) *. x) (fun () ->
    rec3 ((2.0 /. 3.0) *. x) (fun () ->
      rec3 x (fun () -> k ())))

[@@@assert "typeof(rec3) <: { x : prop | x = 1.0 }
  -> (unit -> { r : prop | r = 1.0 })
  -> { ret : prop | 0.0 <= ret && ret <= 2.0 / 3.0 }"]
\end{lstlisting}
\begin{lstlisting}[title=lics16\_coins]
let p1 = 1.0 /. 2.0
let p2 = 1.0 /. 3.0
let coins k =
  p1 *. (p2 *. k 0 0 +. (1.0 -. p2) *. k 0 1)
    +. (1.0 -. p1) *. (p2 *. k 1 0 +. (1.0 -. p2) *. k 1 1)

[@@@assert "typeof(coins) <: ((x:int) -> (y:int)
    -> { r : prop | x = y && r = 1.0 || x <> y && r = 0.0 })
  -> { ret : prop | ret = 0.5 }"]
\end{lstlisting}

\subsection{Expected Cost Analysis}

\begin{lstlisting}[title=random\_walk]
let p = 2.0 /. 3.0
let[@adm] rec f x k =
  if x = 0
    then k ()
    else p *. (1.0 +. f (x - 1) k) +. (1.0 -. p) *. (1.0 +. f (x + 1) k)

[@@@assert "typeof(f) <: { x:int | x = 1 } -> (unit -> { r:prop | r = 0.0 })
  -> { ret : prop | 0.0 <= ret && ret <= 3.0 }"]
\end{lstlisting}
\begin{lstlisting}[title=random\_walk\_unif]
external unif : (float -> float) -> float = "unknown"
let[@adm] rec rw x k =
  if x >= 0.0 then
    unif (fun y ->
        let l = -2.0 in
        let r = 1.0 in
        1.0 +. rw (x +. (r -. l) *. y +. l) k)
  else k ()

[@@@assert "typeof(rw) <: { x:float | x = 1.0 } -> (unit -> { r:prop | r = 0.0 })
  -> { ret : prop | 0.0 <= ret && ret <= 6.0 }"]
\end{lstlisting}
\begin{lstlisting}[title=coin\_flip]
let p = 0.5
let[@adm] rec f x k = p *. k () +. (1.0 -. p) *. (1.0 +. f () k)

[@@@assert "typeof(f) <: unit -> (unit -> { r : prop | r = 0.0 })
  -> { ret : prop | 0.0 <= ret && ret <= 1.0 }"]
\end{lstlisting}
\begin{lstlisting}[title=coin\_flip\_unif]
external unif : (float -> float) -> float = "unknown"
let[@adm] rec f x k = unif (fun p -> p *. k () +. (1.0 -. p) *. (1.0 +. f () k))

[@@@assert "typeof(f) <: unit -> (unit -> { r : prop | r = 0.0 })
  -> { ret : prop | 0.0 <= ret && ret <= 1.0 }"]
\end{lstlisting}
\begin{lstlisting}[title=icfp21\_walk]
let[@adm] rec walk p f n k =
  p n (fun b -> if b then k n else 1.0 +. f n (fun m -> walk p f m k))

let geo = walk (fun _ k -> 0.5 *. k true +. 0.5 *. k false) (fun n k -> k (n + 1))

[@@@assert "typeof(geo) <: int -> (int -> { r : prop | r = 0.0 })
  -> { ret : prop | 0.0 <= ret && ret <= 1.0 }"]

let p = 2.0 /. 3.0
let randomWalk = walk (fun n k -> k (n <= 0))
  (fun n k -> p *. k (n - 1) +. (1.0 -. p) *. k (n + 1))

[@@@assert "typeof(randomWalk) <: (n:int) -> (int -> { r : prop | r = 0.0 })
  -> { ret : prop | 0.0 <= ret && ret <= 3.0 * float_of_int (abs n) }"]
\end{lstlisting}
\begin{lstlisting}[title=icfp21\_coupons]
let rec length = function [] -> 0 | _ :: xs -> 1 + length xs
let rec mem c = function [] -> false | x :: xs -> c = x || mem c xs
let rec is_subset cs os =
  match cs with [] -> true | c :: cs' -> mem c os && is_subset cs' os
let rec sum f = function [] -> 0.0 | x :: xs -> f x +. sum f xs
let[@adm] rec collect cs os k =
  if is_subset cs os then k os
  else
    let len_inv = 1.0 /. float_of_int (length cs) in
    1.0 +. sum (fun c -> len_inv *. collect cs (c :: os) k) cs
let coupons cs k = collect cs [] k

[@@@assert "typeof(coupons) <: { x : int list | x = 1 :: 2 :: [] }
  -> (int list -> { r : prop | r = 0.0 })
  -> { ret : prop | 0.0 <= ret && ret <= 2.0 * (1.0 + 1.0 / 2.0) }"]
\end{lstlisting}
\begin{lstlisting}[title=lics16\_fact]
let[@adm] rec fact x k =
  let p = 5.0 /. 6.0 in
  if x <= 0 then k 1
  else
    1.0 +. p *. fact (x - 1) (fun y -> k (y * x))
      +. (1.0 -. p) *. fact (x - 2) (fun y -> k (y * x))

[@@@assert "typeof(fact) <: (x:{ x:int | x >= 0 })
  -> (int -> { r : prop | r = 0.0 })
  -> { ret : prop | 0.0 <= ret
    && ret <= 1.0 + float_of_int x / (2.0 - 5.0 / 6.0) }"]
\end{lstlisting}

\subsection{Cost Moment Analysis}

\begin{lstlisting}[title=coin\_flip\_ord2]
let p = 0.5
let[@adm] rec f x k =
  let r11, r12 = k () in
  let r21, r22 = f () k in
  p *. r11 +. (1.0 -. p) *. (r21 +. 1.0),
  p *. r12 +. (1.0 -. p) *. (r22 +. 2.0 *. r21 +. 1.0)

[@@@assert "typeof(f) <: unit
  -> (unit -> { r : prop * prop | r = Tuple(0.0, 0.0) })
  -> { ret : prop * prop | 0.0 <= $proj(0, ret) && $proj(0, ret) <= 1.0
    && 0.0 <= $proj(1, ret) && $proj(1, ret) <= 3.0 }"]
\end{lstlisting}
\begin{lstlisting}[title=coin\_flip\_ord3]
let p = 0.5
let[@adm] rec f x k =
  let r11, r12, r13 = k () in
  let r21, r22, r23 = f () k in
  p *. r11 +. (1.0 -. p) *. (r21 +. 1.0),
  p *. r12 +. (1.0 -. p) *. (r22 +. 2.0 *. r21 +. 1.0),
  p *. r13 +. (1.0 -. p) *. (r23 +. 3.0 *. r22 +. 3.0 *. r21 +. 1.0)

[@@@assert "typeof(f) <: unit
  -> (unit -> { r : prop * prop * prop | r = Tuple(0.0, 0.0, 0.0) }) 
  -> { ret : prop * prop * prop | 0.0 <= $proj(0, ret) && $proj(0, ret) <= 1.0
    && 0.0 <= $proj(1, ret) && $proj(1, ret) <= 3.0
    && 0.0 <= $proj(2, ret) && $proj(2, ret) <= 13.0 }"]
\end{lstlisting}

\subsection{Conditional Weakest Pre-Expectation}

\begin{lstlisting}[title=toplas18\_ex4.4]
let[@admc] rec f m k =
  let (r11, r12) = f (m + 1) k in
  let (r21, r22) = k (m + 1) in
  (0.75 *. r11 +. 0.125 *. r21, 0.75 *. r12 +. 0.125 *. r22)

[@@@assert "typeof(f) <: { m : int | m = 0 }
  -> ((m:int) -> { r : prop * real | (m = 1 && $proj(0, r) = 1.0
    || m <> 1 && $proj(0, r) = 0.0) && $proj(1, r) = 1.0 })
  -> { ret : prop * real | 0.0 <= $proj(0, ret)
    && $proj(0, ret) <= 0.25 * $proj(1, ret)
    && 0.0 <= $proj(1, ret) && $proj(1, ret) <= 1.0 }"]

\end{lstlisting}
\begin{lstlisting}[title=two\_coin\_conditioning]
let[@admc] rec diverge () (k : unit -> float * float) : float * float
  = diverge () k

[@@@assert "typeof(diverge) <: unit -> (unit -> prop * real)
  -> { ret : prop * real | $proj(0, ret) = 0.0 && $proj(1, ret) = 1.0 }"]

let[@admc] rec f m k =
  let (r11, r12) = f () k in
  let (r21, r22) = k () in
  let (r31, r32) = diverge () k in
  0.25 *. r11 +. 0.25 *. r21 +. 0.25 *. r31,
  0.25 *. r12 +. 0.25 *. r22 +. 0.25 *. r32

[@@@assert "typeof(f) <: unit
  -> (unit -> { r : prop * real | $proj(0, r) = 1.0 && $proj(1, r) = 1.0 })
  -> { ret : prop * real | 0.0 <= $proj(0, ret)
    && $proj(0, ret) <= 0.5 * $proj(1, ret)
    && 0.0 <= $proj(1, ret) && $proj(1, ret) <= 1.0 }"]

\end{lstlisting}

}{}

\end{document}
\endinput